\documentclass[10pt]{article}
\usepackage{graphicx}
\usepackage[margin=2.5cm]{geometry}
\usepackage{authblk}
\usepackage{amsmath}
\usepackage{amssymb}
\usepackage{amsthm}
\usepackage{dsfont}
\usepackage{braket}
\usepackage{xcolor}
\usepackage{multicol}
\usepackage{cite}
\usepackage{physics}
\usepackage{caption}
\usepackage{subcaption}
\usepackage{circuitikz}
\usepackage{float}
\usepackage{hyperref}
\usepackage{cleveref}
\usepackage{comment}
\usepackage{tocloft}
\usepackage{relsize}
\usepackage{microtype}
\usepackage[all]{tcolorbox}

\tcbset{
    before upper={}
}

\usepackage{stmaryrd}
\usepackage[export]{adjustbox}

\newtheorem{rem}{Theorem}

\makeatletter
\DeclareRobustCommand{\shortarrow}{%
  \mathrel{\mathpalette\short@arrow\relax}%
}
\newcommand{\short@arrow}[2]{%
  \mkern2mu%
  \clipbox{{0.3\width} 0 0 0}{$\m@th#1\vphantom{+}{\shortrightarrow}$}%
  }
\makeatother

\title{\textbf{Asymptotic Compression of Interactive Quantum Communication using Type-Constrained de Finetti Reduction}}

\author{Louis Desruisseaux}
\author{Simon Ducharme\thanks{Corresponding author: simon.ducharme2@usherbrooke.ca}\hspace{0.25em}}
\author{Gurleen Padda}
\author{Dave Touchette}

\affil{Department of Computer Science, Université de Sherbrooke, Sherbrooke, QC, Canada}
\date{\today}
\begin{document}

\newtheorem{theorem}{Theorem}
\newtheorem{proposition}{Proposition}
\newtheorem{definition}{Definition}
\newtheorem{lemma}{Lemma}
\newtheorem{corollary}{Corollary}
\newtheorem{protocol}{Protocol}

\pagestyle{plain}
\maketitle

\vspace{-0.5em}

\begin{abstract}
    For many information processing tasks, de Finetti-style theorems can often simplify the analysis in worst-case input scenarios for which the task exhibits some permutation-invariance symmetry, as they can allow for a reduction from an analysis on worst-case inputs to that of i.i.d. inputs. If further information is available on the inputs, it might be advantageous to reflect this information in the de Finetti reduction. In our work, we focus on a form of such constraint, based on the type of the input. This allows us to obtain a conceptually simple proof of a new de Finetti reduction for classical probability distributions, derived from elementary properties from the method of types. We apply our constrained de Finetti reduction to the compression of quantum interactive communication protocols with classical inputs, and prove that the prior-free quantum information cost equals the worst-case input amortized quantum communication cost. 
\end{abstract}

\setlength{\cftaftertoctitleskip}{0.5pt}

\renewcommand{\cfttoctitlefont}{\Large\bfseries}

\renewcommand{\cftsecfont}{\small\bfseries}
\renewcommand{\cftsecpagefont}{\small\bfseries}
\renewcommand{\cftsubsecfont}{\small\itshape}
\renewcommand{\cftsubsecpagefont}{\small\bfseries}
\renewcommand{\cftsubsubsecfont}{\small\itshape}
\renewcommand{\cftsubsubsecpagefont}{\small\bfseries}

\setlength{\cftbeforesecskip}{0pt}
\setlength{\cftbeforesubsecskip}{0pt}

\tableofcontents

\section{Introduction}
The analysis of information processing tasks often becomes more tractable in the setting where the inputs are distributed independently and identically (i.i.d.), than in the worst-case input setting where there are no distributional or independence assumptions on the inputs. Using the density operator formalism from quantum information theory to conveniently encode probability distributions, an i.i.d. state is of the form $\rho^{\otimes n}$, corresponding to $n$ i.i.d. copies of a single-system state $\rho$. In worst-case input scenarios for which the information processing task admits some permutation-invariance symmetry, de Finetti-style theorems can often simplify the analysis, as it can allow for a reduction from an analysis on worst-case inputs to that of i.i.d. inputs. For these reductions, it is often sufficient to upper bound permutation-invariant inputs by de Finetti distributions, i.e. convex combinations of i.i.d. inputs, rather than to approximate the inputs by them directly. This upper bound is done with operator inequalities of the form $\rho \leq \sigma$, which indicates that $\sigma-\rho$ is positive semidefinite. Interestingly, the same de Finetti reduction can be applied to any permutation-invariant input, even if no additional information is given about the input. In the case where additional information is available however, it can be advantageous to use a constrained de Finetti distribution that somehow reflects this information.

In our work, we develop a type-constrained de Finetti reduction where we restrict ourselves to classical inputs, and we have further information on the type of the input. This allows us to obtain a de Finetti reduction with a conceptually simple proof, that is derived from elementary properties from the method of types. More precisely, we give an explicit form for the de Finetti state used in our reduction as a finite convex combination of i.i.d. states that are each related to a possible type that the input can have. Given approximate knowledge of the type of the input, only the i.i.d. states that correspond to normalized types that are close in total variation distance to the actual one of the input appear in the convex combination. We give the necessary background on de Finetti reductions and the theory of types in Section~\ref{sec:intro_definetti}, and present our novel de Finetti reduction in Section~\ref{sec:results_definetti}.\\
\indent Our type-constrained de Finetti reduction is then used as a key ingredient to adapt and extend a recently developed compression scheme for classical interactive communication protocols \cite{padda2025} to an analogous one for quantum interactive communication protocols with classical inputs. We obtain a proof that the prior-free amortized quantum communication cost of an interactive quantum communication protocol on classical inputs is equal to the prior-free quantum information cost, allowing us to settle a ten-year-old open question on whether it was possible to do so \cite{braverman2015, touchette2015quantum}. This can then significantly streamline the self-reducibility argument given in Ref.~\cite{braverman2015}, and we are confident that our type-constrained de Finetti reduction will find further applications in information and complexity theory. We describe the model used for interactive quantum communication protocols and the definitions for the communication and information costs in Section~\ref{sec:intro_interactive}, and summarize our findings in Section~\ref{sec:results_interactive}.
\subsection{Constrained de Finetti Reductions}
\subsubsection{Background and Related Work}
\label{sec:intro_definetti}
\paragraph{de Finetti reductions}
The seminal de Finetti reduction of Christandl, Konig and Renner~\cite{christandl2009} implements the previously discussed concepts, providing immediate applications in quantum cryptography. It is formulated as follows.
\begin{rem}[Implicit in Lemma~2 of Ref.~\cite{christandl2009}]
    Let $\mathcal{H}$ be a $d$-dimensional Hilbert space. For any permutation-invariant quantum state $\rho$ on $\mathcal{H}^{\otimes n}$,
    \begin{equation}
    \label{eqn:definetti}
        \rho \leq (n+1)^{d^{2}}\int \sigma^{\otimes n}_{\mathcal{H}}\mu(\sigma_{\mathcal{H}}),
    \end{equation}
    in which $\mu(\cdot)$ is a measure on the space of density operators on $\mathcal{H}$.
\end{rem}
In the case where $\rho$ is a classical probability distribution, it would state that there is a universal de Finetti distribution on the r.h.s. of (\ref{eqn:definetti}), which upper bounds any permutation-symmetric $n$-partite probability distribution over an alphabet of size $d$, up to a factor of $(n+1)^{d^2}$. Intuitively, the $(n+1)^{d^2}$ term can be seen as the cost of using the de Finetti distribution. For a constant alphabet size, this term is polynomial in $n$, which suffices for many applications including some demonstrated in quantum cryptography~\cite{christandl2009}, quantum tomography~\cite{christandl2012} and quantum compression~\cite{Berta_2011}.

Examples of constrained de Finetti reductions appear in Refs.~\cite{lancien2017,arnon2015,PRXQuantum.5.040315,lami2025doublycompositechernoffsteinlemma}. In Ref.~\cite{lancien2017}, for a finite alphabet $\mathbb{X}$ and any permutation-invariant classical state $\rho_{\mathbb{X}^{n}}$ on $\mathbb{X}^n$, they obtain the following:
\begin{rem}[Corollary 2.6 of \cite{lancien2017}]
    For every permutation-invariant classical state $\rho_{\mathbb{X}^n}$, there exists a measure~$d\mathcal{\sigma}_{\mathbb{X}}$ over classical states such that
    \begin{equation}
        \rho_{\mathbb{X}^{n}} \leq (n+1)^{3|\mathbb{X}|^{2}}\int_{\mathcal{\sigma}_{\mathbb{X}}}F(\rho_{\mathbb{X}^{n}}, \sigma_{\mathbb{X}}^{\otimes n})^{2}\sigma_{\mathbb{X}}^{\otimes n}d\sigma_{\mathbb{X}},
    \end{equation}
    in which $F(\cdot, \cdot)$ is the fidelity.
\end{rem}

Another example is given in Ref.~\cite{arnon2015}, in which the additional information on the inputs arises from the symmetries found in the observed correlations that result, for example, from the measurement statistics of quantum states, giving de Finetti reductions constrained to the resulting conditional probability distributions. Applying their findings to unconditional probability distributions, which is more of concern to us, gives the following:  
\begin{rem}[Theorem 4 of \cite{arnon2015}]
    For every permutation-invariant classical state $\rho_{\mathbb{X}^{n}}$ with symmetry $\mathcal{S}$, there exists a measure $d\mathcal{\sigma}^{\mathcal{S}}_{\mathbb{X}}$ over classical states such that
    \begin{equation}
        \rho_{\mathbb{X}^{n}} \leq (n+1)^{d}\int {\mathcal{\sigma}^{\mathcal{S}}_{\mathbb{X}}}^{\otimes n}d\mathcal{\sigma}^{\mathcal{S}}_{\mathbb{X}}.
    \end{equation}
\end{rem} 
Here, $d$ is the number of independent parameters (from the symmetry) in $\mathcal{\sigma}^{\mathcal{S}}_{\mathbb{X}}$, and $\rho_{\mathbb{X}^{n}}$ satisfies the symmetry component-wise.

\paragraph{Types and typicality}
We briefly introduce the notation related to types and typicality that we use throughout this work (see Section~\ref{sec:prelim} for formal definitions). If $\mathbb{X}$ is an alphabet, the type of a sequence~$x^n \in \mathbb{X}^n$ of length $n$ is the empirical distribution of the symbols in that sequence and is denoted $t_{x^n}$. For a type $t$, we define $T^t_{\mathbb{X}^n}$ as the set of all sequences for which their type is $t$. We define $\mathcal{P}_{\mathbb{X}^n}$ as the set of all possible types of sequences of length $n$ over the alphabet $\mathbb{X}$. A simple upper bound on $|\mathcal{P}_{\mathbb{X}^n}|$ is $(n + 1)^{|\mathbb{X}|}$.

For $\delta > 0$, we say that a type $t$ is $\delta$-typical to a distribution $p$ over $\mathbb{X}$ if it satisfies $\|t - p\|_1 \leq \delta$, in which~$\|\cdot\|_1$ is the total variation distance. The set of all types $\delta$-typical to $p$ is denoted $\mathcal{P}_{\mathbb{X}^n}^{p,\delta}$. Similarly, we say that a sequence $x^n \in \mathbb{X}^n$ is $\delta$-typical to a distribution $p$ over $\mathbb{X}$ if its type $t_{x^n}$ is $\delta$-typical to $p$. The set of all sequences $x^n$ that are $\delta$-typical to $p$ is denoted $T^{p,\delta}_{\mathbb{X}^n}$.

\subsubsection{Results}
\label{sec:results_definetti}

We now introduce our type-constrained de Finetti reduction. We define the $\delta$-typical de Finetti state with respect to a distribution $p$ over $\mathbb{X}$ as:
\begin{equation}
    \label{eq:de_finetti_state_definition}
    \tau_{\mathbb{X}^n}^{p,\delta} := \frac{1}{\left|\mathcal{P}_{\mathbb{X}^n}^{p,\delta}\right|} \sum_{t \in \mathcal{P}_{\mathbb{X}^n}^{p,\delta}} \sigma_t^{\otimes n},
\end{equation}
in which $\sigma_t$ is the density operator that encodes the distribution of a type $t$:
\begin{equation}
    \label{eq:sigma_t_definition}
    \sigma_t := \sum_{x \in \mathbb{X}} t(x) \ketbra{x}{x}.
\end{equation}
Given that the type of a sequence is a property that is invariant under permutations of the symbols in that sequence, our reduction naturally relate the states $\sigma_t^{\otimes n}$ to permutation-invariant density operators $\overline{\omega}_{t}$ defined as
\begin{equation}
    \label{eq:symmetric_basis_state_definition}
    \overline{\omega}_{t} := \frac{1}{\left|T^t_{\mathbb{X}^n}\right|} \sum_{x^n \in T^t_{\mathbb{X}^n}} \ketbra{x^n}{x^n}.
\end{equation}
Slightly abusing the terminology, we call these states the classical symmetric basis states, as any permutation-invariant density operators diagonal in the basis of the chosen alphabet $\mathbb{X}^n$ can be expressed as a convex combination of these $\overline{\omega}_{t}$ states. Similarly, we define for a distribution $p$ the $\delta$-typical classical symmetric subset as the set of states $\rho$ of the form
\begin{equation}
    \label{eq:convex_combination_of_symmetric_basis_states}
    \rho =  \sum_{t \in \mathcal{P}_{\mathbb{X}^n}^{p,\delta}} q(t) \overline{\omega}_{t},
\end{equation}
in which $q(t)$ is any distribution over the types $t \in \mathcal{P}_{\mathbb{X}^n}^{p,\delta}$.

With these definitions, we prove the following type-constrained de Finetti reduction.
\begin{rem}[Informal version of theorem \ref{thm:de_Finetti_Reduction_operator_inequality_form}]
\label{thm:restrictedClassicalQuantumDeFinettiReduction_intro}
For a distribution $p$, let $\rho$ be any state in its $\delta$-typical classical symmetric subset, then:
\begin{equation}
    \label{eq:restrictedClassicalQuantumDeFinettiReduction}
    \rho \leq \left|\mathcal{P}_{\mathbb{X}^n}\right|\left|\mathcal{P}_{\mathbb{X}^n}^{p,\delta}\right| \tau_{\mathbb{X}^n}^{p,\delta}.
\end{equation}
\end{rem}
To prove this result, we first upper-bound each symmetric basis state $\overline{\omega}_t$ in the sum defining $\rho$ by the corresponding $\sigma_t^{\otimes n}$ state up to a factor $|\mathcal{P}_{\mathbb{X}^n}|$. This follows since the weight on sequences of the same type is uniform in both distributions $\overline{\omega}_t$ and $\sigma_t^{\otimes n}$, and we can then relate the two distributions by only considering the total weight on each type. The type $t$ has the highest total weight in the distribution $\sigma_t^{\otimes n}$. Given that there are $|\mathcal{P}_{\mathbb{X}^n}|$ different types, multiplying the distribution $\sigma_t^{\otimes n}$ by $|\mathcal{P}_{\mathbb{X}^n}|$ ensures that the total weight on the type $t$ is at least one, since $\sigma_t^{\otimes n}$ is normalized. The bound on $\overline{\omega}_{t}$ then follows, since $t$ is the only type that appears in $\overline{\omega}_t$, and therefore has a total weight of one. We also obtain an additional factor of $|\mathcal{P}_{\mathbb{X}^n}^{p,\delta}|$ in \eqref{eq:restrictedClassicalQuantumDeFinettiReduction} to upper bound the distribution $q$ by the uniform one in~$\tau_{\mathbb{X}^n}^{p,\delta}$.

In the application of our results to quantum compression, which is covered in the next section, we consider maximizing a trace norm over classical input sequences that are $\delta$-typical to a distribution $p$ over $\mathbb{X}$. We formalize this by defining the following $\delta$-typical classical-quantum channel norm $\left\|\cdot\right\|_{\diamond,cl}^{p,\delta}$ for a linear map $\Delta$ acting on classical inputs on system $X^n$ and mapping them to quantum outputs:
\begin{equation}
    \label{eq:classical_delta_typical_diamond_norm}
    \|\Delta\|_{\diamond,cl}^{p,\delta}\!:=\!\max_{x^n \in T^{p,\delta}_{\mathbb{X}^n}} \left\|\left(\Delta\!\otimes\!\mathbf{id}\right)\left(\ketbra{x^n}{x^n}_{X^n}\!\otimes\!\ketbra{x^n}{x^n}_{R}\right)\right\|_1,
\end{equation}
in which $\mathbf{id}$ is the identity channel on a reference system $R$ that contains a copy of the classical input. With this, we obtain from our restricted de Finetti reduction the following corollary, similar to Theorem~1 of Ref.~\cite{christandl2009}.
\begin{rem}[Informal summary of Theorem~\ref{thm:restrictedClassicalQuantumDeFinettiReduction}] If for a linear map $\Delta$ and for all permutations $\pi$ there exists a CPTP map $\mathcal{K}_\pi$ such that $\Delta \circ \pi = \mathcal{K}_{\pi} \circ \Delta$ (permutation-covariance), in which $\circ$ denotes composition of maps, then
\begin{equation}
\label{eq:deFinettiReductionResults_informal}
\|\Delta\|_{\diamond,cl}^{p,\delta} \leq \left|\mathcal{P}_{\mathbb{X}^n}\right|\left|\mathcal{P}_{\mathbb{X}^n}^{p,\delta}\right|\left\|\left(\Delta \otimes \mathbf{id}\right)\left(\tau_{X^nR}^{p,\delta}\right)\right\|_1,
\end{equation}
in which $\tau_{X^nR}^{p,\delta}$ contains, on a reference system $R$, a copy of the classical inputs in the sum defining the de Finetti state $\tau_{X^n}^{p,\delta}$, and $\textbf{id}$ is the identity channel on that reference system.
\end{rem}
To prove this result, we first show that it is equivalent to do the classical-quantum channel norm's maximization over the $\delta$-typical classical symmetric subset instead. This follows from using the permutation covariance of the map $\Delta$ which induces an invariance of the trace norm under any permutation of the input state. We then use a similar post-selection technique to Lemma~2 of Ref.~\cite{christandl2009} to transform our de Finetti reduction as stated in \eqref{eq:restrictedClassicalQuantumDeFinettiReduction} into an equality by constructing a suitable trace-non-increasing map for the symmetric state of the form of \eqref{eq:convex_combination_of_symmetric_basis_states} that maximizes the trace norm, to map the de Finetti state $\tau_{X^nR}^{p,\delta}$ to that symmetric state up to a factor $\left|\mathcal{P}_{\mathbb{X}^n}\right||\mathcal{P}_{\mathbb{X}^n}^{p,\delta}|$. The trace-non-increasing property of that map induces the norm inequality in~\eqref{eq:deFinettiReductionResults_informal}. 
\subsection{Compression of Interactive Quantum Communication Protocols}
\subsubsection{Background and Related Work}
\label{sec:intro_interactive}
\paragraph{One-way communication} One-way communication is a communication model in which information flows only from a sender to a receiver, with no messages sent back from the receiver to the sender. In the quantum setting, results like Schumacher source compression \cite{schumacher1996quantum} and state-merging, also known as the fully quantum Slepian-Wolf theorem \cite{Horodecki_2005, Horodecki_2006,Abeyesinghe_2009}, fall under this model. Running the state merging protocol in reverse gives rise to what is known as state splitting, or the quantum reverse Shannon theorem \cite{bennett2002}, which serves as a building block for the interactive quantum communication model we later discuss.

Let $\ket{\psi_{RA}}$ denote a purification of the state $\rho_{A}$ on the Hilbert space $\mathcal{H}_{R} \otimes \mathcal{H}_{A}$, and let $\mathcal{N}_{A \rightarrow C}$ denote a quantum channel that is a CPTP map, with isometric extension $\mathcal{V}_{A \rightarrow A_{1}C}$ on the environment system $A_{1}$. In the state splitting protocol, Alice begins with the $A$ system, and her goal is to ``split her state" and transfer part of her system to Bob in such a way that at the end of the protocol, Bob holds the system $C$ corresponding to the output of the channel $\mathcal{V}_{A \rightarrow A_{1}C}$ acting on $A$. This can be interpreted as Alice communicating her system $A$ to Bob through a feedback channel that returns the system $A_{1}$ to her and sends the system $C$ to Bob. The quantum reverse Shannon theorem gives a bound on the communication complexity of simulating this using a noiseless qubit channel and pre-shared entanglement. We refer to this as the simulation, over a noiseless channel, of a noisy channel with feedback to the sender. In the asymptotic setting, Alice is given the i.i.d. input state $\rho^{\otimes n}_{A}$ and the resulting channel is the product channel $\mathcal{V}_{A \rightarrow A_1C}^{\otimes n}$. In this case, Alice can in some sense transfer the correlations between the $C$ and $A_{1}$ systems to $\frac{n}{2}I(A_{1};C)$ pairs of entangled qubits she shares with Bob, and with this she only has to communicate $\frac{n}{2}I(R;C)$ qubits \cite{Abeyesinghe_2009}, in which $I(\cdot ; \cdot)$ denotes the mutual information (See Ref.~\cite{wilde} for definitions of quantum information theoretic quantities). 

In the setting we just described, Alice and Bob have knowledge of what the input state $\rho_{A}$ is, and we refer to this as the distributional setting. This is in contrast with the prior-free setting, in which Alice and Bob no longer have knowledge of what the input state is. The quantum reverse Shannon theorem has been studied in the prior-free setting in Ref.~\cite{bennett_devetak_harrow_shor_winter_2014}, in which they develop a quantum theory of types to do so, and Ref.~\cite{Berta_2011} where they use one-shot information theory together with the de Finetti reduction of Ref.~\cite{christandl2009}. Although their results do not assume a restriction to classical input states, as we do in our case, their results cannot immediately be applied to the interactive setting that we are interested in studying, as it is not obvious how to extend their results to the case where Bob is also given an input system in the prior-free setting.

Before introducing the interactive communication model, it is useful to first look at quantum state redistribution. Moving back to the distributional setting, in addition to providing Alice with the system $A$, Bob is given some side information, namely a system $B$, in which the global purified input state is $\ket{\psi_{RAB}}$, and after applying the channel $\mathcal{V}_{A \rightarrow A_1C}$ on $A$, becomes $\ket{\psi_{RA_{1}CB}}$ \cite{devetak2008}. Alice's goal remains the same, however, she now may be able to use the correlations between her and Bob's input states to reduce the communication cost of the task. The communication complexity is given by the conditional mutual information, with $\frac{n}{2}I(C;R|B)$ qubits needing to be sent~\cite{devetak2008}.

\paragraph{Interactive communication}
In the interactive setting, Alice and Bob alternate taking on the roles of sender and receiver. The players begin with the same input state $\ket{\psi_{RA_{in}B_{in}}}$ as in quantum state redistribution, in which the $A_{in}$ system is owned by Alice and the $B_{in}$ system is owned by Bob. An interactive protocol is iterated by rounds of communication, and we assume that for odd rounds Alice is the sender and Bob is the receiver, and that for even rounds Bob is the sender and Alice is the receiver. Alice is thus the sender for the first round. In this first round, just as in quantum state redistribution, after having applied a quantum channel $\mathcal{V}_{A_{in} \rightarrow A_{1}C_{1}}$ on the system $A_{in}$ to obtain the state $\ket{\psi_{RA_{1}C_{1}B_{in}}}$, Alice's goal is to transfer the system $C_{1}$ to Bob. In the second round, Bob is the sender, and his goal is to do the same as Alice did in the first round, but with respect to a second quantum channel $\mathcal{V}_{B_{in}C_{1} \rightarrow B_{2}C_{2}}$ acting on the joint system $B_{in}C_{1}$ that he now holds after the first round. After this second round, the global purified state is $\ket{\psi_{RA_{1}C_{2}B_{2}}}$, in which the $A_{1}$ and $C_{2}$ systems are with Alice, and the $B_{2}$ system is with Bob. Subsequent rounds follow in the same way, in which in the $i$th round, the sender uses the $i$th quantum channel acting on the systems they hold, to then transfer part of the resulting state to the other player. The global state after $i$ rounds, assuming $i$ is odd, is $\ket{\psi_{RA_{i}B_{i-1}C_{i}}}$, in which the $A_{i}$ system is with Alice, and the $B_{i-1}$ and $C_{i}$ systems are with Bob. An illustration of an interactive quantum protocol can be found in Figure \ref{fig:interactive_protocol} of Section \ref{section:QuantumInteractiveCommunicationSimulation}.

We focus on the case where the systems $A_{in}, B_{in}$ are restricted to classical inputs from fixed alphabets $\mathbb{X}$ and $\mathbb{Y}$ respectively. We study the communication cost of simulating the protocol $\Pi$ in the asymptotic and prior-free setting. This means that the players have no knowledge of how their classical input sequences were generated, and that we cannot assume the existence of a prior distribution on their classical input sequences (the prior-free setting). In the asymptotic setting, we are interested in using noiseless quantum communication and pre-shared entanglement to simulate $n$ identical and independent executions of the protocol $\Pi$ on $n$ pairs of inputs $\Big((x_{1}, y_{1}), (x_{2}, y_{2}), ..., (x_{n}, y_{n})\Big)\in \left(\mathbb{X}\times\mathbb{Y}\right)^{n}$, which we denote as $(x^{n}, y^{n}) \in \mathbb{X}^{n}\times \mathbb{Y}^{n}$ for short. We use $\Pi^{\otimes n}(x^{n}, y^{n})$ to denote the output state of the protocol $\Pi$ applied identically and independently to each pair of inputs in $(x^{n}, y^{n})$.

One motivation of this setting is quantum communication complexity \cite{yao1993quantum, cleve1997}. In quantum communication complexity, one often compares whether or not quantum communication can help solve 2-party classical communication complexity problems more efficiently, so the problems naturally have classical inputs. Furthermore, a protocol that solves a communication problem should work for any input in the problem, which motivates the prior-free setting. The asymptotic (i.i.d) setting appears in certain communication problems that have an i.i.d structure, for example, when a problem can be decomposed into smaller subproblems identical to one another. The interactive setting allows Alice and Bob to communicate to each other back and forth repeatedly to solve more general communication complexity problems.

For $\varepsilon > 0$, we say that the protocol $\hat{\Pi}_{n, \varepsilon}$ simulates the protocol $\Pi^{\otimes n}$ with error $\varepsilon$ if for every$(x^{n}, y^{n})\in\mathbb{X}^{n}\times\mathbb{Y}^{n}$,
\begin{equation}
\label{eq:simulation_eq_intro}
    \|\Pi^{\otimes n}(x^{n}, y^{n}) - \hat{\Pi}_{n, \varepsilon}(x^{n}, y^{n})\|_{1} < \varepsilon,
\end{equation}
in which $\|\cdot\|_{1}$ denotes the trace distance.
Let $QCC(\Pi)$ denote the number of qubits that are exchanged in either direction during the execution of the protocol $\Pi$. The prior-free amortized quantum communication complexity of a protocol $\Pi$ is defined as
\begin{equation}
    AQCC(\Pi) = \lim_{\varepsilon \rightarrow 0}\lim_{n \rightarrow \infty} \min_{\hat{\Pi}_{n, \varepsilon}}\frac{QCC(\hat{\Pi}_{n, \varepsilon})}{n},
\end{equation}
in which the minimum is over simulation protocols $\hat{\Pi}_{n,\varepsilon}$ that satisfy \eqref{eq:simulation_eq_intro}.

In the distributional setting for interactive quantum communication protocols, in which the inputs $(x^{n}, y^{n})$ are $n$ i.i.d. inputs from a prior-distribution $\mu$, and the simulation only has to satisfy~\eqref{eq:simulation_eq_intro} on average over $\mu^{\otimes n}$, the equivalence between the quantum information cost and the expected amortized quantum communication cost for interactive quantum communication protocols was shown in Ref.~\cite{touchette2015quantum}. The quantum information cost of a $j$-round interactive quantum communication protocol $\Pi$ acting on classical inputs from a prior-distribution~$\mu$ is defined as 
\begin{equation}
    QIC(\Pi, \mu) = \sum_{i>0, odd}\frac{1}{2}I(C_{i};R|B_{i-1})_{\mu}+\sum_{i>0, even} \frac{1}{2}I(C_{i};R|A_{i-1})_{\mu},
\end{equation}
in which $B_{0} = B_{in}$.

The prior-free quantum information complexity for a classical information processing task was first defined in Ref.~\cite{braverman2015}. The prior-free quantum information cost of a protocol $\Pi$ acting on classical inputs is defined as
\begin{equation}
    QIC(\Pi) = \max_{\mu}QIC(\Pi, \mu).
\end{equation}
The connection between the prior-free quantum information and prior-free amortized quantum communication complexity was left open in Ref.~\cite{braverman2015}, and lacking such a result they had to resort to alternative methods to obtain their main result. As an application of our type-constrained de Finetti reduction, we demonstrate the equivalence between the prior-free quantum information and prior-free amortized quantum communication costs of interactive quantum communication protocols with classical inputs. In the process, we obtain optimal bounds on the prior-free communication complexity of quantum state redistribution with type-constrained classical inputs. We summarize our results in the following section.

\subsubsection{Results}
\label{sec:results_interactive}
As a warm-up, we begin with one-way communication and prove the following result for the prior-free quantum reverse Shannon theorem for classical-quantum channels with no side information at the receiver, which matches the results given in Ref.\cite{bennett_devetak_harrow_shor_winter_2014} when applied to classical-quantum channels.
\begin{rem}[Informal summary of Theorem \ref{thm:prior_free_cq_channel_simulation}]
\label{thm:qrst}
Let $\mathcal{N}_{X \to C}$ be a classical-quantum channel, and let $\mathcal{V}_{X \to A_1C}$ be an isometric extension of $\mathcal{N}$. Let $R$ be a system that purifies the input to the channel and let $n \in \mathbb{N}_{>0}$. Then, there exists a 2-message prior-free simulation protocol of $\mathcal{V}^{\otimes n}$ such that for every input sequence, the asymptotic quantum communication cost is $\frac{n}{2}I(C;R)_{\sigma_t}$ up to an additive term sublinear in $n$, in which $\sigma_t$ is the distribution of the type $t$ of the input sequence. The simulation uses entanglement and is accurate up to a trace distance error that vanishes exponentially fast in $n$.
\end{rem}
To prove this, we construct a protocol in which the sender first communicates the type of the input sequence to the receiver using only $O(\log(n))$ bits of classical communication. We can then use the corollary of our de Finetti reduction to reduce the analysis of the rest of the protocol on our $\delta$-typical de Finetti state $\tau_{X^nR}^{t,\delta}$ for some conveniently chosen $\delta > 0$, in which $t$ is the exact type of the input sequence. Then, for every type $t$ we define a state splitting protocol for the output of the isometry $\mathcal{V}^{\otimes n}$ applied to the corresponding de Finetti state. With this, we obtain precision performances that depend on worst-case information quantities over all types $\delta$-typical to the input distribution $t$. We can then apply the continuity of entropy to relate those worst-case information quantities to the information quantities of the actual type $t$ of the input sequence. This gives us, for every type $t$, a protocol for state splitting that uses entanglement and the aforementioned communication to simulate the isometry $\mathcal{V}^{\otimes n}$ on the de Finetti state $\tau_{X^nR}^{t,\delta}$ up to an error that vanishes exponentially fast in $n$ for a fixed $\delta$. A prior-free simulation protocol can then be constructed with the promised performances by communicating the type as a first step, and conditionally applying the state splitting protocol corresponding to that type as a second step.

We then consider the case where the receiver has side information. As an intermediate step, we follow the construction of \cite{Ye_Bai_Wang_2008}, and develop an efficient state redistribution protocol for our type-constrained de Finetti state, which we then apply to obtain a novel result for the prior-free quantum reverse Shannon theorem for classical inputs with side information at the receiver, which is summarized in the following theorem.
\begin{rem}[Informal summary of Theorem \ref{thm:prior_free_cq_channel_simulation_with_side_info}]
\label{thm:qrst_w_side_info}
Let $\mathcal{N}_{X \to C}$ be a classical-quantum channel, and let $\mathcal{V}_{X \to A_1C}$ be an isometric extension of $\mathcal{N}$. Let $R$ be a system that purifies the input of the channel and $Y$ be the system corresponding to the side information held by the receiver and let $n \in \mathbb{N}_{>0}$. Then, there exists a prior-free 2-message simulation protocol of $\mathcal{V}^{\otimes n}$ which, up to an additive term sublinear in $n$, has an asymptotic communication cost of $\frac{n}{2}I(C;R|Y)_{\sigma_t}$ qubits, in which $\sigma_t$ is the distribution of the joint type $t$ of the input sequences. The simulation consumes and generates entanglement, and is accurate up to an error in trace distance that, with high probability, vanishes exponentially fast in $n$.
\end{rem}
To prove this, we follow a similar strategy as in the previous theorem. We first use the joint type estimation protocol from Lemma 1 of Ref.~\cite{padda2025} so that both the sender and the receiver have access, with high probability, to a $\delta$-typical approximation $\tilde{t}$ of the joint type of the input sequences for some conveniently chosen $\delta > 0$. This uses only an amount of communication that is sublinear in $n$, which vanishes in the amortized asymptotic setting. We can then use the corollary to our de Finetti reduction on the joint system $X^nY^n$ to reduce the analysis of the rest of the protocol on the $\delta$-typical de Finetti state $\tau_{X^nY^nR}^{\tilde{t},\delta}$. Instead of state splitting, we construct for every possibly estimated type $\tilde{t}$ the state redistribution protocol of the corresponding de Finetti state to get communication bounds in terms of the conditional mutual information on worse-case information theoretic quantities over all joint types that are $\delta$-typical to the input distribution $\tilde{t}$. Just like the case with no side information, continuity of entropy gives us bounds in terms of the information quantities evaluated on the actual joint type $t$ of the input sequences. We again combine the joint type estimation protocol and the state redistribution protocol corresponding to the estimated joint type $\tilde{t}$ to obtain a prior-free simulation protocol with the promised performances.

Finally, we extend this to the interactive setting to obtain our main result, in which we demonstrate the equivalence between the prior-free quantum information and communication costs of interactive quantum communication protocols with classical inputs. This is summarized in the following theorem:
\begin{rem}[Informal summary of Theorem~\ref{thm:int}]
\label{thm:int_intro}
Let $\Pi$ be a quantum interactive communication protocol applied to classical inputs, consisting of $j$ rounds of communication. For any $n \in \mathbb{N}_{>0}$, there exists a prior-free~$(j+1)$-round simulation protocol $\Pi_{\text{sim}}$ that simulates $\Pi^{\otimes n}$ such that for all input sequences the error of the simulation, with high probability, vanishes exponentially fast in $n$. Up to an additive term sublinear in $n$, the quantum communication cost of the simulation protocol is given by $n\sum_{i>0, \text{odd}}I(C_i;R|B_{i-1})_{\sigma_{t}} + n\sum_{i>0,\text{even}}I(C_{i};R|A_{i-1})_{\sigma_{t}}$, in which $\sigma_{t}$ is the distribution of the joint type $t$ of the input sequences. Here,~$C_j$ is the system communicated in round $j$, $A_j$ and $B_j$ are the local systems of Alice and Bob respectively at round~$j$, and $R$ is the purification system of the classical inputs. This proves that~$QIC(\Pi) = AQCC(\Pi)$.
\end{rem}
To prove this theorem, we again apply the joint type estimation protocol from Lemma 1 of Ref.~\cite{padda2025} and the corollary to our de Finetti reduction on the whole protocol $\Pi^{\otimes n}$ to reduce the prior-free setting to finding simulation protocols for each estimated joint type $\tilde{t}$ on the corresponding de Finetti state $\tau_{X^nY^nR}^{\tilde{t},\delta}$ for some conveniently chosen $\delta > 0$. For each of these estimated joint type $\tilde{t}$, we construct a simulation protocol by replacing each communication round of the original protocol $\Pi^{\otimes n}$ by a state redistribution protocol of our $\delta$-typical de Finetti state at that round. This gives us a communication cost of a sum of worst-case conditional mutual informations over all joint types that are $\delta$-typical to the estimated joint type $\tilde{t}$, in which the sum is over each round. We then use continuity of entropy to relate these worst-case conditional mutual informations to the conditional mutual informations evaluated on the distribution of the actual joint type $t$ of the input sequences at round $i$. Combining the joint type estimation protocol with the simulation protocol of the estimated joint type $\tilde{t}$ gives us a prior-free simulation protocol with the promised performances.

\section{Preliminaries}
\label{sec:prelim}
The Hilbert spaces are assumed to be finite-dimensional. The set of linear operators on a Hilbert space $\mathcal{H}$ is denoted by $\mathcal{L}(\mathcal{H})$, while the set of density operators is expressed as $\mathcal{D}(\mathcal{H})$. We write $\rho_A$ for a density operator on $\mathcal{H}_A$ and $\rho_{AB}$ for a joint density operator on $\mathcal{H}_A \otimes \mathcal{H}_B$. For $\rho_{AB} \in \mathcal{D}(\mathcal{H}_A \otimes \mathcal{H}_B)$, the reduced density operator on $\mathcal{H}_A$ is denoted by $\rho_A = \text{Tr}_B(\rho_{AB})$, in which $\text{Tr}_{B}(\cdot)$ is the partial trace function on system $B$. For a distribution $p$, we denote $p^{\otimes n}$ to be the product distribution over $n$ independently and identically distributed (i.i.d) copies. We use $\exp{\cdot}$ to denote the exponential function in base $2$.

Quantum channels are defined as completely positive trace-preserving (CPTP) maps from $\mathcal{L}(\mathcal{H}_A)$ to $\mathcal{L}(\mathcal{H}_B)$, and denoted by $\mathcal{N}_{A \to B}$. The identity isometry from $\mathcal{L}(\mathcal{H}_A)$ to $\mathcal{L}(\mathcal{H}_B)$ is denoted by $\mathcal{I}_{A \to B}$. Isometries are often applied by conjugation, such that if $V$ is an isometry from $\mathcal{L}(\mathcal{H}_A)$ to $\mathcal{L}(\mathcal{H}_B)$, then for $\rho_A \in \mathcal{L}(\mathcal{H}_A)$, we have $V(\rho_A) = V \rho_A V^\dagger$. For a set of symbols $\mathbb{X}$, we say $\mathbb{X}$ imposes the canonical basis on a Hilbert space $\mathcal{H}_{X}$ if we fix the canonical basis of $\mathcal{H}_{X}$ to be the set of orthonormal vectors $\{\ket{x} : x \in \mathbb{X}\}$. For a system $X$ with Hilbert space $\mathcal{H}_X$, we often use the notation $X$ to refer to the set of symbols imposing the canonical basis of $\mathcal{H}_X$. For a Hilbert space $\mathcal{H}_X$ on system $X$, we denote by $|X|$ the dimension of the Hilbert space (equivalently, the size of the alphabet of its canonical basis).

In the rest of this section, we discuss the definitions and notations used in this paper. For more details, we refer the reader to Ref.~\cite{wilde} for quantum information theory and the theory of types, and Ref.~\cite{yao1993quantum}\cite{cleve1997} for notions on interactive protocols.

\subsection{Quantum Information Theory}
We introduce the trace norm of an operator and the diamond norm of a linear map, which is useful to quantify the distance between quantum states and quantum channels, respectively. The trace norm of an operator~$X~\in~\mathcal{L}(\mathcal{H})$ is defined as
\begin{equation}
\label{eq:trace-norm}
\left\|X\right\|_1 := \text{Tr}\left(\sqrt{X^\dagger X}\right),
\end{equation}

in which $\text{Tr}(\cdot)$ is the trace function. We define the diamond norm of a linear map $\mathcal{M}_{A \to B} : \mathcal{L}(\mathcal{H}_A) \to \mathcal{L}(\mathcal{H}_B)$ as
\begin{equation}
\label{eq:diamond-norm}
\left\|\mathcal{M}_{A \to B}\right\|_{\diamond} := \max_{\ket{\psi}_{RA}}\left\|\left(\mathcal{M}_{A \to B} \otimes \textbf{id}_R\right)(\ketbra{\psi}{\psi}_{AR})\right\|_1,
\end{equation}
in which the maximization is over all pure states $\ket{\psi}_{RA} \in \mathcal{H}_R \otimes \mathcal{H}_A$, with $\mathcal{H}_R$ being isomorphic to $\mathcal{H}_A$. The Von Neumann entropy of a density operator $\rho_A \in \mathcal{D}(\mathcal{H}_A)$ is defined as
\begin{equation}
\label{eq:von-neumann-entropy}
H(A)_\rho := -\text{Tr}(\rho_A \log \rho_A).
\end{equation}
The Von Neumann entropy is the generalization of the Shannon entropy to the quantum setting. The Shannon entropy of a distribution $p$ is defined as
\begin{equation}
\label{eq:shannon-entropy}
H_S(p) := -\sum_{x} p(x) \log p(x).
\end{equation}
The Von Neumann entropy is equivalent to the Shannon entropy when the latter is applied to the distribution induced by the eigenvalues of $\rho_A$. For that reason, we always use the same notation $H(\cdot)$ for both the Von Neumann and Shannon entropies, and the meaning will be clear from the context.

For a bipartite density operator $\rho_{AB} \in \mathcal{D}(\mathcal{H}_A \otimes \mathcal{H}_B)$, the conditional entropy of the system $A$ given system $B$ is defined as
\begin{equation}
\label{eq:conditional-entropy}
H(A|B)_\rho := H(AB)_\rho - H(B)_\rho,
\end{equation}
in which $H(AB)_\rho$ is the Von Neumann entropy of $\rho_{AB}$. The mutual information between systems $A$ and $B$ is defined as
\begin{equation}
\label{eq:mutual-information}
I(A;B)_\rho := H(A)_\rho - H(A|B)_\rho.
\end{equation}
The conditional mutual information between systems $A$ and $B$ given system $C$ is defined as
\begin{equation}
\label{eq:conditional-mutual-information}
I(A;B|C)_\rho := H(A|C)_\rho - H(A|BC)_\rho.
\end{equation}

In this work, we often use the continuity of entropy, also known as the Fannes-Audenaert inequality, which states that for $\rho, \sigma \in \mathcal{D}(\mathcal{H}_A)$ such that $\left\|\rho - \sigma \right\|_1 \leq \delta$, we have
\begin{equation}
\label{eq:fannes-audenaert-inequality}
\left|H(\rho) - H(\sigma)\right| \leq \frac{\delta}{2} \log (|A|- 1) + h_2\left(\frac{\delta}{2}\right),
\end{equation}
in which $h_2(\cdot)$ is the binary entropy function defined as:
\begin{equation}
\label{eq:binary-entropy-function}
h_2(x) := -x \log x - (1-x) \log (1-x), \quad x \in [0,1].
\end{equation}
We denote by $\omega(\delta,d)$ the RHS of \eqref{eq:fannes-audenaert-inequality} for convenience:
\begin{equation}
\label{eq:omega-delta-d}
\omega(\delta,d) := \frac{\delta}{2} \log (d- 1) + h_2\left(\frac{\delta}{2}\right).
\end{equation}

In this work, the quantum channels we work with are mainly classical-quantum channels \cite{wilde}. These channels are such that the input is measured in some chosen basis, for which a quantum state is then generated based on the measured result. More formally, if $\mathcal{N}_{X \to A}$ is a classical-quantum channel from a system $X$ to a system $A$, then for $\rho_{X} \in \mathcal{D}(\mathcal{H}_{X})$, we have
\begin{equation}
\label{eq:classical-quantum-channel}
\mathcal{N}_{X \to \mathcal{A}}(\rho_{X}) = \sum_{x} \bra{x} \rho_{X} \ket{x} \sigma_{x,A},
\end{equation}
in which $\sigma_{x,A} \in \mathcal{D}(\mathcal{H}_A)$ is a quantum state associated with the measurement outcome $\ket{x}$.

Equivalently, for any quantum input to a classical-quantum channel, there is a diagonal state in the chosen canonical basis that outputs the same quantum state. We often refer to classical states as being states that are diagonal in some predetermined canonical basis. This gives the motivation for emphasizing on the idea of a set of symbols imposing the canonical basis of a Hilbert space, as it defines what states can be viewed as a classical probability distribution (classical states). Therefore, a classical-quantum channel can be viewed as mapping a classical probability distribution to a distribution of quantum states. Despite $\mathcal{H}_{X}$ being a general (finite) Hilbert space, we often, if not made explicit, reserve the use of $X$ or $Y$ for systems on which we only define classical states, while other letters are reserved for general quantum systems.

The isometric extension of a quantum channel $\mathcal{N}_{A \to B}$ is defined as an isometry $V_{A \to BE}$ such that for all~$\rho_A~\in~\mathcal{D}(\mathcal{H}_A)$,
\begin{equation}
\label{eq:isometric-extension}
\mathcal{N}_{A \to B}(\rho_A) = \text{Tr}_E \left(V_{A \to BE}(\rho_A)\right).
\end{equation}
The system $E$ can be viewed as the environment that purifies the output of the quantum channel $\mathcal{N}_{A \to B}$. In particular, for a classical-quantum channel $\mathcal{N}_{X \to A}$ with isometric extension $V_{X \to AE}$, we can deduce from \eqref{eq:classical-quantum-channel} that the environment system $E$ can be decomposed into two systems $E_X$ and $E_Q$, in which $E_X$ is a classical system that stores the measurement result of the input in the chosen basis, and $E_Q$ is a quantum system that contains the purification of the output quantum state. Therefore, $V_{X \to AE}$ has Kraus operators $\left\{V_x\right\}_{x}$ such that for all $x$,
\begin{equation}
\label{eq:classical-quantum-channel-isometric-extension-kraus-operators}
V_x := \bra{x}_{A} \otimes \ket{x}_{E_X} \otimes \ket{\sigma_{x}}_{E_Q},
\end{equation}
in which $\ket{\sigma_{x}}_{E_Q}$ is a purification of $\sigma_{x,A}$, the output quantum state associated with the measurement result~$\ket{x}$.

Using the notation $X$ to denote the set of symbols imposing the canonical basis of $\mathcal{H}_X$, we can define the notion of a copy of a classical state. Let $\rho_{X} \in \mathcal{D}(\mathcal{H}_{X})$ be a classical state, that is,
\begin{equation}
\label{eq:classical-state}
\rho_{X} = \sum_{x \in X} p(x) \ketbra{x}{x}_X,
\end{equation}
in which $p$ is a probability distribution over the set of symbols $X$. For a system $R$ with Hilbert space $\mathcal{H}_R$ isomorphic to $\mathcal{H}_X$, a state $\rho_{XR} \in \mathcal{D}(\mathcal{H}_{X} \otimes \mathcal{H}_{R})$ is then said to be a \textit{copy} of $\rho_{X}$ if
\begin{equation}
\label{eq:copy-of-classical-state}
\rho_{XR} = \sum_{x \in X} p(x) \ketbra{x}{x}_{X} \otimes \ketbra{x}{x}_{R}.
\end{equation}
The intuition behind this definition is that the canonical basis states $\ketbra{x}{x}_X$ from which are defined the classical states can be copied and a procedure to copy them from system $X$ to system $R$, despite the classical uncertainty, transforms by linearity a classical state in the form of \eqref{eq:classical-state} to a classical state in the form of equation \eqref{eq:copy-of-classical-state}.

Finally, we use Uhlmann's theorem frequently in our application of our de Finetti reduction. Uhlmann's theorem states that for two density operators $\rho_{AR}, \sigma_{AR'} \in \mathcal{D}(\mathcal{H}_A \otimes \mathcal{H}_R)$, in which $\dim(\mathcal{H}_R) \leq \dim(\mathcal{H}_{R'})$, and
\begin{equation}
\left\|\rho_{A} - \sigma_{A} \right\|_1 \leq \varepsilon,
\end{equation}
there exists an isometry $V_{R \to R'}$, such that 
\begin{equation}
\label{eq:uhlmann-theorem}
\left\|(\mathcal{I}_{A \to A} \otimes V_{R \to R'})(\rho_{AR}) - \sigma_{AR'}\right\|_1 \leq 2 \sqrt{\varepsilon},
\end{equation}
in which the isometry $\mathcal{I}_{A \to A} \otimes V_{R \to R'}$ is applied by conjugation. Because of this result, for any state $\psi_A$, any of its purification have their trace norm in the form of \eqref{eq:diamond-norm} equal. Therefore, the diamond norm in \eqref{eq:diamond-norm} can be maximized over every states with only one of their purification each. A similar idea occurs when the map $\mathcal{M}_{A \to B}$ in \eqref{eq:diamond-norm} is a classical-quantum channel. When this is the case, it turns out it is sufficient to take a copy of the classical input state to maximize the diamond norm in \eqref{eq:diamond-norm}. Furthermore, for every copy of the classical input states the value of the trace norm in \eqref{eq:diamond-norm} is the same, so we can choose any of them to compute the diamond norm.

\subsection{Theory of Types}
The notion of types (using the method of types) is central in achieving optimal performances for our various results, as is the case for many communication protocols \cite{wilde,bennett_devetak_harrow_shor_winter_2014}. Let $\mathbb{X}$ label a set of symbols and $x^n \in \mathbb{X}^n$ a sequence of $n$ symbols from $\mathbb{X}$, we denote $N(x|x^n)$ as the number of occurrences of the symbol $x$ in $x^n$. The \textit{type} $t_{x^n}$ of a sequence $x^n \in \mathbb{X}^n$ is defined as the empirical distribution of the symbols in $x^n$:
\begin{equation}
\label{eq:type}
t_{x^n}(x) := \frac{N(x|x^n)}{n}, \quad x \in \mathbb{X}.
\end{equation}
We may also refer to a sequence $x^n$ as being of type $t$ to imply that $t_{x^n} = t$. The set of all possible types for sequences of length $n$ composed of symbols from $\mathbb{X}$ is denoted by $\mathcal{P}_{\mathbb{X}^n}$. A simple bound on $|\mathcal{P}_{\mathbb{X}^n}|$ is given by
\begin{equation}
\label{eq:number-of-types}
|\mathcal{P}_{\mathbb{X}^n}| \leq (n+1)^{|\mathbb{X}|},
\end{equation}
which is polynomial in $n$ for constant $|\mathbb{X}|$. For a type $t \in \mathcal{P}_{\mathbb{X}^n}$, we define its \textit{type class} $T_{\mathbb{X}^n}^{t}$ as
\begin{equation}
\label{eq:type-class}
T_{\mathbb{X}^n}^{t} := \left\{x^n \in \mathbb{X}^n : t_{x^n} = t \right\},
\end{equation}
which is the set of all sequences of length~$n$ that are of type~$t$. A central result from the method of types is that the size of a type class is bounded by an exponential in~$nH(t)$, in which~$H(t)$ is the Shannon entropy of the type~$t$:
\begin{equation}
\label{eq:type-class-size}
\frac{2^{n H(t)}}{|\mathcal{P}_{\mathbb{X}^n}|} \leq |T_{\mathbb{X}^n}^{t}| \leq 2^{n H(t)}.
\end{equation}
For a distribution $p$ on alphabet $\mathbb{X}$ and a set $A \subseteq \mathbb{X}$, we denote $p(A) := \sum_{x \in A} p(x)$ for convenience. The probability that a sequence $x^n \in T_{\mathbb{X}^n}^{t}$ of type $t$ is generated in an i.i.d. manner according to the distribution $p$ is upper bounded by
\begin{equation}
\label{eq:type-probability-upper-bound}
p^{\otimes n}(x^n) \leq 2^{-nH(t)}.
\end{equation}
For a sequence generated in an i.i.d manner according to the distribution of a type $t$, the probability that it is a sequence of type $t$ is always greater than the probability that it is of another fixed type $t'$, that is:
\begin{equation}
\label{eq:type-probability-comparison}
t^{\otimes n}(T_{\mathbb{X}^n}^{t}) \geq t^{\otimes n}(T_{\mathbb{X}^n}^{t'}), \quad t, t' \in \mathcal{P}_{\mathbb{X}^n}.
\end{equation}

Let~$\delta > 0$, we say that a sequence~$x^n \in \mathbb{X}^n$ is~$\delta$-typical to a distribution~$p$ on~$\mathbb{X}$ if
\begin{equation}
\label{eq:delta-typical-sequence}
\left\|t_{x^n} - p \right\|_1 \leq \delta,
\end{equation}
in which $\left\| \cdot \right\|_1$ is the total variation distance, defined as
\begin{equation}
\label{eq:total-variation-distance}
\left\|t - p \right\|_1 := \sum_{x \in \mathbb{X}} \left|t(x) - p(x)\right|.
\end{equation}
The set of all types that are $\delta$-typical to $p$ is denoted by
\begin{equation}
\label{eq:typical-types}
\mathcal{P}_{\mathbb{X}^n}^{p,\delta} := \left\{t \in \mathcal{P}_{\mathbb{X}^n} : \left\|t - p \right\|_1 \leq \delta \right\}.
\end{equation}
The set of all $\delta$-typical sequences is denoted~$T_{\mathbb{X}^n}^{p,\delta}$ and is thus the union of the (disjoint) type classes of all $\delta$-typical types:
\begin{equation}
\label{eq:typical-sequences-set}
T_{\mathbb{X}^n}^{p,\delta} := \bigcup_{t \in \mathcal{P}_{\mathbb{X}^n}^{p,\delta}} T_{\mathbb{X}^n}^{t}.
\end{equation}

Let $x^n \in \mathbb{X}^n$ be a sequence generated in an i.i.d. manner from a distribution $p$, the probability that it turns out to be a sequence $\delta$-typical to $p$ is exponentially close to $1$:
\begin{equation}
\label{eq:typical-sequences-probability}
p^{\otimes n}(T_{\mathbb{X}^n}^{p,\delta}) \geq 1 - \left|\mathcal{P}_{\mathbb{X}^n}\right|2^{-n \frac{\delta^2}{2\ln(2)}}.
\end{equation}
This last result guaranties that considering only $\delta$-typical sequences turns out to be a good approximation of all the sequences generated by $p^{\otimes n}$ when $n$ is large enough. There are other interesting properties from the method of types which are at the heart of most of our results. Not only does the approximation through $\delta$-typical sequences of a distribution $p$ works well for large $n$, but the resulting number of sequences that must be considered is also significantly reduced and depends on information-theoretic quantities of $p$. To compute the size of the set of $\delta$-typical sequences, we can use its definition from \eqref{eq:typical-sequences-set} to express its size as the sum of the sizes of the (disjoint) type classes of all $\delta$-typical types. Using the bounds on the size of the type class~\eqref{eq:type-class-size}, the continuity of entropy \eqref{eq:fannes-audenaert-inequality} and the definition of the set of $\delta$-typical sequences \eqref{eq:delta-typical-sequence}, we can show that the size of this set of $\delta$-typical sequences is upper-bounded as 
\begin{equation}
\label{eq:typical-sequences-size}
\left|T_{\mathbb{X}^n}^{p,\delta}\right| \leq \exp{n \left(H(p) + \tilde{\omega}_{1,1}(\delta,|\mathbb{X}|)\right)},
\end{equation}
in which~$H(p)$ is the Shannon entropy of the distribution $p$ and where we have defined
\begin{equation}
    \label{eq:tilde-omega}
    \tilde{\omega}_{k,m}(\delta,d) := k\cdot\omega(\delta,d) + m\cdot\frac{d\log(n+1)}{n},
\end{equation}
with~$\omega(\delta,d)$ defined in \eqref{eq:omega-delta-d}. The additional term to $\omega(\delta,d)$ in $\tilde{\omega}_{k,m}(\delta,d)$ comes from the polynomial bound on the number of types from \eqref{eq:number-of-types}. This additional term, together with $\omega(\delta,d)$ from the continuity of entropy, appears often so we generalize $\tilde{\omega}_{k,m}(\delta,d)$ for any $k,m \in \mathbb{N}$ for convenience. Furthermore, we keep for the dependence on $n$ implicit in the definition of $\tilde{\omega}_{k,m}(\delta,d)$ as we never specify the value of $n$. Note that for fixed $k$,$m$ and $d$, $\tilde{\omega}_{k,m}(\delta,d)$ vanishes as $\delta$ goes to $0$ and $n$ goes to infinity. We use the following inequality to simplify the results
\begin{equation}
    \label{eq:crude_tilde_omega_additivity}
    \tilde{\omega}_{k_1,m_1}(\delta,d_1d_3) + \tilde{\omega}_{k_2,m_2}(\delta,d_2d_3) \leq \tilde{\omega}_{k_1 + k_2, m_1 + m_2}(\delta, d_1 d_2 d_3).
\end{equation}

The theory of types has a natural extension to the quantum setting through the notion of typical subspaces. For a system $A$ with Hilbert space $\mathcal{H}_A$, we replace the alphabet notation $\mathbb{X}$ by $A$ to denote the set of symbols imposing the canonical basis of $\mathcal{H}_A$. In particular, this imposed basis could be the eigenbasis of a density operator $\rho_A \in \mathcal{D}(\mathcal{H}_A)$. If that is the case, we can choose a spectral decomposition of $\rho_A$ as
\begin{equation}
\label{eq:spectral-decomposition}
\rho_A = \sum_{a \in A} p(a) \ketbra{a}{a},
\end{equation}
in which $p$ is the distribution induced by the eigenvalues of $\rho_A$. For similar reasons, we often replace the notation of a distribution $p$ by $\rho$ to denote the use of the distribution induced by the eigenvalues of a density operator $\rho$. For a fixed spectral decomposition as in \eqref{eq:spectral-decomposition}, we can define for $\delta > 0$ the induced $\delta$-typical subspace $\mathcal{T}_{A^n}^{\rho,\delta}$ of $\mathcal{H}_A^{\otimes n}$ to be
\begin{equation}
\label{eq:delta-typical-subspace}
\mathcal{T}_{A^n}^{\rho,\delta} := \text{span} \left(\left\{ \ket{a^n} \in \mathcal{H}^{\otimes n} : a^n \in T^{\rho,\delta}_{A^n} \right\}\right).
\end{equation}
The projector onto $\mathcal{T}_{A^n}^{\rho,\delta}$ is denoted $\Pi_{A^n}^{\rho,\delta}$:
\begin{equation}
\label{eq:delta-typical-projector}
\Pi_{A^n}^{\rho,\delta} := \sum_{a^n \in T^{\rho,\delta}_{A^n}} \ketbra{a^n}{a^n}.
\end{equation}

The properties of the method of types translate naturally to the quantum setting. In particular, we can derive from equation \eqref{eq:typical-sequences-probability} that
\begin{equation}
\label{eq:delta-typical-subspace-probability}
\text{Tr}\left(\Pi_{A^n}^{\rho,\delta} \rho_A^{\otimes n}\right) = p^{\otimes n}(T^{\rho,\delta}_{A^n}) \geq 1 - \left(n + 1\right)^{|A|}2^{-n \frac{\delta^2}{2\ln(2)}},
\end{equation}
in which $|A| = \dim(\mathcal{H}_A)$. Similarly, we can derive from \eqref{eq:typical-sequences-size} that the dimension of the $\delta$-typical subspace is
\begin{equation}
\label{eq:delta-typical-subspace-dimension}
\text{Tr}\left(\Pi_{A^n}^{\rho,\delta}\right) = \left|T^{\rho,\delta}_{A^n}\right| \leq \exp{n \left(H(\rho) + \tilde{\omega}_{1,1}(\delta,d)\right)},
\end{equation}
in which $H(\rho)$ is the Von Neumann entropy of $\rho_A$.
\subsection{Interactive Protocols}
\label{preliminaries:InteractiveProtocols}

We define a quantum interactive communication protocol with classical inputs. Interactive protocols in the quantum settings involve multiple rounds of communication between two parties, Alice and Bob, in which at each round, one party sends a quantum system to the other party after applying a quantum operation on their local system combined with the quantum system they received in the previous round. Formally, a quantum interactive communication protocol $\Pi$ with $r$ rounds and classical inputs in the prior-free setting starts off with Alice and Bob receiving classical inputs $x$ and $y$, respectively, drawn from an unknown distribution over the product input set $\mathbb{X} \times \mathbb{Y}$. The protocol proceeds where at each round $i$, one party applies a quantum operation in the form of an isometry $U_i$ on their local quantum system and the quantum system received from the other party in the previous round (if any), and then sends a quantum system to the other party. After $r$ rounds of communication, both parties perform measurements on their local quantum systems to produce their respective outputs. Generally, each individual isometry $U_i$ can use some shared entanglement between Alice and Bob, which is assumed to be independent of the inputs $x$ and $y$. The total communication cost of the protocol is defined as the sum of the sizes (in qubits) of all quantum systems exchanged during the $r$ rounds of communication. 

We can assume without loss of generality that Alice sends the first message in the protocol. As a result, Alice applies the isometries $U_1, U_3, \ldots$ at odd rounds, while Bob applies the isometries $U_2, U_4, \ldots$ at even rounds. Similarly, the quantum systems $C_1, C_3, \ldots$ are sent by Alice at odd rounds, while the quantum systems $C_2, C_4, \ldots$ are sent by Bob at even rounds. At each round of communication, we label $A_i$ the system held by Alice after the $i$-th isometry $U_i$ has been applied, and $B_i$ the system held by Bob after the $i$-th isometry $U_i$ has been applied. In particular, we label $A_0 = X$ and $B_0 = Y$ the classical input systems held by Alice and Bob at the start of the protocol, respectively. At the end of the protocol, Alice and Bob each generate system $A'$ and $B'$, respectively, by applying their last isometry each. If the number of rounds $r$ is even, then Alice applies the last isometry $U_r$ and generates system $A'$, while Bob applies isometry $U_{r-1}$ and generates system $B'$. Conversely, if the number of rounds $r$ is odd, then Bob applies the last isometry $U_r$ and generates system $B'$, while Alice applies isometry $U_{r-1}$ and generates system $A'$. The reference system that purifies the joint input state of Alice and Bob is labeled $R_{XY}$.

In the asymptotic setting, we are interested in using noiseless quantum communication and pre-shared entanglement to simulate $n$ identical and independent executions of the protocol $\Pi$ on $n$ pairs of inputs $\Big((x_{1}, y_{1}), (x_{2}, y_{2}), ..., (x_{n}, y_{n})\Big)\in \left(\mathbb{X}\times\mathbb{Y}\right)^{n}$, which we denote as $(x^{n}, y^{n}) \in \mathbb{X}^{n}\times \mathbb{Y}^{n}$ for short. The prior-free constraint in the asymptotic setting means that the players have no knowledge of how their classical input sequences were generated, and that the simulation thus needs to work for the worst-case input. For a protocol~$\Pi$, we denote $\Pi(x,y)$ to be the output state of the protocol when Alice and Bob have classical inputs~$x$ and $y$ respectively. Similarly,~$\Pi^{\otimes n}(x^{n},y^{n})$ denotes the output state of the protocol $\Pi$ applied identically and independently to each pair of inputs in $(x^{n}, y^{n}) \in \mathbb{X}^{n}\times \mathbb{Y}^{n}$. We say a protocol $\hat{\Pi}_{n, \varepsilon}$ simulates the protocol $\Pi^{\otimes n}$ with error $\varepsilon$ if for every $(x^{n}, y^{n}) \in \mathbb{X}^{n} \times \mathbb{Y}^{n}$,
\begin{equation}
\label{eq:simulation_eq}
    \|\Pi^{\otimes n}(x^{n}, y^{n}) - \hat{\Pi}_{n, \varepsilon}(x^{n}, y^{n})\|_{1} < \varepsilon,
\end{equation}
in which $\|\cdot\|_{1}$ denotes the trace distance. 

For an interactive protocol $\Pi$, we define $QIC(\Pi)$ to be its prior-free quantum information cost defined as (see Refs.~\cite{braverman2015, touchette2015quantum}):
\begin{align}
    QIC(\Pi) &= \sum_{i>0, odd}\frac{1}{2}I(C_{i};R|B_{i-1})_{\mu} + \sum_{i>0, even} \frac{1}{2}I(C_{i};R|A_{i-1})_{\mu}.
\end{align}
The prior-free amortized quantum communication cost of a protocol $\Pi$ is
\begin{equation}
    AQCC(\Pi) = \lim_{\varepsilon \rightarrow 0}\lim_{n \rightarrow \infty} \min_{\hat{\Pi}_{n, \varepsilon}}\frac{QCC(\hat{\Pi}_{n, \varepsilon})}{n},
\end{equation}
in which the minimum is over simulation protocols $\hat{\Pi}_{n,\varepsilon}$ that satisfy \eqref{eq:simulation_eq}, and $QCC(\hat{\Pi}_{n, \varepsilon})$ denotes the number of qubits that are exchanged in either direction during the execution of the protocol $\hat{\Pi}_{n, \varepsilon}$.
\section{Type-Constrained Classical de Finetti Reduction}
In this section, we derive the first main result of our work. We prove a constrained de Finetti reduction in which we upper bound any classical permutation-invariant state that can be expressed as a convex combination of sequences $\delta$-typical to a distribution $p$, by a de Finetti state expressed as a uniform distribution, of i.i.d distributions of each type $\delta$-typical to $p$ (the type-constrained de Finetti state).

In the context of our application, the motivation behind these additional constraints is that we use as a subroutine a protocol that allows two parties to estimate the joint type between two sequences they are given. Therefore, we are guaranteed, with an exponentially vanishing probability of error, that the inputs to some protocol can be expressed as a convex combination of sequences $\delta$-typical to the distribution of some fixed type $t$. By our constrained de Finetti reduction, it is thus sufficient for us to design a protocol for a certain task to work well on our constrained de Finetti state, to then conclude it works well on all inputs with the aformentionned $\delta$-typicality structure.

\subsection{Type-Constrained de Finetti Reduction}
Before citing our result, we define the further notation and concepts. For a type $t \in \mathcal{P}_{\mathbb{X}^n}$, we represent it by a classical state with the following definition:
\begin{equation}
\label{eq:typedState}
\sigma_t := \sum_{x \in \mathbb{X}} t(x) \ketbra{x}{x},
\end{equation}
in which the $\ket{x}$ are the canonical basis states induced by the set of symbols $\mathbb{X}$. We note $\sigma_{t,C}$ to make it explicit that $\sigma_{t}$ lives in some Hilbert space $\mathcal{H}_{C}$ for some system $C$. Furthermore, since we are always working with a predefined set of symbols $\mathbb{X}$ that imposes the canonical basis of the space $\mathcal{H}_{X}$, we omit specifying $\mathbb{X}$ in the various notations if the label for the space $X$ is already present in a context where $\mathbb{X}$ should be present. We define the $\delta$-typical de Finetti state $\tau_{X^n}^{p,\delta}$ as follows:
\begin{equation}
\label{eq:deltaTypicalDeFinettiState}
\tau_{X^n}^{p,\delta} := \frac{1}{\left|\mathcal{P}_{\mathbb{X}^n}^{p,\delta}\right|} \sum_{t \in \mathcal{P}_{\mathbb{X}^n}^{p,\delta}} \sigma_{t,X}^{\otimes n}.
\end{equation}
We thus constrain our de Finetti state to a uniform mixture of a finite number of i.i.d. classical states that represent the distributions of all types $\delta$-typical to $p$, instead of a uniform mixture over all possible i.i.d. states. 

We define the classical symmetric subset as a classical analog of the symmetric subspace of a quantum system. The idea of the symmetric subset is to contain all the classical states that are invariant under permutation. We construct the symmetric subset by its analog of a basis vector, for which we abuse the terminology and call the following a classical symmetric basis state:
\begin{definition}
\label{def:classical-symmetric-basis-state}
Let $\mathbb{X}$ be a finite set of symbols that imposes the canonical basis of a Hilbert space $\mathcal{H}_{X}$. The \textit{classical symmetric basis state} $\overline{\omega}_{t,X^n}$ of type $t \in \mathcal{P}_{\mathbb{X}^n}$ is defined as
\begin{equation}
\label{eq:classical-symmetric-basis-state}
\overline{\omega}_{t,X^n} := \frac{1}{|T^t_{\mathbb{X}^n}|} \sum_{x^n \in T^t_{\mathbb{X}^n}} \ketbra{x^n}{x^n}_{X^n},
\end{equation}
\end{definition}
The idea behind this definition is that the type $t$ of some sequence $x^n$ is invariant under permutation because it only depends on the number of occurrences of each symbol in $x^n$. Furthermore, if a distribution~$p$ is permutation-invariant and contains non-zero weight on a sequence of type $t$, then it must contain the same weight on all sequences of type $t$. Therefore, any distribution $p$ that is permutation-invariant can be represented in density operator formalism as a convex linear combination of the classical symmetric basis states $\overline{\omega}_{t,X^n}$ for all $t \in \mathcal{P}_{\mathbb{X}^n}$. Therefore, we can define the subset of all permutation-invariant classical states as the set $\text{Sym}_{cl}^n(X)$ of convex combinations of all classical symmetric basis states.
\begin{definition}
\label{def:classical-symmetric-subset}
Let $\mathbb{X}$ be a finite set of symbols that imposes the canonical basis of a Hilbert space $\mathcal{H}_{X}$. The \textit{classical symmetric subset} $\text{Sym}_{cl}^n(X)$ is defined as the set of all convex combinations of the classical symmetric basis states $\overline{\omega}_{t,X^n}$ for all $t \in \mathcal{P}_{\mathbb{X}^n}$:
\begin{equation}
\label{eq:classical-symmetric-subset}
\text{Sym}_{cl}^n(X) := \left\{\sum_{t \in \mathcal{P}_{\mathbb{X}^n}} p_t \overline{\omega}_{t,X^n} : p_t \geq 0, \sum_{t \in \mathcal{P}_{\mathbb{X}^n}} p_t = 1\right\}.
\end{equation}
\end{definition}
The next lemma gives the connection between the classical $\delta$-typical symmetric states and our type-constrained de Finetti state. It informs us that the components $\sigma_t^{\otimes n}$ of the de Finetti state $\tau^{p,\delta}$ are lower-bounded by the classical symmetric basis states of the same type $t$ up to a factor of $\left|\mathcal{P}_{\mathbb{X}^n}\right|$, the number of types.
\begin{lemma}
\label{lemma:TypedStateLowerBound}
Let $t$ be a type on an alphabet $\mathbb{X}$, the latter imposing the canonical basis of a Hilbert space $\mathcal{H}_X$ and $\sigma_{t}$ be the state defined in \eqref{eq:typedState}. Let $n \in \mathbb{N}^*$ and let $\overline{\omega}_{t}$ be the classical symmetric basis state associated to the type $t$ on the Hilbert space $\mathcal{H}$, then
\begin{equation}
\label{eq:TypedStateLowerBound}
\sigma_t^{\otimes n} \geq \frac{1}{\left|\mathcal{P}_{\mathbb{X}^n}\right|}\overline{\omega}_t.
\end{equation}
\end{lemma}
The formal proof of this result is given in Appendix \ref{appendix:proofOfTypeStateLowerBound}. The idea is that the state $\overline{\omega}_t$ has all the weights on the type $t$, while the state $\sigma_t^{\otimes n}$ has weights on all the types, but with the highest weight on $t$. Therefore, the weight of $t$ on $\sigma_t^{\otimes n}$ is at least $\frac{1}{\left|\mathcal{P}_{\mathbb{X}^n}\right|}$, else there would be a type with higher weight than $t$ on $\sigma_t^{\otimes n}$, which would contradict the property that the probability distribution $t$ has the highest probability on the type $t$ (See Eq.~\eqref{eq:type-probability-comparison}). We can thus conclude that attenuating $\overline{\omega}_t$ by the number of types $\left|\mathcal{P}_{\mathbb{X}^n}\right|$ gives a lower bound on the state $\sigma_t^{\otimes n}$. Our de Finetti reduction is then a direct consequence of the previous lemma.
\begin{theorem}
    \label{thm:de_Finetti_Reduction_operator_inequality_form}
Let $\delta \in \mathbb{R}^+$ and $p$ be a distribution on an alphabet $\mathbb{X}$ that imposes the canonical basis of a Hilbert space $\mathcal{H}_{X}$. Let $\tau_{X^n}^{p,\delta}$ be the type-constrained de Finetti state defined in \eqref{eq:deltaTypicalDeFinettiState}, then for any $\delta$-typical, classical permutation invariant state $\rho_{X^n}\in~\mathcal{D}(\mathcal{T}_{X^n}^{p,\delta})~\cap~\text{Sym}_{cl}^n(X)$,
\begin{equation}
\label{eq:de_Finetti_Reduction_operator_inequality_form}
\rho_{X^n} \leq \left|\mathcal{P}_{\mathbb{X}^n}\right|\left|\mathcal{P}_{\mathbb{X}^n}^{p,\delta}\right|\tau_{X^n}^{p,\delta}.
\end{equation}
\end{theorem}
\begin{proof}
The only states inside the classical symmetric subset $\text{Sym}_{cl}^n(X)$ that have their support on the $\delta$-typical subspace $\mathcal{T}_{X^n}^{p,\delta}$ are the states of the form
\begin{equation}
    \rho_{X^n} = \sum_{t \in \mathcal{P}_{\mathbb{X}^n}^{p,\delta}} p_t \overline{\omega}_{t,X^n}.
\end{equation}
Upper bounding the states $\overline{\omega}_{t,X^n}$ by $\sigma_{t,X}^{\otimes n}$ using Lemma \ref{lemma:TypedStateLowerBound} gives
\begin{equation}
\label{eq:de_Finetti_Reduction_operator_inequality_form_proof}
\rho_{X^n} \leq \left|\mathcal{P}_{\mathbb{X}^n}\right|\sum_{t \in \mathcal{P}_{\mathbb{X}^n}^{p,\delta}} p_t \sigma_{t,X}^{\otimes n} \leq \left|\mathcal{P}_{\mathbb{X}^n}\right|\left|\mathcal{P}_{\mathbb{X}^n}^{p,\delta}\right|\tau_{X^n}^{p,\delta},
\end{equation}
\end{proof}

\subsection{Typical Classical-Quantum Channel Norm Reduction}
In the application of our de Finetti reduction, we consider evaluating a diamond norm on classical-quantum channels, but where the maximization is not over all possible input states, but rather over a restricted set of states that are classical and $\delta$-typical to some distribution $p$. For that reason, we define the following \textit{$\delta$-typical classical-quantum channel norm}.
\begin{definition}
Let $\mathcal{N}_{X^n \to A^n}$ be a classical-quantum channel. Let $\delta \in \mathbb{R}^+$ and $p$ be a distribution, defined on an alphabet $\mathbb{X}$ that imposes the canonical basis of the space $\mathcal{H}_{X}$. The \textit{$\delta$-typical classical-quantum channel norm} is defined as follows:
\begin{equation}
\label{eq:classicalQuantumDeltaDiamondNorm}
\|\mathcal{N}\|_{\diamond,cl}^{p,\delta} := \max_{\rho_{X^n}} \left\|\left(\mathcal{N} \otimes \textbf{id}_{R}\right)\left(\rho_{X^nR}\right)\right\|_1,
\end{equation}
in which the maximization is over all $\rho_{X^n} \in \mathcal{D}(\mathcal{T}_{X^n}^{p, \delta})$ that are also classical states. Here, $\rho_{X^nR}$ is a purification of the input state $\rho_{X^n}$ on the purification system $R$ and $\textbf{id}_{R}$ is the identity channel on that purification system.
\end{definition}
We prove the following theorem, which allows us to reduce the evaluation of the $\delta$-typical classical-quantum channel norm to evaluating a single trace norm on the type-constrained de Finetti state $\tau^{p,\delta}$. The result that follows and a part of its proof is inspired by the work of Ref.~\cite{christandl2009}.
\begin{theorem}
\label{thm:restrictedClassicalQuantumDeFinettiReduction} Let $\delta \in \mathbb{R}^+$ and $p$ be a distribution on an alphabet $\mathbb{X}$ that imposes the canonical basis of a Hilbert space $\mathcal{H}_{X}$. Let $\mathcal{M}_1$ and $\mathcal{M}_2$ be two classical-quantum channels from Hilbert spaces $\mathcal{H}_{X}^{\otimes n}$ to $\mathcal{H}_{A}^{\otimes n}$ and define $\Delta = \mathcal{M}_1 - \mathcal{M}_2$. If for all permutations $\pi$, there exists a CPTP map $\mathcal{K}_\pi: \mathcal{H}_{A}^{\otimes n} \to \mathcal{H}_{A}^{\otimes n}$ such that
\begin{equation}
\label{eq:deFinettiReductionCondition}
\Delta \circ \pi = \mathcal{K}_{\pi} \circ \Delta,
\end{equation}
then
\begin{equation}
\label{eq:deFinettiReductionResults}
\|\Delta\|_{\diamond,cl}^{p,\delta} \leq \left|\mathcal{P}_{\mathbb{X}^n}\right|\left|\mathcal{P}_{\mathbb{X}^n}^{p,\delta}\right|\left\|\left(\Delta \otimes \textbf{id}_{R'}\right)\left(\tau_{X^nR'}^{p,\delta}\right)\right\|_1,
\end{equation}
$\tau_{X^nR'}^{p,\delta}$ is a copy of the states inside the sum of the de Finetti state $\tau_{X^n}^{p,\delta}$ on a purification system $R'$, and $\textbf{id}_{R'}$ is the identity channel on that purification system.
\end{theorem}
The first following lemma proves that it is sufficient to consider the maximization on states that also belong to the symmetric subset $\text{Sym}^n_{\text{cl}}(X)$.
\begin{lemma}{(Symmetric Maximization Equivalence)}
\label{lemma:symmetricRestrictedMaximization} 
Let $\Delta$ be a linear map as in Theorem \ref{thm:restrictedClassicalQuantumDeFinettiReduction}, $n \in \mathbb{N}^*$, $\delta \in \mathbb{R}^+$ and $p$ be a distribution on the alphabet $\mathbb{X}$, the latter imposing the canonical basis of a Hilbert space $\mathcal{H}_{X}$. Let $\rho_{X^n} \in \mathcal{T}_{X^n}^{p,\delta}$ also be a classical state, then, there exists a state $\Tilde{\rho}_{X^n} \in \mathcal{T}_{X^n}^{p,\delta}\cap\text{Sym}^n_{\text{cl}}(X)$ such that
\begin{equation}
\label{eq:symmetricRestrictedMaximization}
\left\|\left(\Delta \otimes \textbf{id}_{R^n}\right)\left(\rho_{X^nR^n}\right)\right\|_1 = \left\|\left(\Delta \otimes \textbf{id}_{R^n}\right)\left(\Tilde{\rho}_{X^nR^n}\right)\right\|_1,
\end{equation}
in which $\rho_{X^nR^n}$ contains a copy of the state $\rho_{X^n}$ on the purification system $R^n$ (same goes for $\Tilde{\rho}_{X^nR^n}$) and $\textbf{id}_{R^n}$ is the identity channel on that purification (copy) system.
\end{lemma}
The idea is that the map $\Delta$ perceives the states of a same type as equivalent, since the map is covariant by permutation and a type is invariant by permutation. Therefore, for the map $\Delta$, the only difference between two states is the total contribution of each type (each symmetric basis states) in these two states. This implies the existence of an equivalent symmetric state for each state $\rho_{X^n}$. The proof can be found in Appendix \ref{appendix:symmetricRestrictedMaximization}.

Now that we have established an equivalence between any classical $\delta$-typical state and classical $\delta$-typical symmetric state under the action of the map $\Delta$, we are left to link the de Finetti state $\tau^{p,\delta}$ to the symmetric states. However, since we cannot directly use the operator inequality of Theorem \ref{thm:de_Finetti_Reduction_operator_inequality_form} to prove Theorem \ref{thm:restrictedClassicalQuantumDeFinettiReduction} (because the map $\Delta$ is a general linear map), we must prove one last lemma that builds, out of the inequality, a completely positive trace-non-increasing (CPTNI) map that bypasses the limitation from the presence of the linear map $\Delta$ to induce a norm inequality.
\begin{lemma}
\label{lemma:deFinettiReductionCPTNIMap}
Let $\delta \in \mathbb{R}^+$, $n \in \mathbb{N}^*$ and $p$ be a distribution on the alphabet $\mathbb{X}$, the latter imposing the canonical basis of a Hilbert space $\mathcal{H}_{X}$. Let $\Tilde{\rho}_{X^n} \in \mathcal{D}(\mathcal{T}_{\mathbb{X}^n}^{p,\delta})\cap\text{Sym}^n_{\text{cl}}(X)$, then there exists a CPTNI map $\mathcal{T}: End(\mathcal{H}_{R^n} \otimes \mathcal{H}_J) \to \mathbb{C}$ such that
\begin{equation}
\label{eq:deFinettiReductionCPTNIMap}
\Tilde{\rho}_{X^nR^n} = \left|\mathcal{P}_{\mathbb{X}^n}^{p,\delta}\right|\left|\mathcal{P}_{\mathbb{X}^n}\right|\left(\textbf{id}_{X^n} \otimes \mathcal{T}\right)\left(\tau^{p,\delta}_{X^nR^nJ}\right),
\end{equation}
with
$$\tau_{X^nR^nJ}^{p,\delta} := \frac{1}{\left|\mathcal{P}_{\mathbb{X}^n}^{p,\delta}\right|} \sum_{t \in \mathcal{P}_{\mathbb{X}^n}^{p,\delta}} \sigma_{t,XR}^{\otimes n} \otimes \ketbra{t}{t}_{J},$$
in which $J$ is a system that stores each type $t \in \mathcal{P}_{\mathbb{X}^n}^{p,\delta}$, and $\sigma_{t,XR}$ is the state defined in \eqref{eq:typedState} but with a copy on $R$, that is
$$\sigma_{t,XR} = \sum_{x \in X} t(x) \ketbra{x}{x}_{X} \otimes \ketbra{x}{x}_{R}.$$
\end{lemma}
The proof of this lemma is given in Appendix \ref{appendix:ProofOfCPTNIMap}. The idea is to use the inequality of Lemma \ref{lemma:TypedStateLowerBound} to build a CP map that uses the systems that are not used by the map $\Delta$ to change the de Finetti state $\tau^{p,\delta}$, which is composed of the $\sigma_t^{\otimes n}$ states, to a state of $\text{Sym}^n_{\text{cl}}(X)$ that is composed of the classical symmetric basis states $\overline{\omega}_t$. Since we have an inequality in Lemma \ref{lemma:TypedStateLowerBound}, this map is trace-non-increasing, and thus~CPTNI.

Now that we have a map that transforms the type-constrained de Finetti state to a symmetric state, the proof for the type-constrained de Finetti reduction (Theorem \ref{thm:restrictedClassicalQuantumDeFinettiReduction}) follows directly (presented in Appendix \ref{appendix:ProofClassicalQuantumdeFinettiReduction}). That is, by Lemma \ref{lemma:symmetricRestrictedMaximization}, for every state in the maximization, there exists another symmetric state that has the same norm, and thus we can use the CPTNI map of Lemma \ref{lemma:deFinettiReductionCPTNIMap} to transform the de Finetti state $\tau^{p,\delta}$ into that symmetric state. This gives a norm equality between the symmetric state and the de Finetti state on which is applied the map $\mathcal{T}$, but since the map is trace-non-increasing, we can remove it to induce a norm inequality, proving Theorem \ref{thm:restrictedClassicalQuantumDeFinettiReduction}.

\section{Prior-free Quantum Channel Simulation Protocols on Classical Inputs}
We now present the application of our type-constrained de Finetti reduction to the compression of communication across classical-quantum noisy channels. In this section, we consider the simulation setting of the Quantum Reverse Shannon Theorem (QRST) \cite{bennett_devetak_harrow_shor_winter_2014}. In the fully quantum QRST setting, the goal is for two parties Alice and Bob to simulate the use of an i.i.d noisy quantum channel $\mathcal{N}^{\otimes n}$ from Alice to Bob using a noiseless quantum channel and shared entanglement while minimizing the communication cost (a form of compression). We consider the prior-free setting, in which the channel simulation protocol should work for any input state. In our classical-quantum case, we further restrict the input states to be classical. We begin with the simpler setting where only Alice receives a classical input, which has already been studied for the fully quantum setting \cite{bennett_devetak_harrow_shor_winter_2014}, and subsequently we consider the setting where Bob also receives an input, which we refer to as his side information, and present our novel result for the prior-free QRST with side information at the receiver. 

In particular, we take the channel that we simulate to be an isometric extension $V_{\mathcal{N}, X \to EC}^{\otimes n}$ of the classical-quantum channel $\mathcal{N}^{\otimes n}$, in which $X$ is the input system, $E$ is the environment system of the isometric extension that Alice keeps and $C$ is the output system that she sends to Bob. The reason for simulating the isometric extension of the channel instead of the channel itself is that the isometric extension contains more information about the channel (the system $E$), and therefore simulating it is strictly stronger. Furthermore, it allows us to work in the pure state setting, which is necessary to use results such as Uhlmann's theorem \eqref{eq:uhlmann-theorem} that requires a state to be pure.

Both with and without side information at the receiver, we leverage our de Finetti reduction to only have to prove that our compression protocols work for de Finetti states in order to guarantee they work well for any input. The idea is to use a first subprotocol that allows Alice and and Bob to have additional information of, either the type of Alice's input for the case with no side information, or the joint type of Alice and Bob's input in the case with side information at the receiver, while ensuring that the communication cost of this subprotocol is satisfying enough. In the case with no side information, we have Alice directly communicate the type of her input to Bob using an amount of communication that is sublinear in $n$. In the case with side information at the receiver, we take this subprotocol to be the subsampling protocol from Lemma 1 of Ref.~\cite{padda2025}, which uses an amount of communication sublinear in $n$, to ensure that they both know an estimate $\tilde{t}$ of their inputs' joint type $t$, such that $\tilde{t}$ is in statistical distance~$\delta > 0$ of~$t$. In both cases, we now have the promise that their joint input state is in the~$\delta$-typical subspace of some estimated type~$\tilde{t}$ of their inputs' joint type~$t$. We then have that a protocol that reliably simulates~$V_{\mathcal{N}, X \to EC}^{\otimes n}$ and the transmission of its output~$C^n$ from Alice to Bob with our type-constrained de Finetti state associated to $\tilde{t}$ as input will, with high probability, reliably do the same on Alice and Bob's actual input. Therefore, our protocol that simulates $V_{\mathcal{N}, X \to EC}^{\otimes n}$ in the prior-free setting first gets an estimate~$\tilde{t}$ of the type of the inputs and then applies the simulation protocol corresponding to the type of the inputs. The statistical distance relation between~$t$ and~$\tilde{t}$ guarantees that this simulation protocol has satisfying communication cost for the prior-free setting.

To simulate $V_{\mathcal{N}}^{\otimes n}$ when the input is the type-constrained de Finetti state, we have Alice directly apply~$V_{\mathcal{N}}^{\otimes n}$ on her input, and we are thus only left with the task of transferring the system $C^n$ from Alice to Bob. To do so, we use a second subprotocol that does this efficiently while using pre-shared entanglement. In the case with no side information, one such efficient subprotocol is the state splitting protocol. In the case with side information at the receiver, we use the state redistribution protocol as a generalization of the state splitting protocol, which leverages the receiver's additional information to further reduce the communication cost. The state splitting and state redistribution protocols' communication costs depend on the dimensions of the states on which these protocols are applied. In both cases, we will use various $\delta$-typical projections on our type-constrained de Finetti state to get good bounds on these dimensions that relate to information-theoretic quantities and allow us to achieve the desired results.

We first consider the case with no side information at the receiver, and then move to the case with side information at the receiver in Section \ref{subsection:StateRedistributionWithSideInformation}.

\subsection{State Splitting without Side Information at the Receiver}
\label{subsection:stateRedistributionWithoutSideInformation}
In this section, we prove the following theorem regarding the simulation of classical-quantum channels with no side information at the receiver in the prior-free setting, which matches the more general result of Ref.~\cite{bennett_devetak_harrow_shor_winter_2014}.

\begin{theorem}
\label{thm:prior_free_cq_channel_simulation}
    Let $\mathcal{N}_{X \to C}$ be a classical-quantum channel sending classical states on system $X$ to quantum states on system $C$. Let $V_{\mathcal{N},X \to EC}$ be an isometric extension of $\mathcal{N}$, and let $n \in \mathbb{N}_{>0}$ and $\delta > 0$. There exists a variable length prior-free simulation protocol $\Pi_{\text{sim}}$ of~$V_{\mathcal{N}}^{\otimes n}$ such that for every input sequence $x^n \in X^n$, the trace distance error $\varepsilon_{x^n}$ of the simulation
    \begin{equation}
    \varepsilon_{x^n} := \left\|\left( \Pi_{\text{sim}} - V_{\mathcal{N}}^{\otimes n} \right) \left(\ketbra{x^n}{x^n}_{X^n}\right)\right\|_1
    \end{equation}
    is bounded as
    \begin{equation}
    \varepsilon_{x^n} \leq (n + 1)^{2|X|} \left(2\varepsilon_{\delta} + 2\sqrt{\varepsilon_{\text{dec},x^n}} \right),
    \end{equation}
    in which the decoupling error $\varepsilon_{\text{dec},x^n}$ and $\delta$-typical approximation error $\varepsilon_\delta$ are defined as
    \begin{align}
        \varepsilon_{\text{dec},x^n} &\leq \frac{\exp{n\left(I(X;C)_{t_{x^n}} + \tilde{\omega}_{9,6}(\delta,|X||C|) + \frac{f(\delta)}{n}\right)}}{d_{\tilde{C}}^2},\\
        \varepsilon_{\delta} &\leq 2^{-n\frac{\delta^2}{4\ln(2)} + 1}\sqrt{(n+1)^{|C|} + (n+1)^{|E|} + (n+1)^{|X|}}.
    \end{align}
    Here, the mutual information $I(X;C)_{t_{x^n}}$ is evaluated between the distribution $t_{x^n}$ of the type of the input $x^n$ on system $X$ and the quantum output of $t_{x^n}$ through $\mathcal{N}$ on system $C$, in the form of a classical-quantum state. The variable $d_{\tilde{C}}$ is the dimension of the communicated quantum system in the protocol, $\tilde{\omega}_{k,m}(\delta,d)$ is defined in \eqref{eq:tilde-omega} and $f(\delta)$ is defined as
    \begin{equation}
        f(\delta) := -2\log(1 - \frac{\left(\varepsilon_{\delta}\right)^2}{2}) + 1.
    \end{equation}
    Furthermore, the protocol uses pairs of maximally entangled qubits at a rate asymptotically equal to~$\frac{n}{2}I(C;X)_{t_{x^n}}$ up to an additive term sublinear in $n$. In order for the error $\varepsilon$ to vanish asymptotically as $n$ goes to infinity in the amortized setting, it is sufficient to choose the dimension of the communicated system $\tilde{C}$ such that
    \begin{equation}
    \frac{\log d_{\tilde{C}}}{n} = \frac{1}{2}I(X;C)_{t} + \tilde{\omega}_{9,6}(\delta,|X||C||E|) + \frac{f(\delta)}{n}.
    \end{equation}
    In particular, taking $\delta = n^{-1/2 + \alpha}$ for some $\alpha \in (0, 1/2)$ ensures that the error $\varepsilon$ vanishes as $n$ approaches infinity.
\end{theorem}
To achieve this result we apply the state splitting protocol on our type-constrained de Finetti state on which we have applied the isometric extension $V_{\mathcal{N}, X \to EC}^{\otimes n}$, which can be visualized in Figure \ref{fig:cq_channel_simulation_task}. We take the input de Finetti state to be the state
\begin{equation}
\label{eq:delta_typical_de_finetti_input_state_state_splitting}
\tau^{\tilde{t},\delta}_{X^n} := \frac{1}{\left|\mathcal{P}^{\tilde{t},\delta}_{\mathbb{X}^n}\right|} \sum_{t \in \mathcal{P}^{\tilde{t},\delta}_{\mathbb{X}^n}} \sigma_{t,X}^{\otimes n},
\end{equation}
in which the set $\mathcal{P}_{\mathbb{X}^n}^{\tilde{t},\delta}$ of types $\delta$-typical to $\tilde{t}$ is defined as in equation \eqref{eq:typical-types}. We can define a particular purification $\ket{\tau^{\tilde{t},\delta}}_{R_TR_X^nX^n}$ of this de Finetti state on the reference system $R := R_X^nR_T$, in which $R_T$ purifies the type indexation of the de Finetti state:
\begin{equation}
\label{eq:de_finetti_input_state_purification}
\ket{\tau^{\tilde{t},\delta}}_{R_TR_X^nX^n} := \frac{1}{\sqrt{\left|\mathcal{P}^{\tilde{t},\delta}_{\mathbb{X}^n}\right|}} \sum_{t \in \mathcal{P}^{\tilde{t},\delta}_{\mathbb{X}^n}} \ket{\sigma_{t}}_{R_XX}^{\otimes n} \otimes \ket{t}_{R_T},
\end{equation}
with $\ket{\sigma_{t}}_{R_XX}$ being the purification of the state $\sigma_{t,X}$ on the reference system $R_X$ which is of same dimension as system $X$. We denote the same state but on which we have applied the channel $V_{\mathcal{N}, X \to EC}^{\otimes n}$ as
\begin{equation}
\label{eq:de_finetti_output_state_state_splitting}
\ket{\tau^{\tilde{t},\delta}}_{R_TR_X^nE^nC^n} := \frac{1}{\sqrt{\left|\mathcal{P}^{\tilde{t},\delta}_{\mathbb{X}^n}\right|}} \sum_{t \in \mathcal{P}^{\tilde{t},\delta}_{\mathbb{X}^n}} \ket{\sigma_{t}'}_{R_XEC}^{\otimes n} \otimes \ket{t}_{R_T},
\end{equation}
in which $\ket{\sigma_{t}'}_{R_XEC} := V_{\mathcal{N}, X \to EC}\ket{\sigma_{t}}_{R_XX}$ is the state obtained by applying the isometric extension of the channel on the purification of $\sigma_{t,X}$. We prove the following proposition regarding the state splitting of the state in \eqref{eq:de_finetti_output_state_state_splitting} as an intermediate result.
\begin{proposition}
    \label{prop:de_finetti_state_state_splitting}
    Let $\ket{\tau^{\tilde{t},\delta}}_{R_TR_{X}^nE^nC^n}$ be a purified de Finetti state in the form of \eqref{eq:de_finetti_output_state_state_splitting}, then using state state splitting, the system $C^n$ can be transferred from a sender to a receiver up to trace distance precision $\varepsilon_{\text{ss}}$ bounded by
    \begin{equation}
        \varepsilon_{\text{ss}} \leq 2\varepsilon_{\delta} + 2\sqrt{\varepsilon_{\text{dec}}},
    \end{equation}
    in which
    \begin{align}
        \varepsilon_{\text{dec}} &\leq \frac{\exp{-n\left(\tilde{H} - \tilde{\omega}_{3,6}(\delta,|R_{X}C E|) - \frac{f(\delta)}{n}\right)}}{d_{\tilde{C}}^2},\\
        \varepsilon_{\delta} &\leq 2^{-n\frac{\delta^2}{4\ln(2)} + 1}\sqrt{(n+1)^{|C|} + (n+1)^{|E|} + (n+1)^{|R_X|}},
    \end{align}
    with
    \begin{equation}
        \tilde{H} := \max_{t \in T} H(C)_{t} + \max_{t \in T} H(R_{X})_{t} - \min_{t \in T} H(R_{X}C)_{t}.
    \end{equation}
    Here, $\tilde{\omega}_{k,m}(\delta,d)$ is defined in \eqref{eq:tilde-omega}, and $f(\delta)$ is defined as
    \begin{equation}
        f(\delta) := -2\log(1 - \frac{\left(\varepsilon_{\delta}\right)^2}{2}) + 1.
    \end{equation}
    This protocol requires the communication of a system of dimension $d_{\tilde{C}}$ and uses pre-shared entanglement between the sender and the receiver.
\end{proposition}
In the state splitting protocol, a state $\tau_{ECR}$ is shared between Alice and a reference system $R$ that purifies the state. Alice holds the systems $E$ and $C$ and her goal is to send the system $C$ to Bob while ensuring that the global state $\tau_{ECR}$ is unchanged. Therefore, using the state splitting protocol on our type-constrained de Finetti state on which we applied $V_{\mathcal{N}, X \to EC}^{\otimes n}$ allows to compress the transmission of the system $C^n$. The time-reversed version of this protocol is known as state merging, in which Bob is the one that owns $C$, and he must send it to Alice so that the two systems $E$ and $C$ ``merge". Therefore, in order to construct a protocol for state splitting, we construct a protocol for the state merging and then reverse it for convenience. 

To construct the state merging protocol, we use the following decoupling theorem. The idea of the decoupling theorem is that sampling a unitary at random (from the Haar measure) decouples, with high probability, our state from the reference system. In particular, it insures that the decoupled state is close to a maximally mixed state, the latter being the reduced state of a maximally entangled state, which can be recovered by Bob through the pre-shared entanglement he shares with Alice, and can reduce the communication cost of the protocol. We use the bounds of Ref.~\cite{WalterSymmetryQI} for the decoupling theorem.
\begin{theorem}[Decoupling theorem \cite{Abeyesinghe_2009}]
\label{thm:decoupling}
Let $\Psi_{AR}$ be a density operator on $\mathbb{C}^{d_A} \otimes \mathbb{C}^{d_R}$, in which~$d_A~=~d_{A_1}d_{A_2}$. Then,
\begin{align*}
    &\int dU_A \left\| \text{Tr}_{A_1}\left[\left(U_A \otimes I_R\right) \Psi_{AR} \left(U_A^\dagger \otimes I_R\right)\right] - \frac{I_{A_2}}{d_{A_2}} \otimes \Psi_R \right\|_1^2 \leq \frac{d_Ad_R}{d_{A_1}^2}\text{Tr}\left[\Psi_{AR}^2\right].
\end{align*}
\end{theorem}

The other tool that is used in the state splitting protocol is Uhlmann's theorem, which links two purifications of a state by an isometry. The decoupling theorem combined with Uhlmann's theorem and shared entanglement yields the state merging protocol, which is the protocol we use to compress the transmission of the system $C^n$ from Alice to Bob (in reverse). Figure \ref{fig:cq_channel_simulation_protocol} shows a visual of the protocol that performs state splitting of our type-constrained de Finetti state.

\begin{figure}[t]
    \centering
    \resizebox{0.75\textwidth}{!}{
    \begin{circuitikz}
    \tikzstyle{every node}=[font=\fontsize{18.2pt}{23.7pt}\selectfont]
    \draw [ color={rgb,255:red,169; green,186; blue,255}, draw opacity=1, line width=0.5pt, short] (21.125,15.625) .. controls (17.375,15.625) and (17.5,15.625) .. (13.75,15.625);
    \draw [ fill={rgb,255:red,218; green,217; blue,255}, fill opacity=1] (15,18.5) rectangle (16.25,16.625);
    \node [font=\fontsize{18.2pt}{23.7pt}\selectfont, inner xsep=0.080cm, inner ysep=0.085cm, rounded corners=0.020cm] at (15.65,17.5) {$V_{\mathcal{N}}^{\otimes n}$};
    \draw [ color={rgb,255:red,52; green,82; blue,199}, draw opacity=1, line width=1pt, short] (15,17.5) -- (11.25,17.5);
    \draw [ color={rgb,255:red,52; green,82; blue,199}, draw opacity=1, line width=1pt, short] (13.75,21.25) .. controls (12.5,19.375) and (12.5,19.375) .. (11.25,17.5);
    \draw [ color={rgb,255:red,52; green,82; blue,199}, draw opacity=1, line width=1pt, short] (21.25,21.25) .. controls (17.5,21.25) and (17.5,21.25) .. (13.75,21.25);
    \draw [ color={rgb,255:red,52; green,82; blue,199}, draw opacity=1, line width=1pt, short] (23.75,18) .. controls (22.5,19.625) and (22.5,19.625) .. (21.25,21.25);
    \draw [ color={rgb,255:red,52; green,82; blue,199}, draw opacity=1, line width=1pt, short] (23.75,18) .. controls (20,18) and (20,18) .. (16.25,18);
    \draw [ color={rgb,255:red,52; green,82; blue,199}, draw opacity=1, line width=1pt, short] (17.5,17) .. controls (16.875,17) and (16.875,17) .. (16.25,17);
    \draw [ color={rgb,255:red,52; green,82; blue,199}, draw opacity=1, line width=1pt, short] (17.5,13.75) .. controls (17.5,15.375) and (17.5,15.375) .. (17.5,17);
    \draw [ color={rgb,255:red,52; green,82; blue,199}, draw opacity=1, line width=1pt, short] (21.25,13.75) .. controls (19.375,13.75) and (19.375,13.75) .. (17.5,13.75);
    \draw [ color={rgb,255:red,52; green,82; blue,199}, draw opacity=1, line width=1pt, short] (23.75,18) .. controls (22.5,15.875) and (22.5,15.875) .. (21.25,13.75);
    \draw [ color={rgb,255:red,169; green,186; blue,255}, draw opacity=1, line width=0.5pt, short] (21.25,19.375) .. controls (17.5,19.375) and (17.5,19.375) .. (13.75,19.375);
    \node [font=\fontsize{18.2pt}{23.7pt}\selectfont, fill={rgb,255:red,255; green,255; blue,255}, fill opacity=1, text opacity=1, inner xsep=0.080cm, inner ysep=0.085cm, rounded corners=0.020cm] at (14.25,21.75) {$R$};
    \node [font=\fontsize{18.2pt}{23.7pt}\selectfont, fill={rgb,255:red,255; green,255; blue,255}, fill opacity=1, text opacity=1, inner xsep=0.080cm, inner ysep=0.085cm, rounded corners=0.020cm] at (17,17.375) {$C^n$};
    \node [font=\fontsize{18.2pt}{23.7pt}\selectfont, fill={rgb,255:red,255; green,255; blue,255}, fill opacity=1, text opacity=1, inner xsep=0.080cm, inner ysep=0.085cm, rounded corners=0.020cm] at (17,18.375) {$E^n$};
    \node [font=\fontsize{18.2pt}{23.7pt}\selectfont, fill={rgb,255:red,255; green,255; blue,255}, fill opacity=1, text opacity=1, inner xsep=0.080cm, inner ysep=0.085cm, rounded corners=0.020cm] at (14.25,17.875) {$X^n$};
    \node [font=\fontsize{18.2pt}{23.7pt}\selectfont, fill={rgb,255:red,255; green,255; blue,255}, fill opacity=1, text opacity=1, inner xsep=0.080cm, inner ysep=0.085cm, rounded corners=0.020cm] at (10.375,17.625) {$\tau_{\text{in}}$};
    \node [font=\fontsize{18.2pt}{23.7pt}\selectfont, fill={rgb,255:red,255; green,255; blue,255}, fill opacity=1, text opacity=1, inner xsep=0.080cm, inner ysep=0.085cm, rounded corners=0.020cm] at (24.625,18) {$\tau_{\text{out}}$};
    \node [font=\fontsize{18.2pt}{23.7pt}\selectfont, fill={rgb,255:red,255; green,255; blue,255}, fill opacity=1, text opacity=1, inner xsep=0.080cm, inner ysep=0.085cm, rounded corners=0.020cm] at (7.5,20.25) {$Reference$};
    \node [font=\fontsize{18.2pt}{23.7pt}\selectfont, fill={rgb,255:red,255; green,255; blue,255}, fill opacity=1, text opacity=1, inner xsep=0.080cm, inner ysep=0.085cm, rounded corners=0.020cm] at (8.125,17.5) {$Alice$};
    \node [font=\fontsize{18.2pt}{23.7pt}\selectfont, fill={rgb,255:red,255; green,255; blue,255}, fill opacity=1, text opacity=1, inner xsep=0.080cm, inner ysep=0.085cm, rounded corners=0.020cm] at (8.125,14.625) {$Bob$};
    \end{circuitikz}
    }%
    \caption{Visual representation of the QRST task that has to be simulated on a type-constrained de Finetti state. Alice owns a classical system $X^n$ of some input de Finetti state $\tau_{\text{in}}$, which is purified by some reference system $R$. She applies the isometric extension $V_{\mathcal{N}}^{\otimes n}$ of the classical-quantum channel $\mathcal{N}^{\otimes n}$ to obtain the de Finetti state $\tau_{\text{out}}$, from which she keeps the environment system $E^n$ and sends the output system $C^n$ to Bob.}
    \label{fig:cq_channel_simulation_task}
\end{figure}

\textit{Proof of Proposition \ref{prop:de_finetti_state_state_splitting}.} We relabel the purified output type-constrained de Finetti state as follows:
\begin{equation}
\label{eq:de_finetti_output_state_state_splitting_relabeled}
\ket{\tau_0^{\tilde{t},\delta}}_{R_TR_X^nE_0^nC_0^n} := \frac{1}{\sqrt{\left|\mathcal{P}^{\tilde{t},\delta}_{\mathbb{X}^n}\right|}} \sum_{t \in \mathcal{P}^{\tilde{t},\delta}_{\mathbb{X}^n}} \ket{\sigma_{0,t}'}_{R_XE_0C_0}^{\otimes n} \otimes \ket{t}_{R_T}.
\end{equation}
We use the subscripts $s = 0, 1, 2$ to specify the exact time of the protocol in Figure \ref{fig:cq_channel_simulation_protocol} that each state corresponds to. At time $s=0$, the full state is $\ket{\tau_0^{\tilde{t},\delta}}_{R_TR_X^nE_0^nC_0^n} \otimes \ket{\Phi}_{\tilde{A}\tilde{B}}$ with $\ket{\Phi}_{\tilde{A}\tilde{B}}$ being the pre-shared entanglement between Alice (system $\tilde{A}$) and Bob (system $\tilde{B}$) that is independent of the input. At the end of the state splitting protocol, the desired state is labeled $\ket{\tau_2^{\tilde{t},\delta}}_{R_TR_X^nE_2^nC_2^n}$, and is the same as $\ket{\tau_0^{\tilde{t},\delta}}_{R_TR_X^nE_0^nC_0^n}$, but in which Bob owns the system $C_2$ which is isomorphic to $C_0$ (and $E_0$ is isomorphic to $E_2$).
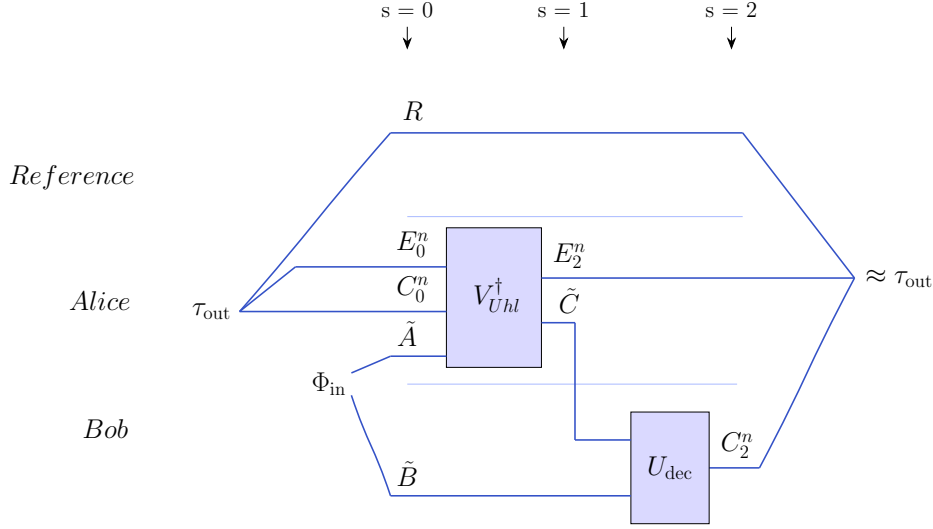
\begin{figure*}[t]
    \centering
    \resizebox{0.75\textwidth}{!}{
    \begin{circuitikz}
    \tikzstyle{every node}=[font=\fontsize{14.2pt}{18.5pt}\selectfont]
    \node [font=\fontsize{18.2pt}{23.7pt}\selectfont, fill={rgb,255:red,255; green,255; blue,255}, fill opacity=1, text opacity=1, inner xsep=0.080cm, inner ysep=0.085cm, rounded corners=0.020cm] at (24.75,18) {$\approx \tau_{\text{out}}$};
    \node [font=\fontsize{18.2pt}{23.7pt}\selectfont, fill={rgb,255:red,255; green,255; blue,255}, fill opacity=1, text opacity=1, inner xsep=0.080cm, inner ysep=0.085cm, rounded corners=0.020cm] at (9.375,17.25) {$\tau_{\text{out}}$};
    \draw [ color={rgb,255:red,169; green,186; blue,255}, draw opacity=1, line width=0.5pt, short] (21.125,15.625) .. controls (17.375,15.625) and (17.5,15.625) .. (13.75,15.625);
    \draw [ fill={rgb,255:red,218; green,217; blue,255}, fill opacity=1] (14.625,19.125) rectangle (16.75,16);
    \node [font=\fontsize{18.2pt}{23.7pt}\selectfont, inner xsep=0.080cm, inner ysep=0.085cm, rounded corners=0.020cm] at (15.75,17.625) {$V^{\dagger}_{Uhl}$};
    \draw [ color={rgb,255:red,52; green,82; blue,199}, draw opacity=1, line width=1pt, short] (14.625,18.25) -- (11.25,18.25);
    \draw [ color={rgb,255:red,52; green,82; blue,199}, draw opacity=1, line width=1pt, short] (13.375,21.25) .. controls (11.625,19.25) and (11.75,19.25) .. (10,17.25);
    \draw [ color={rgb,255:red,52; green,82; blue,199}, draw opacity=1, line width=1pt, short] (21.25,21.25) .. controls (17.25,21.25) and (17.375,21.25) .. (13.375,21.25);
    \draw [ color={rgb,255:red,52; green,82; blue,199}, draw opacity=1, line width=1pt, short] (23.75,18) .. controls (22.5,19.625) and (22.5,19.625) .. (21.25,21.25);
    \draw [ color={rgb,255:red,52; green,82; blue,199}, draw opacity=1, line width=1pt, short] (23.75,18) .. controls (20.25,18) and (20.25,18) .. (16.75,18);
    \draw [ color={rgb,255:red,52; green,82; blue,199}, draw opacity=1, line width=1pt, short] (17.5,17) .. controls (17.125,17) and (17.125,17) .. (16.75,17);
    \draw [ color={rgb,255:red,52; green,82; blue,199}, draw opacity=1, line width=1pt, short] (17.5,14.375) .. controls (17.5,15.625) and (17.5,15.625) .. (17.5,17);
    \draw [ color={rgb,255:red,52; green,82; blue,199}, draw opacity=1, line width=1pt, short] (21.625,13.75) .. controls (21,13.75) and (20.5,13.75) .. (20.5,13.75);
    \draw [ color={rgb,255:red,52; green,82; blue,199}, draw opacity=1, line width=1pt, short] (23.75,18) .. controls (22.625,15.875) and (22.75,15.875) .. (21.625,13.75);
    \draw [ color={rgb,255:red,169; green,186; blue,255}, draw opacity=1, line width=0.5pt, short] (21.25,19.375) .. controls (17.5,19.375) and (17.5,19.375) .. (13.75,19.375);
    \node [font=\fontsize{18.2pt}{23.7pt}\selectfont, fill={rgb,255:red,255; green,255; blue,255}, fill opacity=1, text opacity=1, inner xsep=0.080cm, inner ysep=0.085cm, rounded corners=0.020cm] at (13.875,21.75) {$R$};
    \node [font=\fontsize{18.2pt}{23.7pt}\selectfont, fill={rgb,255:red,255; green,255; blue,255}, fill opacity=1, text opacity=1, inner xsep=0.080cm, inner ysep=0.085cm, rounded corners=0.020cm] at (17.375,17.5) {$\tilde{C}$};
    \node [font=\fontsize{18.2pt}{23.7pt}\selectfont, fill={rgb,255:red,255; green,255; blue,255}, fill opacity=1, text opacity=1, inner xsep=0.080cm, inner ysep=0.085cm, rounded corners=0.020cm] at (17.375,18.5) {$E_2^n$};
    \node [font=\fontsize{18.2pt}{23.7pt}\selectfont, fill={rgb,255:red,255; green,255; blue,255}, fill opacity=1, text opacity=1, inner xsep=0.080cm, inner ysep=0.085cm, rounded corners=0.020cm] at (13.875,18.75) {$E_0^n$};
    \node [font=\fontsize{18.2pt}{23.7pt}\selectfont, fill={rgb,255:red,255; green,255; blue,255}, fill opacity=1, text opacity=1, inner xsep=0.080cm, inner ysep=0.085cm, rounded corners=0.020cm] at (6.25,20.25) {$Reference$};
    \node [font=\fontsize{18.2pt}{23.7pt}\selectfont, fill={rgb,255:red,255; green,255; blue,255}, fill opacity=1, text opacity=1, inner xsep=0.080cm, inner ysep=0.085cm, rounded corners=0.020cm] at (6.875,17.5) {$Alice$};
    \node [font=\fontsize{18.2pt}{23.7pt}\selectfont, fill={rgb,255:red,255; green,255; blue,255}, fill opacity=1, text opacity=1, inner xsep=0.080cm, inner ysep=0.085cm, rounded corners=0.020cm] at (7,14.625) {$Bob$};
    \draw [ color={rgb,255:red,52; green,82; blue,199}, draw opacity=1, line width=1pt, short] (14.625,17.25) -- (10,17.25);
    \draw [ color={rgb,255:red,52; green,82; blue,199}, draw opacity=1, line width=1pt, short] (14.625,16.25) -- (13.375,16.25);
    \node [font=\fontsize{18.2pt}{23.7pt}\selectfont, fill={rgb,255:red,255; green,255; blue,255}, fill opacity=1, text opacity=1, inner xsep=0.080cm, inner ysep=0.085cm, rounded corners=0.020cm] at (13.875,17.75) {$C_0^n$};
    \node [font=\fontsize{18.2pt}{23.7pt}\selectfont, fill={rgb,255:red,255; green,255; blue,255}, fill opacity=1, text opacity=1, inner xsep=0.080cm, inner ysep=0.085cm, rounded corners=0.020cm] at (13.75,16.75) {$\tilde{A}$};
    \draw [ fill={rgb,255:red,218; green,217; blue,255}, fill opacity=1] (18.75,15) rectangle (20.5,12.5);
    \draw [ color={rgb,255:red,52; green,82; blue,199}, draw opacity=1, line width=1pt, short] (17.5,14.375) .. controls (18.125,14.375) and (18.125,14.375) .. (18.75,14.375);
    \draw [ color={rgb,255:red,52; green,82; blue,199}, draw opacity=1, line width=1pt, short] (13.375,13.125) .. controls (16,13.125) and (16.125,13.125) .. (18.75,13.125);
    \draw [ color={rgb,255:red,52; green,82; blue,199}, draw opacity=1, line width=1pt, short] (12.5,15.375) .. controls (12.875,14.25) and (13,14.25) .. (13.375,13.125);
    \node [font=\fontsize{18.2pt}{23.7pt}\selectfont, fill={rgb,255:red,255; green,255; blue,255}, fill opacity=1, text opacity=1, inner xsep=0.080cm, inner ysep=0.085cm, rounded corners=0.020cm] at (12,15.625) {$\Phi_{\text{in}}$};
    \draw [ color={rgb,255:red,52; green,82; blue,199}, draw opacity=1, line width=1pt, short] (11.25,18.25) .. controls (10.625,17.75) and (10.625,17.75) .. (10,17.25);
    \node [font=\fontsize{18.2pt}{23.7pt}\selectfont, inner xsep=0.080cm, inner ysep=0.085cm, rounded corners=0.020cm] at (19.625,13.75) {$U_{\text{dec}}$};
    \node [font=\fontsize{18.2pt}{23.7pt}\selectfont, fill={rgb,255:red,255; green,255; blue,255}, fill opacity=1, text opacity=1, inner xsep=0.080cm, inner ysep=0.085cm, rounded corners=0.020cm] at (13.75,13.625) {$\tilde{B}$};
    \draw [ color={rgb,255:red,52; green,82; blue,199}, draw opacity=1, line width=1pt, short] (13.375,16.25) -- (12.5,15.875);
    \node [font=\fontsize{14.2pt}{18.5pt}\selectfont, fill={rgb,255:red,255; green,255; blue,255}, fill opacity=1, text opacity=1, inner xsep=0.080cm, inner ysep=0.085cm, rounded corners=0.020cm] at (13.75,24) {s = 0};
    \draw [-{Stealth[scale=1.5]}, ] (13.75,23.625) -- (13.75,23.125);
    \node [font=\fontsize{14.2pt}{18.5pt}\selectfont, fill={rgb,255:red,255; green,255; blue,255}, fill opacity=1, text opacity=1, inner xsep=0.080cm, inner ysep=0.085cm, rounded corners=0.020cm] at (17.25,24) {s = 1};
    \draw [-{Stealth[scale=1.5]}, ] (17.25,23.625) -- (17.25,23.125);
    \node [font=\fontsize{14.2pt}{18.5pt}\selectfont, fill={rgb,255:red,255; green,255; blue,255}, fill opacity=1, text opacity=1, inner xsep=0.080cm, inner ysep=0.085cm, rounded corners=0.020cm] at (21,24) {s = 2};
    \draw [-{Stealth[scale=1.5]}, ] (21,23.625) -- (21,23.125);
    \node [font=\fontsize{18.2pt}{23.7pt}\selectfont, fill={rgb,255:red,255; green,255; blue,255}, fill opacity=1, text opacity=1, inner xsep=0.080cm, inner ysep=0.085cm, rounded corners=0.020cm] at (21.125,14.25) {$C_2^n$};
    \end{circuitikz}
    }%
    \caption{Visual representation of the state redistribution protocol applied to the output de Finetti state $\tau_{\text{out}}$ obtained after passing the input type-constrained de Finetti state through the isometric extension $V_{\mathcal{N}}^{\otimes n}$ of the classical-quantum channel $\mathcal{N}^{\otimes n}$. At time $s=0$, Alice and Bob share the entanglement state $\Phi_{{\text{in}}}$ from which Alice owns the system $\tilde{A}$ and Bob owns the system $\tilde{B}$. Alice owns the systems $E_0^n$ and $C_0^n$ of the output de Finetti state $\tau_{\text{out}}$ which is purified onto some system $R$. At time $s=1$, Alice applied the Uhlmann isometry $V_\text{Uhl}^{\dagger}$ to produce the environment system $E_2^n$, which is isomorphic to the system $E_0^n$, and to the produce the system $\tilde{C}$ that she sends to Bob. Finally, at time $s=2$, Bob applied the unitary $U_{\text{dec}}$ to recover the output system $C_2^n$, which is isomorphic to $C_0^n$ and the the state on all the systems is a state close to $\tau_{\text{out}}$, with Bob owning the equivalent to the initial system $C_0^n$ that Alice used to own.}
    \label{fig:cq_channel_simulation_protocol}
\end{figure*}

Before using the decoupling theorem, we project our states on various $\delta$-typical subspaces to reduce the dimensions on which it lives while keeping a state relatively close to the original state, which allows the decoupling theorem to give good communication cost on the state splitting protocol. Let the projectors $\Pi^{t,\delta}_{R_X^n}$, $\Pi^{t,\delta}_{E_0^n}$ and $\Pi^{t,\delta}_{C_0^n}$ be the projectors on the $\delta$-typical subspaces of the reduced states ${\sigma'}_{0,t,R_X}^{\otimes n}$, ${\sigma'}_{0,t,E_0}^{\otimes n}$ and ${\sigma'}_{0,t,C_0}^{\otimes n}$ respectively (in the sense of \eqref{eq:delta-typical-projector}). We define a projector on the $\delta$-typical subspaces of every type as
\begin{equation}
\label{eq:projector_on_typical_subspaces_cq_channel_simulation}
\Tilde{\Pi}^{\tilde{t},\delta}_{R_TR_X^nE_0^nC_0^n} := \sum_{t \in \mathcal{P}^{\tilde{t},\delta}_{\mathbb{X}^n}} \Pi^{t,\delta}_{R_X^n} \otimes \Pi^{t,\delta}_{E_0^n} \otimes \Pi^{t,\delta}_{C_0^n} \otimes \ketbra{t}{t}_{R_T}.
\end{equation}
We can then define the projected de Finetti state at time $s=0$ as
\begin{equation}
\label{eq:projected_de_finetti_state_at_time_2_cq_channel_simulation}
\ket{\tilde{\tau}_0^{\tilde{t},\delta}}_{R_TR_X^nE_0^nC_0^n} := \frac{\tilde{\Pi}^{\tilde{t},\delta}_{R_TR_X^nE_0^nC_0^n} \ket{\tau_0^{\tilde{t},\delta}}_{R_TR_X^nE_0^nC_0^n}}{\left\| \tilde{\Pi}^{\tilde{t},\delta}_{R_TR_X^nE_0^nC_0^n} \ket{\tau_0^{\tilde{t},\delta}}_{R_TR_X^nE_0^nC_0^n} \right\|_2}.
\end{equation}
Similarly, since $\ket{\tau_2^{\tilde{t},\delta}}_{R_TR_X^nE_2^nC_2^n}=\ket{\tau_0^{\tilde{t},\delta}}_{R_TR_X^nE_0^nC_0^n}$, we can define $\ket{\tilde{\tau}_2^{\tilde{t},\delta}}_{R_TR_X^nE_2^nC_2^n}$ to be the same renormalized projected state but on systems $E_2$ and $C_2$ instead of $E_0$ and $C_0$. The following lemma gives an explicit bound on a general $\delta$-typical approximation of our type-constrained de Finetti state.
\begin{lemma}
    \label{lemma:deltaTypicalTraceDistanceBound}
    Let $\delta > 0$, $n,k \in \mathbb{N}_{>0}$ and $\tilde{t}$ be a fixed type in $\mathcal{P}_{\mathbb{X}^n}$ where $\mathbb{X}$ is a finite alphabet. Let $\tau^{\tilde{t},\delta}_{R_TA_1^n \ldots A_k^n}$ be a pure $\delta$-typical de Finetti state of $\tilde{t}$ defined on systems $A_1^n, A_2^n, \ldots, A_k^n$ with reference system $R_T$ purifying the classical type system:
    \begin{equation} 
    \ket{\tau^{\tilde{t},\delta}}_{R_TA_1 \ldots A_n} = \frac{1}{\sqrt{\left|\mathcal{P}_{\mathbb{X}^n}^{\tilde{t},\delta}\right|}} \sum_{t \in \mathcal{P}_{\mathbb{X}^n}^{\tilde{t},\delta}} \ket{\sigma_t}_{A_1\ldots A_k}^{\otimes n} \otimes \ket{t}_{R_T}.
    \end{equation}
    Define $\ket{\tilde{\tau}^{\tilde{t},\delta}}_{R_TA_1 \ldots A_n}$, the non-normalized projected de Finetti state on the $\delta$-typical subspace of each system~$A_i$:
    \begin{equation}
        \label{eq:projected_de_finetti_state_general}
    \ket{\tilde{\tau}^{\tilde{t},\delta}} := \left(\sum_{t \in \mathcal{P}_{\mathbb{X}^n}^{\tilde{t},\delta}} \left(\bigotimes_{i=1}^k \Pi_{A_i^n}^{t,\delta} \otimes \ketbra{t}{t}_{R_T}\right)\right) \ket{\tau^{\tilde{t},\delta}},
    \end{equation}
    in which $\Pi_{A_i^n}^{t,\delta}$ is the $\delta$-typical projector of the state $\sigma_t$ on system $A_i$. Then,~$\varepsilon_{\delta}~:=~\left\|\tilde{\tau}^{\tilde{t},\delta} - \tau^{\tilde{t},\delta}\right\|_1$ is upper-bounded as follows:
    \begin{equation}
        \varepsilon_{\delta} \leq 2^{-n\frac{\delta^2}{4\ln(2)} + 1}\sqrt{\sum_{i=1}^k |\mathcal{P}_{A_i^n}|}.
    \end{equation}
    Furthermore, this bound is the same if we instead take the renormalized state after projection in \eqref{eq:projected_de_finetti_state_general}.
\end{lemma}
\noindent We prove this lemma in Appendix~\ref{appendix:typicalProjectionsBounds}. Note that this bound vanishes exponentially fast in $n$ for fixed $\delta > 0$ and fixed systems~$\left\{A_i^n\right\}_{i=1}^k$. Applied to the projection in \eqref{eq:projected_de_finetti_state_at_time_2_cq_channel_simulation}, this lemma gives us
\begin{equation}
    \label{eq:precision_on_state_splitting_projections_at_time_0_and_2}
    \varepsilon_{\delta} \leq 2^{-n\frac{\delta^2}{4\ln(2)} + 1}\sqrt{(n+1)^{|C|} + (n+1)^{|E|} + (n+1)^{|R_X|}},
\end{equation}
in which we implicitly used the bound in \eqref{eq:number-of-types}. The same bound applies for the precision of the $\delta$-typical approximation $\ket{\tilde{\tau}_2^{\tilde{t},\delta}}_{R_TR_X^nE_2^nC_2^n}$ of $\ket{\tau_2^{\tilde{t},\delta}}_{R_TR_X^nE_2^nC_2^n}$.

Let us label $R_X^\delta$, $E_0^\delta$ and $C_0^\delta$ the systems obtained by projecting the systems $R_X^n$, $E_0^n$ and $C_0^n$ using the projectors~$\sum_t \Tilde{\Pi}^t_{R_X^n} \otimes \ketbra{t}{t}$,~$\sum_t \Tilde{\Pi}^t_{E_0^n} \otimes \ketbra{t}{t}$ and~$\sum_t \Tilde{\Pi}^t_{C_0^n} \otimes \ketbra{t}{t}$ respectively, and do the same for $E_2$ and $C_2$ by defining $E_2^{\delta}$ and $C_2^\delta$ (Note that the product of these three projectors gives $\tilde{\Pi}^{\tilde{t},\delta}_{R_TR_X^nE_0^nC_0^n}$). We can now rewrite $\ket{\tilde{\tau}_2^{\tilde{t},\delta}}_{R_TR_X^nE_0^nC_0^n}$ as a state on systems $R_TR_X^\delta E_0^\delta C_0^\delta$. Note that we can always decompose the system $C_2^\delta$ into two subsystems $\tilde{C}$ and $\tilde{B}$ such that $C_2^\delta = \tilde{C} \otimes \tilde{B}$. Applying the decoupling theorem on the state $\tilde{\tau}_{2, R_TR_X^\delta E_2^\delta C_2^\delta}^{\tilde{t},\delta}$ reduced to the systems $R_TR_X^\delta C_2^\delta$, we obtain the existence of a unitary $U_{C_2^\delta}$ acting on the system $C_2^\delta$ such that
\begin{equation}
    \label{eq:decoupling_inequality_cq_channel_simulation}
    \left\| \text{Tr}_{\tilde{C}}\left[\left(I_{R_TR_X^\delta} \otimes U_{C_2^\delta}\right) \tilde{\tau}_{2, R_TR_X^\delta C_2^\delta}^{\tilde{t},\delta} \left(I_{R_TR_X^\delta} \otimes U_{C_2^\delta}\right)\right] - \frac{I_{\tilde{B}}}{d_{\tilde{B}}} \otimes \tilde{\tau}_{1,R_TR_X^\delta}^{\tilde{t},\delta} \right\|_1 \leq \varepsilon_{\text{dec}},
\end{equation}
\normalsize
in which 
\begin{equation}
    \tilde{\tau}_{1,R_TR_X^\delta}^{\tilde{t},\delta} = \text{Tr}_{E_2^\delta C_2^\delta}\left[\left(U_{C_2^\delta} \otimes I_{R_TR_X^\delta E_2^\delta}\right) \tilde{\tau}_{2, R_TR_X^\delta E_2^\delta C_2^\delta}^{\tilde{t},\delta} \left(U_{C_2^\delta}^{\dagger} \otimes I_{R_TR_X^\delta E_2^\delta}\right)\right]
\end{equation} 
and
\begin{equation}
\label{eq:decoupling_error}
    \varepsilon_{\text{dec}} := \frac{d_{C_2^\delta}d_{R_X^\delta R_T}}{d_{\tilde{C}}^2} \text{Tr}\left[\left(\tilde{\tau}_{3,R_TR_X^\delta C_2^\delta}^{\tilde{t},\delta}\right)^2\right].
\end{equation}
We eventually prove that this $\varepsilon_{\text{dec}}$ bound is satisfying enough to ensure that we have an optimal communication cost. 

For now, let us complete the rest of the protocol by defining the Uhlmann isometry $V_{\text{Uhl}}$. We define the state at time $s=1$ as
\begin{equation}
\label{eq:state_at_time_2_cq_channel_simulation}
\ket{\Tilde{\tau}_1^{\tilde{t},\delta}}_{R_TR_X^\delta E_2^\delta C_2^\delta} := \left(I_{R_TR_X^\delta E_2^\delta} \otimes U_{C_2^\delta}\right) \ket{\tilde{\tau}_2^{\tilde{t},\delta}}_{R_TR_X^\delta E_2^\delta C_2^\delta}.
\end{equation}
By the decoupling theorem, assuming we picked a unitary that satisfies the decoupling inequality in \eqref{eq:decoupling_inequality_cq_channel_simulation}, if we trace out systems $\Tilde{C}$ and $E_2^\delta$ from the state~$\ket{\Tilde{\tau}_1^{\tilde{t},\delta}}_{R_TR_X^\delta E_2^\delta \tilde{C}\tilde{B}}$, we obtain up to an error~$\varepsilon_{\text{dec}}$, the state~$\frac{I_{\tilde{B}}}{d_{\tilde{B}}} \otimes \tilde{\tau}_{1,R_TR_X^\delta}^{\tilde{t},\delta}$. This turns out to be the same reduced state as if we traced out systems $E_0^\delta$, $C_0^\delta$ and $\tilde{A}$ from the input state~$\ket{\Tilde{\tau}_0^{\tilde{t},\delta}}_{R_TR_X^\delta C_0^\delta E_0^\delta} \otimes \ket{\Phi}_{\tilde{A}\tilde{B}}$. Therefore, by Uhlmann's theorem~\eqref{eq:uhlmann-theorem}, there exists an isometry $V_{\text{Uhl}, E_2^{\delta} \tilde{C} \rightarrow C_0^{\delta} E_0^{\delta}\tilde{A}}$ such that
\begin{equation}
\label{eq:uhlmann_isometry_cq_channel_simulation}
\left\| V_{\text{Uhl}}\ket{\Tilde{\tau}_1^{\tilde{t},\delta}}_{R_TR_X^\delta E_2^\delta \tilde{C}\tilde{B}} -  \ket{\Tilde{\tau}_0^{\tilde{t},\delta}}_{R_TR_X^\delta E_0^\delta C_0^\delta} \otimes \ket{\Phi}_{\tilde{A}\tilde{B}}\right\|_1 \leq 2\sqrt{\varepsilon_{\text{dec}}}.
\end{equation}

Therefore, combining equations \eqref{eq:state_at_time_2_cq_channel_simulation} and \eqref{eq:uhlmann_isometry_cq_channel_simulation} with the triangle inequality and the invariance of the trace norm under isometric transformation, we obtain that the final state at time $s=2$ of our protocol, when the input state is the projected de Finetti state $\ket{\Tilde{\tau}_0^{\tilde{t},\delta}}_{R_TR_X^\delta E_0^\delta C_0^\delta}$, satisfies
\begin{equation}
\label{eq:final_state_de_finetti_input_cq_channel_simulation}
\left\| \ket{\tilde{\tau}_2^{\tilde{t},\delta}}_{R_TR_X^\delta E_2^\delta C_2^\delta} - U_{C_2^\delta}^\dagger V^{\dagger}_\text{Uhl} \ket{\Tilde{\tau}_0^{\tilde{t},\delta}}_{R_TR_X^\delta E_0^\delta C_0^\delta} \otimes \ket{\Phi}_{\tilde{A}\tilde{B}} \right\|_1 \leq 2\sqrt{\varepsilon_{\text{dec}}}.
\end{equation}
From the triangle inequality again, we have that the precision of the protocol on our type-constrained de Finetti state is
\begin{align}
    &\left\|\ket{\tau_2^{\tilde{t},\delta}}_{R_TR_X^nE_2^nC_2^n} - U_{C\delta}^\dagger V_{\text{Uhl}}^\dagger \ket{\tau_0^{\tilde{t},\delta}}_{R_TR_X^nE_0^nC_0^n} \otimes \ket{\Phi}_{\tilde{A}\tilde{B}} \right\|_1\\
    &\leq \left\| \ket{\tau_2^{\tilde{t},\delta}}_{R_TR_X^nE_2^nC_2^n} - \ket{\tilde{\tau}_2^{\tilde{t},\delta}}_{R_TR_X^\delta E_2^\delta C_2^\delta} \right\|_1\\
    &+ \left\| \ket{\tilde{\tau}_2^{\tilde{t},\delta}}_{R_TR_X^\delta E_2^\delta C_2^\delta} - U_{C^\delta}^\dagger V_\text{Uhl}^{\dagger} \ket{\Tilde{\tau}_0^{\tilde{t},\delta}}_{R_TR_X^\delta E_2^\delta C_2^\delta} \otimes \ket{\Phi}_{\tilde{A}\tilde{B}} \right\|_1\\
    &+ \left\| U_{C^\delta}^\dagger V_\text{Uhl}^{\dagger} \left( \ket{\Tilde{\tau}_0^{\tilde{t},\delta}}_{R_TR_X^nE_0^nC_0^n} - \ket{\tau_0^{\tilde{t},\delta}}_{R_TR_X^nE_0^nC_0^n} \right) \otimes \ket{\Phi}_{\tilde{A}\tilde{B}} \right\|_1\\
    &\leq \varepsilon_{\delta} + 2\sqrt{\varepsilon_{\text{dec}}} + \varepsilon_{\delta},\\
    &\leq 2\varepsilon_{\delta} + 2\sqrt{\varepsilon_{\text{dec}}},
\end{align}
in which we used the invariance of the trace distance under isometries in the last inequality to obtain $\varepsilon_{\delta}$ (which is the precision of the approximation in \eqref{eq:precision_on_state_splitting_projections_at_time_0_and_2}).

We know that $\varepsilon_{\delta}$ vanishes exponentially fast in $n$ for a fixed $\delta > 0$, we now show that we can bound the value of $\varepsilon_\text{dec}$ in \eqref{eq:decoupling_error} in such a way that allows to obtain an optimal communication cost for our protocol while insuring this term also vanishes exponentially fast in $n$. By $\delta$-typicality, we know the dimensions $d_{R_X^\delta}$ and $d_{C_2^\delta}$ are at most the sum over all dimensions of the respective $\delta$-typical subspaces for every $\delta$-typical type. This gives us the following upper bound on the dimensions $d_{R_X^\delta}$ and $d_{C_2^\delta}$:
\begin{align}
    d_{R_X^\delta} &\leq \left|\mathcal{P}^{\tilde{t},\delta}_{\mathbb{X}^n}\right| \exp{n\left(\max_{t \in \mathcal{P}_{\mathbb{X}^n}^{\tilde{t},\delta}} H(R_X)_t + \tilde{\omega}_{1,1}(\delta, |R_X|)\right)},\\ 
    d_{C_2^\delta} &\leq \left|\mathcal{P}^{\tilde{t},\delta}_{\mathbb{X}^n}\right| \exp{n\left(\max_{t \in \mathcal{P}_{\mathbb{X}^n}^{\tilde{t},\delta}} H(C_2)_t + \tilde{\omega}_{1,1}(\delta, |C_2|)\right)},
\end{align}
in which $H(R_X)_t$ and $H(C_2)_t$ are the von Neumann entropies of the states ${\sigma'}_{2,t,R_X}$ and ${\sigma'}_{2,t,C_2}$ respectively, and the maximums are taken over all types $t$ in the set $\mathcal{P}^{\tilde{t},\delta}_{\mathbb{X}^n}$. We can bound the size of the set of $\delta$-typical types by the size of the set of all types on $\mathbb{X}^n$, which is in turn bounded by $(n+1)^{|X|} = (n+1)^{|R_X|}$. This gives
\begin{align}
    d_{R_X^\delta} &\leq \exp{n\left(\max_{t \in \mathcal{P}_{\mathbb{X}^n}^{\tilde{t},\delta}} H(R_X)_t + \tilde{\omega}_{1,2}(\delta, |R_X|)\right)},\\
    d_{C_2^\delta} &\leq \exp{n\left(\max_{t \in \mathcal{P}_{\mathbb{X}^n}^{\tilde{t},\delta}} H(C_2)_t + \tilde{\omega}_{1,2}(\delta, |C_2||R_X|)\right)},
\end{align}
We also know that $d_{R_X^\delta R_T} = d_{R_X^\delta}d_{R_T}$ and that $d_{R_T}$ is the number of $\delta$-typical types $\left|\mathcal{P}^{\tilde{t},\delta}_{\mathbb{X}^n}\right|$, which is upper bounded by $|\mathcal{P}_{\mathbb{X}^n}| \leq (n + 1)^{|\mathbb{X}|} = (n+1)^{|R_X|}$.

We are now left to compute the purity $\text{Tr}\left[\left(\tilde{\tau}_{2,R_TR_X^\delta C_2^\delta}^{\tilde{t},\delta}\right)^2\right]$ in order to bound $\varepsilon_{\text{dec}}$. The following lemma gives an upper bound on such purity for general $\delta$-typical projections on various systems.
\begin{lemma}
    \label{lemma:purityBoundDeltaTypicalDeFinetti}
    Let $\delta > 0$, $n,k \in \mathbb{N}_{>0}$ and $\tilde{t}$ be a fixed type in $\mathcal{P}_{\mathbb{X}^n}$ where $\mathbb{X}$ is a finite alphabet. Let $\tau^{\tilde{t},\delta}_{R_TA_1^n \ldots A_k^n}$ be a pure $\delta$-typical de Finetti state of $\tilde{t}$ on systems $A_1^n, A_2^n, \ldots, A_k^n$ with reference system $R_T$ purifying the classical type system:
    \begin{equation}
        \ket{\tau^{\tilde{t},\delta}}_{R_TA_1^n \ldots A_k^n} = \frac{1}{\sqrt{\left|\mathcal{P}_{\mathbb{X}^n}^{\tilde{t},\delta}\right|}} \sum_{t \in \mathcal{P}_{\mathbb{X}^n}^{\tilde{t},\delta}} \ket{\sigma_t}_{A_1\ldots A_k}^{\otimes n} \otimes \ket{t}_{R_T}.
    \end{equation}
    Define $\ket{\tilde{\tau}^{\tilde{t},\delta}}_{R_TA_1^n \ldots A_k^n}$ to be the \textbf{renormalized} projected de Finetti state on the $\delta$-typical subspace of each system $A_i$:
    \begin{equation}
        \label{eq:projected_de_finetti_state_general_purity}
        \ket{\tilde{\tau}^{\tilde{t},\delta}} \propto \left(\sum_{t \in \mathcal{P}_{\mathbb{X}^n}^{\tilde{t},\delta}} \left(\bigotimes_{i=1}^k \Pi_{A_i^n}^{t,\delta} \otimes \ketbra{t}{t}_{R_T}\right)\right) \ket{\tau^{\tilde{t},\delta}},
    \end{equation}
    in which $\Pi_{A_i^n}^{t,\delta}$ is the $\delta$-typical projector of the state $\sigma_t$ on system $A_i$. Then, for any $i' \in [k]$, we have the following upper bound on the purity of the reduced state $\tilde{\tau}^{\tilde{t},\delta}_{A_{i'}^n}$:
    \begin{align}
    \text{Tr}\left[\left(\tilde{\tau}_{A_{i'}^n}^{\tilde{t},\delta}\right)^2\right] &\leq \exp{-n\left(\min_{t \in \mathcal{P}_{\mathbb{X}^n}^{\tilde{t},\delta}} H(A_{i'})_t - \tilde{\omega}_{1,1}(\delta,|A_{i'}|) - \frac{f(\delta)}{n}\right)},
    \end{align}
    in which the minimum is over all types $t$ in the set $\mathcal{P}_{\mathbb{X}^n}^{\tilde{t},\delta}$, $H(A_{i'})_t$ is the von Neumann entropy of the reduced state $\sigma_{t,A_{i'}}$, with $\tilde{\omega}_{1,1}(\delta,|A_{i'}|)$ defined in \eqref{eq:tilde-omega} and $f(\delta)$ is defined as
    \begin{equation}
    f(\delta) := -2\log(1 - \frac{\left(\varepsilon_{\delta}\right)^2}{2}) + 1,
    \end{equation}
    with $\varepsilon_{\delta}$ being the precision of the $\delta$-typical approximation from the projection in \eqref{eq:projected_de_finetti_state_general_purity} as given by Lemma~\ref{lemma:deltaTypicalTraceDistanceBound}.
\end{lemma}
\noindent We prove this lemma in Appendix~\ref{appendix:purityBoundTypicalDeFinetti}. From this lemma, we can deduce an upper bound on the purity of the state $\tilde{\tau}_{2,R_TR_X^\delta C_2^\delta}^{\tilde{t},\delta}$. Note that since the full state is pure on all the systems $R_TR_X^\delta E_2^\delta C_2^\delta$, the purity of the reduced state on systems $R_TR_X^\delta C_2^\delta$ is the same as the purity of the reduced state on system $E_2^\delta$. Therefore, using Lemma~\ref{lemma:purityBoundDeltaTypicalDeFinetti} with $A_{i'} = E_2$ gives
\begin{align}
\text{Tr}\left[\left(\tilde{\tau}_{2,R_TR_X^\delta C_2^\delta}^{\tilde{t},\delta}\right)^2\right] \leq \exp{-n\left(\min_{t\in\mathcal{P}_{\mathbb{X}^n}^{\tilde{t},\delta}} H(E_2)_t - \tilde{\omega}_{1,1}(\delta,|E_2|) - \frac{f(\delta)}{n}\right)},
\end{align}
in which $H(E_2)_t$ is the von Neumann entropy of the state ${\sigma'}_{2,t,E_2}$. Since the states $\ket{{\sigma'}_t}_{R_XC_2E_2}$ are pure, we have that $H(E_2)_t = H(R_XC_2)_t$. Therefore, we get
\begin{align}
\text{Tr}\left[\left(\tilde{\tau}_{2,R_TR_X^\delta C_2^\delta}^{\tilde{t},\delta}\right)^2\right]] \leq \exp{-n\left(\min_{t\in \mathcal{P}_{\mathbb{X}^n}^{\tilde{t},\delta}} H(R_XC_2)_t - \tilde{\omega}_{1,1}(\delta,|E_2|) - \frac{f(\delta)}{n}\right)}.
\end{align}

Thus, since $d_{R} = d_{R_T} d_{R_X^\delta}$, we get the following upper bound on $\varepsilon_{\text{dec}}$:
\begin{align}
    \varepsilon_{\text{dec}} &\leq \frac{d_{C_2^\delta}d_{R_X^\delta}d_{R_T}}{d_{\tilde{C}}^2} \text{Tr}\left[\left(\tilde{\tau}_{3,R_TR_X^\delta C_2^\delta}^{\tilde{t},\delta}\right)^2\right]\\
    \label{eq:worst_case_entropic_quantities_state_splitting}
    &\leq \frac{\exp{n\left(\max_{t}H(C_2)_{t} + \max_{t}H(R_X)_{t} - \min_{t} H(E_2)_t + \tilde{\omega}_{3,6}(\delta,|R_XC_2E_2|) + \frac{f(\delta)}{n}\right)}}{d_{\tilde{C}}^2},
\end{align}
in which we used the crude upper bound in \eqref{eq:crude_tilde_omega_additivity}. We keep the dimension of the entanglement system implicit for now and develop it in Theorem \ref{thm:prior_free_cq_channel_simulation}. This concludes the proof of Proposition \ref{prop:de_finetti_state_state_splitting}. \begin{flushright}
    $\square$
\end{flushright}

To prove Theorem \ref{thm:prior_free_cq_channel_simulation}, we apply Proposition \ref{prop:de_finetti_state_state_splitting} by taking into account the monotonicity of the trace distance between the various ${\sigma}_{t}$ states in the definition of the input de Finetti state in \eqref{eq:delta_typical_de_finetti_input_state_state_splitting} so that we can guarantee that the trace distance is preserved between the ${\sigma'}_t$ states of the output de Finetti state in \eqref{eq:de_finetti_output_state_state_splitting}. This allows us to use the continuity of entropy to turn the worst case entropic quantities in \eqref{eq:worst_case_entropic_quantities_state_splitting} into entropic quantities that depend on the estimated type $\tilde{t}$. Applying continuity of entropy a second time allows us to transform these entropic quantities into those of the type $t$ of the input sequence $x^n \in X^n$.

\textit{Proof of Theorem \ref{thm:prior_free_cq_channel_simulation}}
Since~$\frac{1}{2}\left\| \sigma_{t,X} - \sigma_{\tilde{t},X} \right\|_1 \leq \delta$ holds for every type $t$ in the set $\mathcal{P}^{\tilde{t},\delta}_{\mathbb{X}^n}$, then by the monotonicity of the trace norm under the application of the channel $\mathcal{N}$ we also have~$\frac{1}{2}\left\| {\sigma'}_{2,t,C} - {\sigma'}_{2,\tilde{t},C} \right\|_1 \leq \delta$, which allows us to use the continuity of entropy to express the entropy in terms of the type $\tilde{t}$. Applying it a second time allows us to express these entropic quantities in terms of the type $t$ of the input $x^n$. Applying this idea with the inequality of \eqref{eq:crude_tilde_omega_additivity} gives us
\begin{align}
    d_{R_X^\delta} &\leq \exp{n\left(H(R_X)_{t} + \tilde{\omega}_{3,2}(\delta, |R_X|)\right)},\\
    d_{C_2^\delta} &\leq \exp{n\left(H(C)_{t} + \tilde{\omega}_{3,2}(\delta, |C||R_X|)\right)}.
\end{align}
Similarly, we have the following for the purity bound
\begin{align}  \text{Tr}\left[\left(\tilde{\tau}_{2,R_TR_X^\delta C^\delta}^{\tilde{t},\delta}\right)^2\right] \leq \exp{-n\left(H(R_XC)_{t} - \tilde{\omega}_{3,1}(\delta,|E|) - \frac{f(\delta)}{n}\right)}.
\end{align}
Combining all these bounds together gives
\begin{align}
    \varepsilon_{\text{dec}} &\leq \frac{d_{C^\delta}d_{R_X^\delta}d_{R_T}}{d_{\tilde{C}}^2} \text{Tr}\left[\left(\tilde{\tau}_{3,R_TR_X^\delta C^\delta}^{\tilde{t},\delta}\right)^2\right]\\
    &\leq \frac{\exp{n\left(H(C)_{t} + H(R_X)_{t} - H(R_XC)_{t} + \tilde{\omega}_{9,6}(\delta,|R_XCE|) + \frac{f(\delta)}{n}\right)}}{d_{\tilde{C}}^2}\\
    &\leq \frac{\exp{n\left(I(R_X;C)_{t} + \tilde{\omega}_{9,6}(\delta,|R_XCE|) + \frac{f(\delta)}{n}\right)}}{d_{\tilde{C}}^2},
\end{align}
in which we used the crude upper bound in \eqref{eq:crude_tilde_omega_additivity}. 

Therefore, to ensure that $\varepsilon_{\text{dec}}$ asymptotically goes to $0$ as $n$ goes to infinity, it is sufficient to choose the dimension of the communicated system $\tilde{C}$ such that
\begin{equation}
\log d_{\tilde{C}} \gg \frac{n}{2}\left(I(R_X;C)_{\tilde{t}} + \tilde{\omega}_{6,6}(\delta,|R_XCE|) + \frac{f(\delta)}{n}\right).
\end{equation}
which has an asymptotic communication cost of $\frac{n}{2}I(R_X;C)$ which is optimal for this task. To calculate the cost in pre-shared entanglement, since $U_{\text{dec}}$ is a unitary, the dimensions of $C^\delta$ and $\tilde{B}\tilde{C}$ are the same. Therefore, we have
\begin{equation}
\log d_{\tilde{B}} = \log d_{C^\delta} - \log d_{\tilde{C}},
\end{equation}
which gives an entanglement cost that scales as
\begin{equation}
\log d_{\tilde{B}}~\approx \frac{n}{2}\left(H(C) - I(R_X;C)\right) = \frac{n}{2}I(C;X).
\end{equation}
Therefore, we define the following protocol that simulates $V_{\mathcal{N}, X \rightarrow EC}^{\otimes n}$ and the transmission of the system $C^n$ from a sender to a receiver.
\begin{tcolorbox}[boxrule=0.6pt,colback=white,breakable,sharp corners=all]
\begin{protocol}
\label{protocol:cq_channel_simulation_one_way}
Let $\delta > 0$ and $n \in \mathbb{N}_{> 0}$. Let $\mathcal{N}_{X \rightarrow C}$ be a classical-quantum channel and $V_{\mathcal{N}, X \rightarrow EC}$ be an isometric extension of that channel on a purification system $E$. The protocol works as follows:
\begin{enumerate}
    \item The sender Alice sends the receiver Bob her input's type using at most $|X|\log(n + 1)$ bits of communication, in which $n$ is the length of her input sequence and $X$ is the set of symbols that can appear in her input.
    \item Alice applies the isometric extension $V_{\mathcal{N}}^{\otimes n}$ of the channel $\mathcal{N}^{\otimes n}$ on her input sequence $x^n$.
    \item Alice and Bob apply the state splitting protocol of the $\delta$-typical de Finetti state associated to the type $\tilde{t} = t$, in which $t$ is the type of the input $x^n$. \footnote{Note here that technically, we could have split the definition of $\delta$ into som~$\delta_{\text{Finetti}}$ for the de Finetti reduction and~$\delta_{\text{approx}}$ for the typical projections. In this case, this would have allowed us to choose $\delta_{\text{Finetti}} = 0$ since Alice and Bob know the exact type of the input sequence. However, this section is just a warm-up for the one with side-information at the receiver, in which we won't be able to assume that $\delta_{\text{Finetti}} = 0$ and it then becomes much simpler to presume that~$\delta_{\text{approx}} = \delta_{\text{Finetti}} = \delta$.}
\end{enumerate}
\end{protocol}
\end{tcolorbox}

Now, using our de Finetti reduction, we can show that Protocol \ref{protocol:cq_channel_simulation_one_way} successfully simulates the action of the channel $V_{\mathcal{N}}^{\otimes n}$ on any $\delta$-typical classical input. Indeed, Let $\Pi_{\text{sim}}$ be our simulation protocol defined as $\Pi_{\text{sim}} := \Pi_{ss}^{\tilde{t},\delta}V_{\mathcal{N}}^{\otimes n}$, in which $\Pi_{ss}^{\tilde{t},\delta}$ is the state splitting protocol of the $\delta$-typical de Finetti state associated with the estimated type $\tilde{t}$ from Proposition \ref{prop:de_finetti_state_state_splitting}. Since the protocol $\Pi_{ss}^{\tilde{t},\delta}$ is defined on the $\delta$-typical subspace, it can be extended to the whole space by appending the identity on the subspace orthogonal to it. This makes $\Pi_{ss}^{\otimes n}$ permutation-covariant since the $\delta$-typical subspace is invariant under permutation, and therefore a permutation only acts on the identity on the orthogonal space which is permutation-covariant. We know that $V_{\mathcal{N}}^{\otimes n}$ is also permutation-covariant with respect to the action of any permutation $\pi$ on the inputs, since the inverse permutation on the $n$ output quantum systems negates its effect. Therefore, $\Pi_{\text{sim}}$ is also permutation-covariant and so is $\Delta := \Pi_{\text{sim}} - V_{\mathcal{N}}^{\otimes n}$. Then, by our de Finetti reduction in the form of Theorem \ref{thm:restrictedClassicalQuantumDeFinettiReduction}, we have that for every $x^n \in X^n$,
\begin{align}
\label{eq:de_finetti_reduction_cq_channel_simulation}
\left\|\left( \Pi_{\text{sim}} - V_{\mathcal{N}}^{\otimes n} \right) \left(\ketbra{x^n}{x^n}_{X^n}\right)\right\|_1 &\leq \left\|\left( \Pi_{\text{sim}} - V_{\mathcal{N}}^{\otimes n} \right) \left(\ketbra{x^n}{x^n}_{X^n} \otimes \ketbra{x^n}{x^n}_{R_X^n}\right)\right\|_1\\
&\leq \left\|\left( \Pi_{\text{sim}} - V_{\mathcal{N}}^{\otimes n}\right)\right\|_{\diamond,cl}^{\tilde{t},\delta}\\
&\leq\left|\mathcal{P}_{\mathbb{X}^n}\right|\left|\mathcal{P}_{\mathbb{X}^n}^{\tilde{t},\delta}\right|\left\| \left(\Delta \otimes \textbf{id}_{R}\right) \left(\ketbra{\tau^{\tilde{t},\delta}}{\tau^{\tilde{t},\delta}}_{X^n R}\right) \right\|_1\\
&\leq \left|\mathcal{P}_{\mathbb{X}^n}\right|\left|\mathcal{P}_{\mathbb{X}^n}^{\tilde{t},\delta}\right| \left(2\varepsilon_{\delta} + 2\sqrt{\varepsilon_{\text{dec}}} \right)\\
&\leq (n+1)^{2|X|} \left(2\varepsilon_{\delta} + 2\sqrt{\varepsilon_{\text{dec}}} \right).
\end{align}
We note that $|X| = |R_X|$ and that for a probability distribution $t$ passed through a classical-quantum channel, we have $I(R_X;C)_t = I(X;C)_t$ (when viewed as a classical-quantum state). This concludes the proof of Theorem \ref{thm:prior_free_cq_channel_simulation}.
\begin{flushright}
    $\square$
\end{flushright}

\subsection{State Redistribution with Side information at the Receiver}
\label{subsection:StateRedistributionWithSideInformation}
We now turn to a similar setting of prior-free classical-quantum channel simulation, but in which the receiver Bob also owns a classical input. We prove the following theorem similar to Theorem \ref{thm:prior_free_cq_channel_simulation} but with side information at the receiver.
\begin{theorem}
\label{thm:prior_free_cq_channel_simulation_with_side_info}
    Let $\mathcal{N}_{X \to C}$ be a classical-quantum channel sending classical states on system $X$ to quantum state on system $C$. Let $V_{\mathcal{N},X \to EC}$ be an isometric extension of $\mathcal{N}$, and let $\varepsilon > 0$, $n \in \mathbb{N}_{>0}$ and $\delta > 0$. Let $Y$ be a system owned by a receiver. Then, there exists a variable length prior-free simulation protocol $\Pi_{\text{sim}}$ of $V_{\mathcal{N}}^{\otimes n}$ such that for every input sequence $x^ny^n \in \left(X \times Y\right)^n$, the trace distance error $\varepsilon_{x^ny^n}$ of the simulation
    \begin{equation}
        \varepsilon_{x^ny^n} := \left\|\left(\Pi_{\text{sim}} - V_{\mathcal{N}}^{\otimes n}\right)\left(\ketbra{x^ny^n}{x^ny^n}_{X^nY^n}\right)\right\|_1
    \end{equation}
    is bounded by $\varepsilon$, and the dimension of the communicated system during the protocol is
    \begin{equation}
        d_{\tilde{C}} = \frac{2^{10}\cdot\exp{\frac{n}{2}\left(I(C;X|Y)_{t} + \frac{4}{3}\tilde{\omega}_{9,6}(\delta,|Y|^2|X||C||E|) + \frac{2f_{\text{max}}(\delta)}{n}\right) - \Delta^t_{\tilde{t},\delta}}}{{\left(\varepsilon_{\text{red}} - 2\hat{\varepsilon}_{\delta} - 2\overline{\varepsilon}_{\delta}\right)^4}}.
    \end{equation}
    in which
    \begin{align}
        \hat{\varepsilon}_{\delta} \leq 2^{-n\frac{\delta^2}{4\ln(2)} + 1}\sqrt{(n+1)^{|Y|^2|X|} + (n+1)^{|C|} + (n+1)^{|E|}},\\
        \overline{\varepsilon}_{\delta} \leq 2^{-n\frac{\delta^2}{4\ln(2)} + 1}\sqrt{(n+1)^{|E||X||Y|} + (n+1)^{|C|} + (n+1)^{|Y|}}.
    \end{align}
    Here, the conditional mutual information $I(C;X|Y)_{t_{x^ny^n}}$ is evaluated between the distribution $t_{x^ny^n}$ of the type of the input $x^ny^n$ on system $XY$ and the quantum output of $t_{x^n}$ through $\mathcal{N}$ on system $C$ conditional to the distribution $t_{y^n}$ on system $Y$, in the form of a classical-state with classical system $XY$ and quantum system $C$. The variable $d_\text{comm}$ is the dimension of the system that needs to be sent from the sender to the receiver, $\tilde{\omega}_{k,m}(\delta, d)$ is defined in \eqref{eq:tilde-omega} and $f_{\text{max}}(\delta)$ is defined as
    \begin{equation}
        f_{\text{max}}(\delta) := -2\log(1 - \frac{\left(\varepsilon_{\delta,\text{max}}\right)^2}{2}),
    \end{equation}
    in which $\varepsilon_{\delta,\text{max}}$ is defined as
    \begin{equation}
    \varepsilon_{\delta,\text{max}} := 2^{-n\frac{\delta^2}{4\ln(2)} + 1}\sqrt{2(n+1)^{|Y|^2|X|} + 2(n+1)^{|E||X||Y|} + (n+1)^{|C|}}.
    \end{equation}
    The term $\Delta^t_{\tilde{t},\delta}$ is bounded as
    \begin{equation}
        \Delta^t_{\tilde{t},\delta} \leq \tilde{\omega}_{4,3}(\delta,|Y|^2|X||C||E|)
    \end{equation}
    
    Furthermore, the protocol uses pairs of maximally entangled qubits at a rate of $\frac{n}{2}I(C;Y)_{t}$ and generate pairs of maximally entangled qubits at a rate of $\frac{n}{2}I(C;E)_{t}$. In particular, taking $\delta = n^{-1/2 + \alpha}$ for some $\alpha \in (0, 1/2)$ ensures that the error $\varepsilon_{x^ny^n}$ vanishes for every $x^ny^n \in (X \times Y)^n$ as $n$ approaches infinity.
\end{theorem}
Similar to the case with no side information at the receiver, we achieve this result by applying the state redistribution protocol on the type-constrained de Finetti state on which we have applied the isometric extension $V_{\mathcal{N}, X \rightarrow EC}^{\otimes n}$. Since we are now working with an additional system $Y$ on Bob's side, the input de Finetti state is
\begin{equation}
    \label{eq:input_de_finetti_state_redistribution}
    \tau_{X^nY^n}^{\tilde{t},\delta} := \frac{1}{\left|\mathcal{P}_{\mathbb{X}^n \times \mathbb{Y}^n}^{\tilde{t},\delta}\right|}\sum_{t \in \mathcal{P}_{\mathbb{X}^n \times \mathbb{Y}^n}^{\tilde{t},\delta}}\sigma_{t,XY}^{\otimes n},
\end{equation}
in which $\sigma_{t,XY}$ is the state associated to the joint type $t$ of the alphabet $\mathbb{X} \times \mathbb{Y}$, and $\mathcal{P}_{\mathbb{X}^n \times \mathbb{Y}^n}^{\tilde{t},\delta}$ is the set of joint types $t$ that are $\delta$-close to the joint type $\tilde{t}$. We can define a purification $\ket{\tau^{\tilde{t},\delta}}_{R_TR_{XY}^nX^nY^n}$ of this de Finetti state on the reference system $R = R_{XY}^nR_T$, in which $R_T$ purifies the type indexation of the de Finetti state
\begin{equation}
    \label{eq:purified_input_de_finetti_state_redistribution}
    \ket{\tau^{\tilde{t},\delta}}_{R_TR_{XY}^n X^n Y^n} := \frac{1}{\sqrt{\left|\mathcal{P}_{\mathbb{X}^n \times \mathbb{Y}^n}^{\tilde{t},\delta}\right|}}\sum_{t \in \mathcal{P}_{\mathbb{X}^n \times \mathbb{Y}^n}^{\tilde{t},\delta}} \ket{\sigma_t}_{R_{XY}XY}^{\otimes n} \otimes \ket{t}_{R_T},
\end{equation}
in which $\ket{\sigma_t}_{R_{XY}XY}$ is a purification of the state $\sigma_{t,XY}$ on the reference system $R_{XY}$ of same dimension as~$XY$. After applying the channel $V_{\mathcal{N},X \rightarrow EC}^{\otimes n}$, the state is labeled
\begin{equation}
    \label{eq:output_state_redistribution}
    \ket{\tau^{\tilde{t},\delta}}_{R_TR_{XY}^n E^n C^n Y^n} := \frac{1}{\sqrt{\left|\mathcal{P}_{\mathbb{X}^n \times \mathbb{Y}^n}^{\tilde{t},\delta}\right|}} \sum_{t \in \mathcal{P}_{\mathbb{X}^n \times \mathbb{Y}^n}^{\tilde{t},\delta}} \ket{{\sigma'}_t}_{R_{XY}ECY}^{\otimes n} \otimes \ket{t}_{R_T},
\end{equation}
with $\ket{{\sigma'}_t}_{R_{XY}ECY} := V_{\mathcal{N},X \rightarrow EC}\ket{\sigma_t}_{R_{XY}XY}$. We are left to apply the state redistribution protocol to transfer the system $C^n$ from a sender to a receiver. We get the following proposition for the state redistribution of the de Finetti state in $\eqref{eq:output_state_redistribution}$, which is used as an intermediate result to prove Theorem \ref{thm:prior_free_cq_channel_simulation_with_side_info}.
\begin{proposition}
    \label{prop:de_finetti_state_redistribution}
    Let $\ket{\tau^{\tilde{t},\delta}}_{R_TR_{XY}^nE^nC^nY^n}$ be a purified de Finetti state with side information on system~$Y$ as in \eqref{eq:output_state_redistribution}. Using the state redistribution protocol, the system $C^n$ can be transferred from a sender to a receiver with total trace distance error $\varepsilon_{\text{red}}$ bounded as
    \begin{equation}
        \varepsilon_{\text{red}} \leq 2\tilde{\varepsilon}_{\delta} + 2\overline{\varepsilon}_{\delta} + 2\sqrt{\tilde{\varepsilon}_{\text{dec}}} + 2\sqrt{\overline{\varepsilon}_{\text{dec}}},
    \end{equation}
    in which $\tilde{\varepsilon}_{\delta}$, $\overline{\varepsilon}_{\delta}$, $\hat{\varepsilon}_{\text{dec}}$ and $\overline{\varepsilon}_{\text{dec}}$ are bounded as
    \begin{align}
        \hat{\varepsilon}_{\text{dec}} \leq \sqrt{4\frac{\exp{n\left(\tilde{H} + \tilde{\omega}_{3,6}(\delta,|Y R_{XY}C E|) + \frac{f(\delta)}{n}\right)}}{d_{\tilde{A}_{\text{out}}}^2d_{\tilde{C}}^2}},\\
        \overline{\varepsilon}_{\text{dec}} \leq \sqrt{4\frac{\exp{n\left(\overline{H} + \tilde{\omega}_{3,6}(\delta,|Y R_{XY}CE|) + \frac{g(\delta)}{n}\right)}}{d_{\tilde{B}_{\text{in}}}^2d_{\tilde{C}}^2}},\\
        \hat{\varepsilon}_{\delta} \leq 2^{-n\frac{\delta^2}{4\ln(2)} + 1}\sqrt{(n+1)^{|Y|^2|X|} + (n+1)^{|C|} + (n+1)^{|E|}},\\
        \overline{\varepsilon}_{\delta} \leq 2^{-n\frac{\delta^2}{4\ln(2)} + 1}\sqrt{(n+1)^{|E||X||Y|} + (n+1)^{|C|} + (n+1)^{|Y|}}.
    \end{align}
    with
    \begin{align}
        \tilde{H} := \max_{t \in T} H(C)_{t} + \max_{t \in T} H(YR_{XY})_{t} - \min_{t \in T} H(YR_{XY}C)_{t},\\
        \overline{H} := \max_{t \in T} H(C)_{t} + \max_{t \in T} H(ER_{XY})_{t} - \min_{t \in T} H(ER_{XY}C)_{t}.
    \end{align}
    Here, $\tilde{\omega}_{k,m}(\delta,d)$ is defined in \eqref{eq:tilde-omega}, and $f(\delta)$ and $g(\delta)$ are defined as
    \begin{equation}
        f(\delta) := -2\log(1 - \frac{\left(\hat{\varepsilon}_{\delta}\right)^2}{2}) + 1,
    \end{equation}
    and
    \begin{equation}
        g(\delta) := -2\log(1 - \frac{\left(\overline{\varepsilon}_{\delta}\right)^2}{2}) + 1,
    \end{equation}
    This protocol requires the communication of a system of dimension $d_{\tilde{C}}$ from the sender to the receiver, and $d_{\tilde{B}_{\text{in}}}$ and $d_{\tilde{A}_{\text{out}}}$ are the local dimensions of the shared entanglement systems between the sender and the receiver before and after the protocol respectively.
\end{proposition}
The state redistribution protocol achieves the same tasks as a the state splitting protocol, but takes into account the receiver Bob's side information on an additional system to reduce the communication cost. Mathematically, a state $\tau_{RECY}$ is shared between a sender Alice that contains systems $E$ and $C$, a reference system $R$ and a receiver Bob who holds the system $Y$. The state redistribution protocol redistributes the system $C$ from Alice to Bob. This protocol works by combining a state splitting with a state merging. This can be achieved by combining a state splitting protocol with a state merging protocol \cite{oppenheim2008stateredistributionmergingintroducing}. Informally, a first state splitting from Alice to Bob achieves the task of last section, then Bob sends back the system $C$ to Alice, and a state merging is done from Alice to Bob again, essentially taking into account Bob's side information. However, a simplification allows to remove the need for Bob to send back the system $C$ to Alice in the middle of the protocol, making it efficient.

The state redistribution protocol uses the bidecoupling theorem, which allows to use the same decoupling unitary for the state splitting and the state merging. The proof of the bidecoupling theorem relies directly on the decoupling theorem itself. We use the statement of the bidecoupling theorem from Ref.~\cite{Ye_Bai_Wang_2008}, in which it's proof can be found.
\begin{theorem}[Bidecoupling Theorem \cite{Ye_Bai_Wang_2008}]
\label{thm:bidecoupling}
    Let $\psi_{CA}$ and $\phi_{CB}$ be two density operators such that $C = C_1C_2\tilde{C}$, then there exists a unitary operator $U$ acting on system $C$ such that
    \begin{equation}
    \left\|\text{Tr}_{C_2\tilde{C}}\left(U\psi_{CA}U^\dagger\right) - \frac{I_{C_1}}{d_{C_1}} \otimes \psi_{A}\right\|_1^2 \leq 2\alpha,
    \end{equation}
    and
    \begin{equation}
    \left\|\text{Tr}_{C_1\tilde{C}}\left(U\phi_{CB}U^\dagger\right) - \frac{I_{C_2}}{d_{C_2}} \otimes \phi_{B}\right\|_1^2 \leq 2\beta,
    \end{equation}
    with
    \begin{equation}
    \alpha := \frac{d_Cd_{A}\text{Tr}\left(\psi_{CA}^2\right)}{d_{C_2\tilde{C}}^2},
    \end{equation}
    and
    \begin{equation}
    \beta := \frac{d_Cd_{B}\text{Tr}\left(\phi_{CB}^2\right)}{d_{C_1\tilde{C}}^2}.
    \end{equation}
\end{theorem} 
\noindent Figure \ref{fig:state_redistribution_task} shows the protocol that does the state redistribution of the type-constrained de Finetti state.
\begin{figure*}[t]
    \centering
    \resizebox{1\textwidth}{!}{
    \begin{circuitikz}
    \tikzstyle{every node}=[font=\fontsize{18.2pt}{23.7pt}\selectfont]
    \node [font=\fontsize{18.2pt}{23.7pt}\selectfont, fill={rgb,255:red,255; green,255; blue,255}, fill opacity=1, text opacity=1, inner xsep=0.080cm, inner ysep=0.085cm, rounded corners=0.020cm] at (13.75,13.5) {$Y_0^n$};
    \node [font=\fontsize{18.2pt}{23.7pt}\selectfont, fill={rgb,255:red,255; green,255; blue,255}, fill opacity=1, text opacity=1, inner xsep=0.080cm, inner ysep=0.085cm, rounded corners=0.020cm] at (13.75,13.625) {$\tilde{B}_{\text{in}}$};
    \node [font=\fontsize{18.2pt}{23.7pt}\selectfont, fill={rgb,255:red,255; green,255; blue,255}, fill opacity=1, text opacity=1, inner xsep=0.080cm, inner ysep=0.085cm, rounded corners=0.020cm] at (24.75,17) {$\tilde{B}_{\text{in}}\tilde{C}$};
    \node [font=\fontsize{18.2pt}{23.7pt}\selectfont, fill={rgb,255:red,255; green,255; blue,255}, fill opacity=1, text opacity=1, inner xsep=0.080cm, inner ysep=0.085cm, rounded corners=0.020cm] at (28.75,14.625) {$\tilde{B}_{\text{out}}$};
    \node [font=\fontsize{18.2pt}{23.7pt}\selectfont, fill={rgb,255:red,255; green,255; blue,255}, fill opacity=1, text opacity=1, inner xsep=0.080cm, inner ysep=0.085cm, rounded corners=0.020cm] at (9.375,17.75) {$\tau_{\text{out}}$};
    \draw [ color={rgb,255:red,169; green,186; blue,255}, draw opacity=1, line width=0.5pt, short] (28.75,15.625) .. controls (21.25,15.625) and (21.25,15.625) .. (13.75,15.625);
    \draw [ fill={rgb,255:red,218; green,217; blue,255}, fill opacity=1] (14.625,19.375) rectangle (16.75,16.25);
    \node [font=\fontsize{18.2pt}{23.7pt}\selectfont, inner xsep=0.080cm, inner ysep=0.085cm, rounded corners=0.020cm] at (15.75,17.875) {$V^{\dagger}_{Uhl}$};
    \draw [ color={rgb,255:red,52; green,82; blue,199}, draw opacity=1, line width=1pt, short] (14.625,18.75) -- (11.25,18.75);
    \draw [ color={rgb,255:red,52; green,82; blue,199}, draw opacity=1, line width=1pt, short] (13.375,21.25) .. controls (11.625,19.5) and (11.75,19.5) .. (10,17.75);
    \draw [ color={rgb,255:red,52; green,82; blue,199}, draw opacity=1, line width=1pt, short] (30,21.25) .. controls (21.625,21.25) and (21.75,21.25) .. (13.375,21.25);
    \draw [ color={rgb,255:red,52; green,82; blue,199}, draw opacity=1, line width=1pt, short] (33.25,17.25) .. controls (31.625,19.25) and (31.625,19.25) .. (30,21.25);
    \draw [ color={rgb,255:red,52; green,82; blue,199}, draw opacity=1, line width=1pt, short] (31.25,18.875) .. controls (24,18.875) and (24,18.875) .. (16.75,18.875);
    \draw [ color={rgb,255:red,52; green,82; blue,199}, draw opacity=1, line width=1pt, short] (17.5,17.5) .. controls (17.125,17.5) and (17.125,17.5) .. (16.75,17.5);
    \draw [ color={rgb,255:red,52; green,82; blue,199}, draw opacity=1, line width=1pt, short] (17.5,14.375) .. controls (17.5,15.875) and (17.5,15.875) .. (17.5,17.5);
    \draw [ color={rgb,255:red,52; green,82; blue,199}, draw opacity=1, line width=1pt, short] (21.625,13.75) .. controls (21,13.75) and (20.5,13.75) .. (20.5,13.75);
    \draw [ color={rgb,255:red,52; green,82; blue,199}, draw opacity=1, line width=1pt, short] (21.625,17.125) .. controls (21.625,15.375) and (21.625,15.375) .. (21.625,13.75);
    \draw [ color={rgb,255:red,169; green,186; blue,255}, draw opacity=1, line width=0.5pt, short] (28.75,19.875) .. controls (21.25,19.875) and (21.25,19.875) .. (13.75,20);
    \node [font=\fontsize{18.2pt}{23.7pt}\selectfont, fill={rgb,255:red,255; green,255; blue,255}, fill opacity=1, text opacity=1, inner xsep=0.080cm, inner ysep=0.085cm, rounded corners=0.020cm] at (13.875,21.75) {$R_{X}^n$};
    \node [font=\fontsize{18.2pt}{23.7pt}\selectfont, fill={rgb,255:red,255; green,255; blue,255}, fill opacity=1, text opacity=1, inner xsep=0.080cm, inner ysep=0.085cm, rounded corners=0.020cm] at (17.625,18) {$\tilde{A}_{\text{out}}\tilde{C}$};
    \node [font=\fontsize{18.2pt}{23.7pt}\selectfont, fill={rgb,255:red,255; green,255; blue,255}, fill opacity=1, text opacity=1, inner xsep=0.080cm, inner ysep=0.085cm, rounded corners=0.020cm] at (17.375,19.375) {$E_2^n$};
    \node [font=\fontsize{18.2pt}{23.7pt}\selectfont, fill={rgb,255:red,255; green,255; blue,255}, fill opacity=1, text opacity=1, inner xsep=0.080cm, inner ysep=0.085cm, rounded corners=0.020cm] at (13.875,19.25) {$E_0^n$};
    \node [font=\fontsize{18.2pt}{23.7pt}\selectfont, fill={rgb,255:red,255; green,255; blue,255}, fill opacity=1, text opacity=1, inner xsep=0.080cm, inner ysep=0.085cm, rounded corners=0.020cm] at (6.25,20.25) {$Reference$};
    \node [font=\fontsize{18.2pt}{23.7pt}\selectfont, fill={rgb,255:red,255; green,255; blue,255}, fill opacity=1, text opacity=1, inner xsep=0.080cm, inner ysep=0.085cm, rounded corners=0.020cm] at (6.875,17.5) {$Alice$};
    \node [font=\fontsize{18.2pt}{23.7pt}\selectfont, fill={rgb,255:red,255; green,255; blue,255}, fill opacity=1, text opacity=1, inner xsep=0.080cm, inner ysep=0.085cm, rounded corners=0.020cm] at (7,14.625) {$Bob$};
    \draw [ color={rgb,255:red,52; green,82; blue,199}, draw opacity=1, line width=1pt, short] (14.625,17.75) -- (10,17.75);
    \draw [ color={rgb,255:red,52; green,82; blue,199}, draw opacity=1, line width=1pt, short] (14.625,16.75) -- (13.375,16.75);
    \node [font=\fontsize{18.2pt}{23.7pt}\selectfont, fill={rgb,255:red,255; green,255; blue,255}, fill opacity=1, text opacity=1, inner xsep=0.080cm, inner ysep=0.085cm, rounded corners=0.020cm] at (13.875,18.25) {$C_0^n$};
    \node [font=\fontsize{18.2pt}{23.7pt}\selectfont, fill={rgb,255:red,255; green,255; blue,255}, fill opacity=1, text opacity=1, inner xsep=0.080cm, inner ysep=0.085cm, rounded corners=0.020cm] at (13.75,17.25) {$\tilde{A}_{\text{in}}$};
    \draw [ fill={rgb,255:red,218; green,217; blue,255}, fill opacity=1] (18.75,15.25) rectangle (20.5,12.75);
    \draw [ color={rgb,255:red,52; green,82; blue,199}, draw opacity=1, line width=1pt, short] (17.5,14.375) .. controls (18.125,14.375) and (18.125,14.375) .. (18.75,14.375);
    \draw [ color={rgb,255:red,52; green,82; blue,199}, draw opacity=1, line width=1pt, short] (13.375,13.125) .. controls (16,13.125) and (16.125,13.125) .. (18.75,13.125);
    \draw [ color={rgb,255:red,52; green,82; blue,199}, draw opacity=1, line width=1pt, short] (12.5,15.375) .. controls (12.875,14.25) and (13,14.25) .. (13.375,13.125);
    \node [font=\fontsize{18.2pt}{23.7pt}\selectfont, fill={rgb,255:red,255; green,255; blue,255}, fill opacity=1, text opacity=1, inner xsep=0.080cm, inner ysep=0.085cm, rounded corners=0.020cm] at (12,15.625) {$\Phi_{\text{in}}$};
    \draw [ color={rgb,255:red,52; green,82; blue,199}, draw opacity=1, line width=1pt, short] (11.25,18.75) .. controls (10.625,18.25) and (10.625,18.25) .. (10,17.75);
    \node [font=\fontsize{18.2pt}{23.7pt}\selectfont, inner xsep=0.080cm, inner ysep=0.085cm, rounded corners=0.020cm] at (19.625,14) {$U_{\text{dec}}^{\dagger}$};
    \draw [ color={rgb,255:red,52; green,82; blue,199}, draw opacity=1, line width=1pt, short] (13.375,16.75) -- (12.5,15.875);
    \node [font=\fontsize{14.2pt}{18.5pt}\selectfont, fill={rgb,255:red,255; green,255; blue,255}, fill opacity=1, text opacity=1, inner xsep=0.080cm, inner ysep=0.085cm, rounded corners=0.020cm] at (13.75,24) {s = 0};
    \draw [-{Stealth[scale=1.5]}, ] (13.75,23.625) -- (13.75,23.125);
    \node [font=\fontsize{14.2pt}{18.5pt}\selectfont, fill={rgb,255:red,255; green,255; blue,255}, fill opacity=1, text opacity=1, inner xsep=0.080cm, inner ysep=0.085cm, rounded corners=0.020cm] at (17.25,24) {s = 1};
    \draw [-{Stealth[scale=1.5]}, ] (17.25,23.625) -- (17.25,23.125);
    \node [font=\fontsize{14.2pt}{18.5pt}\selectfont, fill={rgb,255:red,255; green,255; blue,255}, fill opacity=1, text opacity=1, inner xsep=0.080cm, inner ysep=0.085cm, rounded corners=0.020cm] at (21,24) {s = 2};
    \draw [-{Stealth[scale=1.5]}, ] (21,23.625) -- (21,23.125);
    \node [font=\fontsize{18.2pt}{23.7pt}\selectfont, fill={rgb,255:red,255; green,255; blue,255}, fill opacity=1, text opacity=1, inner xsep=0.080cm, inner ysep=0.085cm, rounded corners=0.020cm] at (21.125,14.25) {$C_2^n$};
    \draw [ fill={rgb,255:red,218; green,217; blue,255}, fill opacity=1] (22.25,18.375) rectangle (24,15.875);
    \node [font=\fontsize{18.2pt}{23.7pt}\selectfont, inner xsep=0.080cm, inner ysep=0.085cm, rounded corners=0.020cm] at (23.125,17.125) {$U_{\text{dec}}$};
    \draw [ color={rgb,255:red,52; green,82; blue,199}, draw opacity=1, line width=1pt, short] (22.25,17.125) .. controls (21.875,17.125) and (22,17.125) .. (21.625,17.125);
    \draw [ color={rgb,255:red,52; green,82; blue,199}, draw opacity=1, line width=1pt, short] (29.625,17.625) .. controls (26.75,17.625) and (26.875,17.625) .. (24,17.625);
    \draw [ color={rgb,255:red,52; green,82; blue,199}, draw opacity=1, line width=1pt, short] (24,16.5) .. controls (24.625,16.5) and (24.625,16.5) .. (25.25,16.5);
    \draw [ fill={rgb,255:red,218; green,217; blue,255}, fill opacity=1] (26,14.75) rectangle (28.125,11.625);
    \node [font=\fontsize{18.2pt}{23.7pt}\selectfont, inner xsep=0.080cm, inner ysep=0.085cm, rounded corners=0.020cm] at (27.125,13.25) {$W_{Uhl}$};
    \draw [ color={rgb,255:red,52; green,82; blue,199}, draw opacity=1, line width=1pt, short] (26,13.75) .. controls (25.625,13.75) and (25.625,13.75) .. (25.25,13.75);
    \draw [ color={rgb,255:red,52; green,82; blue,199}, draw opacity=1, line width=1pt, short] (12.5,12) .. controls (19.25,12) and (19.25,12) .. (26,12);
    \draw [ color={rgb,255:red,52; green,82; blue,199}, draw opacity=1, line width=1pt, short] (29.5,14.125) .. controls (28.75,14.125) and (28.875,14.125) .. (28.125,14.125);
    \draw [ color={rgb,255:red,52; green,82; blue,199}, draw opacity=1, line width=1pt, short] (29.5,13.125) .. controls (28.75,13.125) and (28.875,13.125) .. (28.125,13.125);
    \draw [ color={rgb,255:red,52; green,82; blue,199}, draw opacity=1, line width=1pt, short] (29.5,12.125) .. controls (28.75,12.125) and (28.875,12.125) .. (28.125,12.125);
    \draw [ color={rgb,255:red,52; green,82; blue,199}, draw opacity=1, line width=1pt, short] (10,17.75) .. controls (11.25,14.875) and (11.25,14.875) .. (12.5,12);
    \draw [ color={rgb,255:red,52; green,82; blue,199}, draw opacity=1, line width=1pt, short] (25.25,16.5) .. controls (25.25,15.125) and (25.25,15.125) .. (25.25,13.75);
    \draw [ color={rgb,255:red,52; green,82; blue,199}, draw opacity=1, line width=1pt, short] (30.875,16.125) .. controls (30.125,15.125) and (30.25,15.125) .. (29.5,14.125);
    \draw [ color={rgb,255:red,52; green,82; blue,199}, draw opacity=1, line width=1pt, short] (30.875,16.75) .. controls (30.25,17.125) and (30.25,17.125) .. (29.625,17.625);
    \draw [ color={rgb,255:red,52; green,82; blue,199}, draw opacity=1, line width=1pt, short] (33.25,17.25) .. controls (31.375,15.125) and (31.375,15.125) .. (29.5,13.125);
    \draw [ color={rgb,255:red,52; green,82; blue,199}, draw opacity=1, line width=1pt, short] (33.25,17.25) .. controls (31.375,14.625) and (31.375,14.625) .. (29.5,12.125);
    \draw [ color={rgb,255:red,52; green,82; blue,199}, draw opacity=1, line width=1pt, short] (33.25,17.25) .. controls (32.25,18) and (32.25,18) .. (31.25,18.875);
    \node [font=\fontsize{14.2pt}{18.5pt}\selectfont, fill={rgb,255:red,255; green,255; blue,255}, fill opacity=1, text opacity=1, inner xsep=0.080cm, inner ysep=0.085cm, rounded corners=0.020cm] at (24.75,24) {s = 3};
    \draw [-{Stealth[scale=1.5]}, ] (24.75,23.625) -- (24.75,23.125);
    \node [font=\fontsize{14.2pt}{18.5pt}\selectfont, fill={rgb,255:red,255; green,255; blue,255}, fill opacity=1, text opacity=1, inner xsep=0.080cm, inner ysep=0.085cm, rounded corners=0.020cm] at (28.75,24) {s = 4};
    \draw [-{Stealth[scale=1.5]}, ] (28.75,23.625) -- (28.75,23.125);
    \node [font=\fontsize{18.2pt}{23.7pt}\selectfont, fill={rgb,255:red,255; green,255; blue,255}, fill opacity=1, text opacity=1, inner xsep=0.080cm, inner ysep=0.085cm, rounded corners=0.020cm] at (24.625,18.125) {$\tilde{A}_{\text{out}}$};
    \node [font=\fontsize{18.2pt}{23.7pt}\selectfont, fill={rgb,255:red,255; green,255; blue,255}, fill opacity=1, text opacity=1, inner xsep=0.080cm, inner ysep=0.085cm, rounded corners=0.020cm] at (28.875,12.5) {$Y_4^n$};
    \node [font=\fontsize{18.2pt}{23.7pt}\selectfont, fill={rgb,255:red,255; green,255; blue,255}, fill opacity=1, text opacity=1, inner xsep=0.080cm, inner ysep=0.085cm, rounded corners=0.020cm] at (28.75,13.625) {$C_4^n$};
    \node [font=\fontsize{18.2pt}{23.7pt}\selectfont, fill={rgb,255:red,255; green,255; blue,255}, fill opacity=1, text opacity=1, inner xsep=0.080cm, inner ysep=0.085cm, rounded corners=0.020cm] at (31.5,16.375) {$\Phi_{\text{out}}$};
    \node [font=\fontsize{18.2pt}{23.7pt}\selectfont, fill={rgb,255:red,255; green,255; blue,255}, fill opacity=1, text opacity=1, inner xsep=0.080cm, inner ysep=0.085cm, rounded corners=0.020cm] at (33.875,17.25) {$\approx \tau_{\text{out}}$};
    \node [font=\fontsize{18.2pt}{23.7pt}\selectfont, fill={rgb,255:red,255; green,255; blue,255}, fill opacity=1, text opacity=1, inner xsep=0.080cm, inner ysep=0.085cm, rounded corners=0.020cm] at (13.75,12.5) {$Y_0^n$};
    \end{circuitikz}
    }%
    \caption{State redistribution of a type-constrained de Finetti state with side information at the receiver. At time $s = 0$, Alice owns the systems $E_0^n$ and $C_0^n$ while Bob owns the system $Y_0^n$. They both share some entanglement state $\Phi_{\text{in}}$, from which Alice owns the system $\tilde{A}_{\text{in}}$ and Bob owns the system $\tilde{B}_{\text{in}}$. At time~$s = 1$, Alice applied the Uhlmann Isometry $V_{\text{Uhl}}^\dagger$ to generate $E_2^n$ and $\tilde{A}_{\text{out}}\tilde{C}$, the latter being sent to Bob. At time~$s = 2$, Bob applied the decoupling unitary $U_{\text{dec}}^\dagger$ to recover the system $C_2^n$ that he sends to Alice. At time~$s = 3$, Alice applied the decoupling unitary $U_{\text{dec}}$ to generate $\tilde{A}_{\text{out}}$, her part of the shared entanglement at the end of the protocol, and $\tilde{B}_{\text{in}}\tilde{C}$ that she sends to Bob. Finally, at time~$s = 4$, Bob applies the Uhlmann isometry $W_{\text{Uhl}}$ to recover the system $C_4^n$, which is equivalent to the initial system $C_0^n$ that Alice owned, the system $\tilde{B}_{\text{out}}$ being his part of the shared entanglement state at the end of the protocol, and the system $Y_4^n$ isomorphic to $Y_0^n$. The full state at the end of the protocol, excluding the shared entanglement, is a state close to $\tau_{\text{out}}$.}
    \label{fig:state_redistribution_task}
\end{figure*}
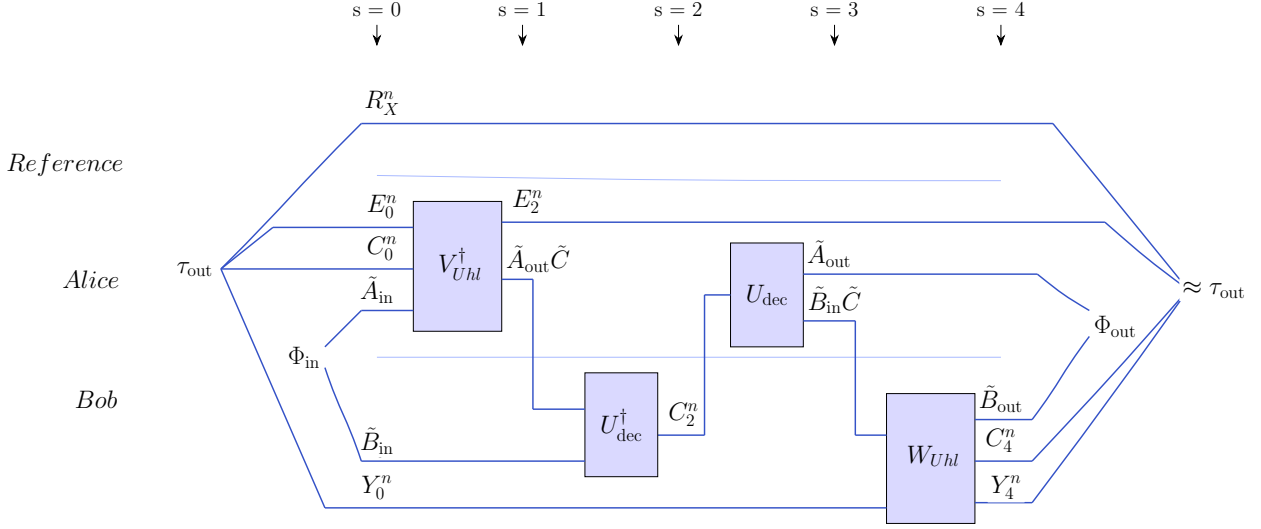

\textit{Proof of Proposition \ref{prop:de_finetti_state_redistribution}} We use the subscripts $s = 0,1,2,3,4$ to keep track of the output de Finetti state at each time of the state redistribution protocol in Figure \ref{fig:state_redistribution_task}.  At time $s = 0$, we relabel the output type-constrained de Finetti state to be
\begin{equation}
    \label{eq:output_state_redistribution_relabeled}
    \ket{\tau_0^{\tilde{t},\delta}}_{R_TR_{XY}^n E_0^n C_0^n Y_0^n} := \frac{1}{\sqrt{\left|\mathcal{P}_{\mathbb{X}^n \times \mathbb{Y}^n}^{\tilde{t},\delta}\right|}} \sum_{t \in \mathcal{P}_{\mathbb{X}^n \times \mathbb{Y}^n}^{\tilde{t},\delta}} \ket{{\sigma'}_t}_{R_{XY}E_0C_0Y_0}^{\otimes n} \otimes \ket{t}_{R_T}.
\end{equation}
We can define
\begin{equation}
    \label{eq:state_at_time_2_redistribution}
    \ket{\tau_2^{\tilde{t},\delta}}_{R_TR_{XY}^n E_2^n C_2^n Y_0^n} := \mathcal{I}_{E_0^nC_0^n \rightarrow E_2^n C_2^n} \ket{\tau_0^{\tilde{t},\delta}}_{R_TR_{XY}^n E_0^n C_0^n Y_0^n}
\end{equation}
and
\begin{equation}
    \label{eq:state_at_time_4_redistribution}
    \ket{\tau_4^{\tilde{t},\delta}}_{R_TR_{XY}^n E_2^n C_4^n Y_4^n} := \mathcal{I}_{C_2^nY_0^n \rightarrow C_4^n Y_4^n} \ket{\tau_2^{\tilde{t},\delta}}_{R_TR_{XY}^n E_2^n C_2^n Y_0^n}
\end{equation}
to be the same states but on different systems at different times in the protocol, in which the isometry $\mathcal{I}$ simply relabels the systems. Equivalently, $\ket{\tau_2^{\tilde{t},\delta}}_{R_TR_{XY}^n E_2^n C_2^n Y_0^n}$ is obtained after applying the state splitting protocol to~$\ket{\tau_0^{\tilde{t},\delta}}_{R_TR_{XY}^n E_0^n C_0^n Y_0^n}$, and~$\ket{\tau_4^{\tilde{t},\delta}}_{R_TR_{XY}^n E_2^n C_4^n Y_4^n}$ is obtained after applying the state merging protocol to~$\ket{\tau_0^{\tilde{t},\delta}}_{R_TR_{XY}^n E_0^n C_0^n Y_0^n}$ (up to small errors from these protocols). At time $s = 0$ of the state redistribution protocol, the full state is $\ket{\tau_0^{\tilde{t},\delta}}_{R_TR_{XY}^n E_0^n C_0^n Y_0^n} \otimes \ket{\phi_{\text{in}}}_{\tilde{A}_{\text{in}}\tilde{B}_{\text{in}}}$ with $\ket{\phi_{\text{in}}}_{\tilde{A}_{\text{in}}\tilde{B}_{\text{in}}}$ being the pre-shared entanglement between Alice (on system $\tilde{A}_{\text{in}}$) and Bob (on system $\tilde{B}_{\text{in}}$) which is independent of the input. At the end of the state splitting protocol, the full state is, up to a small error from the whole protocol, $\ket{\tau_4^{\tilde{t},\delta}}_{R_TR_{XY}^n E_2^n C_4^n Y_4^n} \otimes \ket{\phi_{\text{out}}}_{\tilde{A}_{\text{out}}\tilde{B}_{\text{out}}}$ with $\ket{\phi_{\text{out}}}_{\tilde{A}_{\text{out}}\tilde{B}_{\text{out}}}$ being the resulting, from the state merging protocol, shared entanglement between Alice (on system $\tilde{A}_{\text{out}}$) and Bob (on system $\tilde{B}_{\text{out}}$), which is again independent of the output state.

Before applying the bidecoupling theorem on the state $\ket{\tau_2^{\tilde{t},\delta}}_{R_TR_{XY}^n E_2^n C_2^n Y_0^n}$ and constructing the state splitting protocol from $\ket{\tau_0^{\tilde{t},\delta}}_{R_TR_{XY}^n E_0^n C_0^n Y_0^n}$ to $\ket{\tau_2^{\tilde{t},\delta}}_{R_TR_{XY}^n E_2^n C_2^n Y_0^n}$ and the state merging protocol from $\ket{\tau_2^{\tilde{t},\delta}}_{R_TR_{XY}^n E_2^n C_2^n Y_0^n}$ to $\ket{\tau_4^{\tilde{t},\delta}}_{R_TR_{XY}^n E_2^n C_4^n Y_4^n}$, we make two different $\delta$-typical projections on the state $\ket{\tau_2^{\tilde{t},\delta}}_{R_TR_{XY}^n E_2^n C_2^n Y_0^n}$. Just like in the previous section, this allows to reduce the dimensions on which that state lives to get good communication cost when using the bidecoupling theorem, all while preserving a state close enough to the original, non-projected state. One $\delta$-typical projection is necessary to construct the state splitting protocol from time $s = 0$ to $s = 2$, while another $\delta$-typical projection is necessary to construct the state merging protocol from time $s = 2$ to $s = 4$.

For the protected de Finetti state for the state splitting protocol from time $s = 0$ to $s = 2$, we label it by $\tilde{\tau}$ instead of $\tau$, and for the projected de Finetti state for the state merging protocol from time $s = 2$ to $s = 4$, we label it $\overline{\tau}$. Similarly, the projections are denoted~$\hat{\Pi}^{\tilde{t},\delta}_0$ and~$\hat{\Pi}^{\tilde{t},\delta}_2$ at time $s=0$ and $s=2$ respectively for the state splitting protocol, and $\overline{\Pi}^{\tilde{t},\delta}_2$ and~$\overline{\Pi}^{\tilde{t},\delta}_4$ at time $s=2$ and $s=4$ respectively for the state merging protocol. These projectors are defined as
\begin{equation}
    \label{eq:projection_de_finetti_state_redistribution_0}
    \hat{\Pi}^{\tilde{t},\delta}_0 := \sum_{t \in \mathcal{P}_{\mathbb{X}^n \times \mathbb{Y}^n}^{\tilde{t},\delta}} \Pi_{(Y_0R_{XY})^n}^{t,\delta} \otimes \Pi_{C_0^n}^{t,\delta} \otimes \Pi_{E_0^n}^{t,\delta} \otimes \ketbra{t}{t}_{R_T},
\end{equation}
\begin{equation}
    \label{eq:projection_de_finetti_state_redistribution_2_E}
    \hat{\Pi}^{\tilde{t},\delta}_2 := \sum_{t \in \mathcal{P}_{\mathbb{X}^n \times \mathbb{Y}^n}^{\tilde{t},\delta}} \Pi_{(Y_0R_{XY})^n}^{t,\delta} \otimes \Pi_{C_2^n}^{t,\delta} \otimes \Pi_{E_2^n}^{t,\delta} \otimes \ketbra{t}{t}_{R_T},
\end{equation}
\begin{equation}
    \label{eq:projection_de_finetti_state_redistribution_2_B}
    \overline{\Pi}^{\tilde{t},\delta}_2 := \sum_{t \in \mathcal{P}_{\mathbb{X}^n \times \mathbb{Y}^n}^{\tilde{t},\delta}}\Pi_{(E_2R_{XY})^n}^{t,\delta} \otimes \Pi_{C_2^n}^{t,\delta} \otimes \Pi_{Y_0^n}^{t,\delta} \otimes \ketbra{t}{t}_{R_T},
\end{equation}
and
\begin{equation}
    \label{eq:projection_de_finetti_state_redistribution_4}
    \overline{\Pi}^{\tilde{t},\delta}_4 := \sum_{t \in \mathcal{P}_{\mathbb{X}^n \times \mathbb{Y}^n}^{\tilde{t},\delta}} \Pi_{(E_2R_{XY})^n}^{t,\delta} \otimes \Pi_{C_4^n}^{t,\delta} \otimes \Pi_{Y_4^n}^{t,\delta} \otimes \ketbra{t}{t}_{R_T}.
\end{equation}
Here, for some system $Z$, the projector $\Pi_{Z^n}^{t,\delta}$ is the projector onto the $\delta$-typical subspace of the eigenspace of the reduced state ${\sigma'}_{t,Z}^{\otimes n}$. We can define the resulting projected de Finetti states as
\begin{equation}
    \label{eq:projected_de_finetti_state_redistribution_0}
    \ket{\hat{\tau}_0^{\tilde{t},\delta}}_{R_TR_{XY}^n Y_0^n E_0^n C_0^n} := \frac{\hat{\Pi}^{\tilde{t},\delta}_0\ket{\tau_0^{\tilde{t},\delta}}_{R_TR_{XY}^n E_0^n C_0^n Y_0^n}}{\left\|\hat{\Pi}^{\tilde{t},\delta}_0 \ket{\tau_0^{\tilde{t},\delta}}_{R_TR_{XY}^n E_0^n C_0^n Y_0^n}\right\|_2},
\end{equation}
\begin{equation}
    \label{eq:projected_de_finetti_state_redistribution_2_E}
    \ket{\hat{\tau}_2^{\tilde{t},\delta}}_{R_TR_{XY}^n E_2^nC_2^nY_0^n} := \frac{\hat{\Pi}^{\tilde{t},\delta}_2\ket{\tau_2^{\tilde{t},\delta}}_{R_TR_{XY}^n E_2^n C_2^n Y_0^n}}{\left\|\hat{\Pi}^{\tilde{t},\delta}_2 \ket{\tau_2^{\tilde{t},\delta}}_{R_TR_{XY}^n E_2^n C_2^n Y_0^n}\right\|_2},
\end{equation}
\begin{equation}
    \label{eq:projected_de_finetti_state_redistribution_2_B}
    \ket{\overline{\tau}_2^{\tilde{t},\delta}}_{R_TR_{XY}^n E_2^n C_2^n Y_0^n} := \frac{\overline{\Pi}^{\tilde{t},\delta}_2\ket{\tau_2^{\tilde{t},\delta}}_{R_TR_{XY}^n E_2^n C_2^n Y_0^n}}{\left\|\overline{\Pi}^{\tilde{t},\delta}_2 \ket{\tau_2^{\tilde{t},\delta}}_{R_TR_{XY}^n E_2^n C_2^n Y_0^n}\right\|_2},
\end{equation}
and
\begin{equation}
    \label{eq:projected_de_finetti_state_redistribution_4}
    \ket{\overline{\tau}_4^{\tilde{t},\delta}}_{R_TR_{XY}^n E_2^n C_4^n Y_4^n} := \frac{\overline{\Pi}^{\tilde{t},\delta}_4\ket{\tau_4^{\tilde{t},\delta}}_{R_TR_{XY}^n E_2^n C_4^n Y_4^n}}{\left\|\overline{\Pi}^{\tilde{t},\delta}_4\ket{\tau_4^{\tilde{t},\delta}}_{R_TR_{XY}^n E_2^n C_4^n Y_4^n}\right\|_2}. 
\end{equation}

The projected de Finetti states can all be viewed to live on the $\delta$-typical subspaces on which they are projected. For a system $Z$, we label $Z^\delta \subseteq Z^n$ the subspace induced by projecting $Z^nR_T$ with the projector $\sum_t \Pi_{Z^n}^{t,\delta} \otimes \ketbra{t}{t}_{R_T}$ to obtain the subspace $Z^\delta R_T \subseteq Z^nR_T$. Therefore, we can relabel the states~$\ket{\hat{\tau}_0^{\tilde{t},\delta}}_{R_TR_{XY}^n Y_0^n E_0^n C_0^n}$, $\ket{\hat{\tau}_2^{\tilde{t},\delta}}_{R_TR_{XY}^n E_2^nC_2^nY_0^n}$, $\ket{\overline{\tau}_2^{\tilde{t},\delta}}_{R_TR_{XY}^n E_2^n C_2^n Y_0^n}$ and~$\ket{\overline{\tau}_4^{\tilde{t},\delta}}_{R_TR_{XY}^n E_2^n C_4^n Y_4^n}$\newline as~$\ket{\hat{\tau}_0^{\tilde{t},\delta}}_{R_T(R_{XY}Y_0)^\delta E_0^\delta C_0^\delta}$, $\ket{\hat{\tau}_2^{\tilde{t},\delta}}_{R_T(R_{XY}Y_0)^\delta E_2^\delta C_2^\delta}$, $\ket{\overline{\tau}_2^{\tilde{t},\delta}}_{R_T(R_{XY}E_2)^\delta C_2^\delta Y_0^\delta}$ and~$\ket{\overline{\tau}_4^{\tilde{t},\delta}}_{R_T(R_{XY}E_2)^\delta C_4^\delta Y_4^\delta}$ respectively.

Armed with the projected type-constrained de Finetti states, we can apply the bidecoupling theorem on the state $\tilde{\tau}_{2}^{\tilde{t},\delta}$. We start by constructing the state splitting protocol from time $s=0$ to $s=2$. Since $C^{\delta}$ can be expressed as $C_2^\delta = \tilde{B}_{\text{in}} \otimes \tilde{A}_{\text{out}} \otimes \tilde{C}$ for some systems $\tilde{B}_{\text{in}}$, $\tilde{A}_{\text{out}}$, $\tilde{C}$, then by the bidecoupling theorem \ref{thm:bidecoupling}, there exists a unitary operator $U_{C_2^\delta}$ such that
\begin{equation}
    \label{eq:bidecoupling_1_redistribution}
\left\|\text{Tr}_{\tilde{A}_{\text{out}}\tilde{C}}\left(U_{C_2^\delta}\hat{\tau}_{2,R_T(R_{XY}Y_0)^\delta C_2^\delta}^{\tilde{t},\delta}U_{C_2^\delta}^\dagger\right) - \frac{I_{\tilde{B}_{\text{in}}}}{d_{\tilde{B}_{\text{in}}}} \otimes \hat{\tau}^{\tilde{t},\delta}_{2,R_T(R_{XY}Y_0)^\delta}\right\|_1 \leq \hat{\varepsilon}_{\text{dec}}.
\end{equation}
The parameter $\hat{\varepsilon}_{\text{dec}}$ is given by
\begin{equation}
    \label{eq:epsilon_dec_1_redistribution}
    \hat{\varepsilon}_{\text{dec}} := \sqrt{2\frac{d_{C_2^\delta} d_{R_T (R_{XY}Y_0)^\delta} \text{Tr}\left(\left(\hat{\tau}_{2,R_T (Y_0 R_{XY})^\delta C_2^\delta}^{\tilde{t},\delta}\right)^2\right)}{d_{\tilde{A}_{\text{out}}\tilde{C}}^2}}.
\end{equation}
If we take the full state at time $s=0$, which is $\ket{\hat{\tau}_{0,R_T(R_{XY}Y_0)^\delta E_0^\delta C_0^\delta}^{\tilde{t},\delta}} \otimes \ket{\Phi}_{\tilde{A}_{\text{in}} \tilde{B}_{\text{in}}}$, and trace out systems~$E_0^\delta$,~$C_0^\delta$ and $\tilde{A}_{\text{in}}$, we obtain the same reduced state $\frac{I_{\tilde{B}_{\text{in}}}}{d_{\tilde{B}_{\text{in}}}} \otimes \hat{\tau}^{\tilde{t},\delta}_{2,R_T(R_{XY}Y_0)^\delta}$ as in equation~\eqref{eq:bidecoupling_1_redistribution}. Therefore, by Uhlmann's theorem, there exists an isometry $W_{E_2^\delta \tilde{A}_{\text{out}}\tilde{C} \rightarrow E_2^\delta \tilde{A}_{\text{out}}\tilde{C}}$ such that
\begin{equation}
    \label{eq:uhlmann_isometry_1_redistribution}
    \left\| W_{E_2^\delta \tilde{A}_{\text{out}}\tilde{C} \rightarrow E_2^\delta \tilde{A}_{\text{out}}\tilde{C}}U_{C_2^\delta}\ket{\hat{\tau}_2^{\tilde{t},\delta}}_{R_TR_{XY}^n E_2^nC_2^nY_0^n} - \ket{\hat{\tau}_0^{\tilde{t},\delta}}_{R_TR_{XY}^n Y_0^n E_0^n C_0^n} \otimes \ket{\Phi}_{\tilde{A}_{\text{in}} \tilde{B}_{\text{in}}} \right\|_2 \leq 2\sqrt{\hat{\varepsilon}_{\text{dec}}}.
\end{equation}

Proceeding similarly for the state merging protocol from time $s = 2$ to $s = 4$, we have that by the bidecoupling theorem, the same unitary $U_{C_2^\delta}$ applied on the second projected type-constrained de Finetti state $\ket{\overline{\tau}_2^{\tilde{t},\delta}}_{R_TR_{XY}^n E_2^n C_2^n Y_0^n}$ satisfies
\begin{equation}
    \label{eq:bidecoupling_2_redistribution}
    \left\|\text{Tr}_{\tilde{B}_{\text{in}}\tilde{C}}\left(U_{C_2^\delta}\overline{\tau}_{2,R_T(R_{XY}E_2)^\delta C_2^\delta}^{\tilde{t},\delta}U_{C_2^\delta}^\dagger\right) - \frac{I_{\tilde{A}_{\text{out}}}}{d_{\tilde{A}_{\text{out}}}} \otimes \overline{\tau}^{\tilde{t},\delta}_{2,R_T(R_{XY}E_2)^\delta}\right\|_1 \leq \overline{\varepsilon}_{\text{dec}},
\end{equation}
The parameter $\overline{\varepsilon}_{\text{dec}}$ from the bidecoupling theorem is given by
\begin{equation}
    \label{eq:epsilon_dec_2_redistribution}
    \overline{\varepsilon}_{\text{dec}} := \sqrt{2\frac{d_{C_2^\delta} d_{R_T (R_{XY}E_2)^\delta} \text{Tr}\left(\left(\overline{\tau}_{2,R_T (E_2 R_{XY})^\delta C_2^\delta}^{\tilde{t},\delta}\right)^2\right)}{d_{\tilde{B}_{\text{in}}\tilde{C}}^2}}.
\end{equation}
Again, if we take the full state $\ket{\overline{\tau}_{4,R_T(R_{XY}E_2)^\delta C_4^\delta Y_4^\delta}^{\tilde{t},\delta}} \otimes \ket{\Phi}_{\tilde{A}_{\text{out}} \tilde{B}_{\text{out}}}$ at time $s=4$ and trace out systems $C_4^\delta$, $Y_4^\delta$ and $\tilde{B}_{\text{out}}$, we obtain the same reduced state $\frac{I_{\tilde{A}_{\text{out}}}}{d_{\tilde{A}_{\text{out}}}} \otimes \overline{\tau}^{\tilde{t},\delta}_{2,R_T(R_{XY}E_2)^\delta}$ as in equation \eqref{eq:bidecoupling_2_redistribution}. By  Uhlmann's theorem, there exists an isometry $V_{Y_0^\delta \tilde{B}_{\text{in}}\tilde{C} \rightarrow C_4^\delta Y_4^\delta \tilde{B}_{\text{out}}}$ such that
\begin{equation}
    \label{eq:uhlmann_isometry_2_redistribution}
    \left\| V_{Y_0^\delta \tilde{B}_{\text{in}}\tilde{C} \rightarrow C_4^\delta Y_4^\delta \tilde{B}_{\text{out}}}U_{C_2^\delta}\ket{\overline{\tau}_2^{\tilde{t},\delta}}_{R_TR_{XY}^n E_2^n C_2^n Y_0^n} - \ket{\overline{\tau}_4^{\tilde{t},\delta}}_{R_TR_{XY}^n E_2^n C_4^n Y_4^n} \otimes \ket{\Phi}_{\tilde{A}_{\text{out}} \tilde{B}_{\text{out}}} \right\|_2 \leq 2\sqrt{\overline{\varepsilon}_{\text{dec}}}.
\end{equation}
In order to combine the state splitting protocol and the state merging protocol, we first derive a trace distance inequality that links the two different projected de Finetti states at time $s=2$. by the triangle inequality, we have that
\begin{align}
    \left\| \hat{\tau}_2^{\tilde{t},\delta} - \overline{\tau}_2^{\tilde{t},\delta} \right\|_1 &\leq \left\| \hat{\tau}_2^{\tilde{t},\delta} - \tau_2^{\tilde{t},\delta} \right\|_1 + \left\| \tau_2^{\tilde{t},\delta} - \overline{\tau}_2^{\tilde{t},\delta} \right\|_1 \nonumber \\
    &\leq \hat{\varepsilon}_{\delta,2} + \overline{\varepsilon}_{\delta,2},
\end{align}
in which we have defined $\hat{\varepsilon}_{\delta,2}$ and $\overline{\varepsilon}_{\delta,2}$ to be the error from the $\delta$-typical projections in Lemma \ref{lemma:deltaTypicalTraceDistanceBound}. Similarly, we also define $\hat{\varepsilon}_{\delta,0}$ and $\overline{\varepsilon}_{\delta,4}$ for the trace distance between the non-projected and projected de Finetti states at time $s=0$ and $s=4$ respectively. We can now analyze the performance of our protocol on the non-projected de Finetti state $\tau_0^{\tilde{t},\delta}$ at time $s=0$. First, note that we can now cancel the unitaries $U_{C_2^\delta}$ and $U_{C_2^\delta}^\dagger$, and the only system that now needs to be sent is the system $\tilde{C}$, which allows us to obtain optimal communication costs. This induces a new and final version of the protocol, defined simply by $VW^\dagger$. We can deduce the following inequalities (we shorten the notation by excluding $\tilde{t}$ and $\delta$):
\begin{align}
    &\left\| VW^\dagger \ket{\hat{\tau}_0^{\tilde{t},\delta}} \otimes \ket{\Phi}_{\tilde{A}_{\text{in}} \tilde{B}_{\text{in}}} - \ket{\overline{\tau}_4^{\tilde{t},\delta}} \otimes \ket{\Phi}_{\tilde{A}_{\text{out}}\tilde{B}_{\text{out}}} \right\|_1 \\
    &= \left\|VUU^\dagger W^\dagger \ket{\hat{\tau}_0^{\tilde{t},\delta}} \otimes \ket{\Phi}_{\tilde{A}_{\text{in}} \tilde{B}_{\text{in}}} - \ket{\overline{\tau}_4^{\tilde{t},\delta}} \otimes \ket{\Phi}_{\tilde{A}_{\text{out}}\tilde{B}_{\text{out}}} \right\|_1\\
    &= \left\|U^\dagger W^\dagger \ket{\hat{\tau}_0^{\tilde{t},\delta}} \otimes \ket{\Phi}_{\tilde{A}_{\text{in}} \tilde{B}_{\text{in}}} - U^\dagger V^\dagger \ket{\overline{\tau}_4^{\tilde{t},\delta}} \otimes \ket{\Phi}_{\tilde{A}_{\text{out}}\tilde{B}_{\text{out}}} \right\|_1 \\
    &\leq \left\|U^\dagger W^\dagger \ket{\hat{\tau}_0^{\tilde{t},\delta}} \otimes \ket{\Phi}_{\tilde{A}_{\text{in}} \tilde{B}_{\text{in}}} - \ket{\hat{\tau}_2^{\tilde{t},\delta}} \right\|_1\\
    &+ \left\|\ket{\hat{\tau}_2^{\tilde{t},\delta}} - \ket{\overline{\tau}_2^{\tilde{t},\delta}} \right\|_1\\
    &+ \left\|\ket{\overline{\tau}_2^{\tilde{t},\delta}} - U^\dagger V^\dagger \ket{\overline{\tau}_4^{\tilde{t},\delta}} \otimes \ket{\Phi}_{\tilde{A}_{\text{out}}\tilde{B}_{\text{out}}} \right\|_1 \\
    &\leq 2\sqrt{\hat{\varepsilon}_{\text{dec}}} + \hat{\varepsilon}_{\delta,2} + \overline{\varepsilon}_{\delta,2} + 2\sqrt{\overline{\varepsilon}_{\text{dec}}}.
\end{align}
Furthermore, we have that
\begin{align}
    &\left\| VW^\dagger \ket{\tau_0^{\tilde{t},\delta}} \otimes \ket{\Phi}_{\tilde{A}_{\text{in}} \tilde{B}_{\text{in}}} - \ket{\tau_4^{\tilde{t},\delta}} \otimes \ket{\Phi}_{\tilde{A}_{\text{out}}\tilde{B}_{\text{out}}} \right\|_1 \\
    &\leq \left\| VW^\dagger \ket{\tau_0^{\tilde{t},\delta}} \otimes \ket{\Phi}_{\tilde{A}_{\text{in}} \tilde{B}_{\text{in}}} - VW^\dagger \ket{\hat{\tau}_0^{\tilde{t},\delta}} \otimes \ket{\Phi}_{\tilde{A}_{\text{in}} \tilde{B}_{\text{in}}} \right\|_1 \\
    &+ \left\| VW^\dagger \ket{\hat{\tau}_0^{\tilde{t},\delta}} \otimes \ket{\Phi}_{\tilde{A}_{\text{in}} \tilde{B}_{\text{in}}} - \ket{\overline{\tau}_4^{\tilde{t},\delta}} \otimes \ket{\Phi}_{\tilde{A}_{\text{out}}\tilde{B}_{\text{out}}} \right\|_1 \\
    &+ \left\| \ket{\overline{\tau}_4^{\tilde{t},\delta}} \otimes \ket{\Phi}_{\tilde{A}_{\text{out}}\tilde{B}_{\text{out}}} - \ket{\tau_4^{\tilde{t},\delta}} \otimes \ket{\Phi}_{\tilde{A}_{\text{out}}\tilde{B}_{\text{out}}} \right\|_1 \\
    &\leq \hat{\varepsilon}_{\delta,0} + 2\sqrt{\hat{\varepsilon}_{\text{dec}}} + \hat{\varepsilon}_{\delta,2} + \overline{\varepsilon}_{\delta,2} + 2\sqrt{\overline{\varepsilon}_{\text{dec}}} + \overline{\varepsilon}_{\delta,4},
\end{align}
in which we used the isometric invariance of the trace norm to obtain $\hat{\varepsilon}_{\delta,0}$ and $\overline{\varepsilon}_{\delta,4}$. From Lemma \ref{lemma:deltaTypicalTraceDistanceBound}, we know the bounds on $\hat{\varepsilon}_{\delta,0}$, $\hat{\varepsilon}_{\delta,2}$, $\overline{\varepsilon}_{\delta,2}$ and $\overline{\varepsilon}_{\delta,4}$ vanish exponentially in $n$ for $\delta$ fixed. 

In order to ensure that the total error vanishes rapidly in $n$, we are left to analyze the values of $\hat{\varepsilon}_{\text{dec}}$ and $\overline{\varepsilon}_{\text{dec}}$, which depend on the communication cost. First, let us compute $\hat{\varepsilon}_{\text{dec}}$ from \eqref{eq:epsilon_dec_1_redistribution}. We know that $d_{R_T(R_{XY}Y_0)^\delta} = d_{R_T} d_{(R_{XY}Y_0)^\delta}$, in which $d_{R_T} = \left|\mathcal{P}_{\mathbb{X}^n \times \mathbb{Y}^n}^{\tilde{t},\delta}\right|$ is the dimension of the type register (the number of $\delta$-typical joint types). By typicality, let $d_{C}^{t,\delta}$ be the dimension of the $\delta$-typical subspace of the system $C$ for the state ${\sigma'}_{t,C^n}^{\otimes n}$, and similarly define $d_{Y_0R_{XY}}^{t,\delta}$ for the system $Y_0R_{XY}$. We then have that
\begin{align}
    d_{C_2^\delta} &\leq \sum_{t \in \mathcal{P}_{\mathbb{X}^n \times \mathbb{Y}^n}^{\tilde{t},\delta}} d_{C}^{t,\delta} \\
    \label{eq:upper_bound_on_dc_2_delta_dimension}
    &\leq \left|\mathcal{P}_{\mathbb{X}^n \times \mathbb{Y}^n}^{\tilde{t},\delta}\right| \exp{n\left(\max_{t} H(C)_{t} + \tilde{\omega}_{1,1}(\delta,|C|)\right)},
\end{align}
and
\begin{align}
    d_{R_T (R_{XY}Y_0)^\delta} &\leq \sum_{t \in \mathcal{P}_{\mathbb{X}^n \times \mathbb{Y}^n}^{\tilde{t},\delta}} d_{Y_0R_{XY}}^{t,\delta} \\
    &\leq \left|\mathcal{P}_{\mathbb{X}^n \times \mathbb{Y}^n}^{\tilde{t},\delta}\right| \exp{n\left(\max_{t} H(Y_0R_{XY})_{t} + \tilde{\omega}_{1,1}(\delta,|Y_0R_{XY}|)\right)},
\end{align}
in which $H(C)_t$ and $H(Y_0R_{XY})_t$ are the von Neumann entropies of the reduced states ${\sigma'}_{t,C}$ and ${\sigma'}_{t,Y_0R_{XY}}$ respectively, and $\tilde{\omega}_{k,m}(\delta,d)$ is defined in \eqref{eq:tilde-omega}. Here the maximum is taken over all types $t$ in the set $\mathcal{P}_{\mathbb{X}^n \times \mathbb{Y}^n}^{\tilde{t},\delta}$. 

We now compute the purity of the state $\tilde{\tau}_{2,R_T(R_{XY}Y_0)^\delta C_2^\delta}^{\tilde{t},\delta}$. Since the state $\tilde{\tau}_{2,R_T(R_{XY}Y_0)^\delta C_2^\delta E_2^\delta}^{\tilde{t},\delta}$ is pure, then the purity of~$\tilde{\tau}_{2,R_T(R_{XY}Y_0)^\delta C_2^\delta}^{\tilde{t},\delta}$ is the same as the purity of~$\tilde{\tau}_{2,E_2^\delta}^{\tilde{t},\delta}$. Then, by Lemma~\ref{lemma:purityBoundDeltaTypicalDeFinetti} applied to $A_{i'} = E_2$, we have that
\begin{align}
    \text{Tr}\left[\left(\hat{\tau}_{2,R_T(R_{XY}Y_0)^\delta C_2^\delta}^{\tilde{t},\delta}\right)^2\right] \leq \exp{-n\left(\min_{t} H(E_2)_{t} - \tilde{\omega}_{1,1}(\delta,|E_2|) - \frac{f(\delta)}{n}\right)},
\end{align}
in which the minimum is taken over all types $t$ in the set $\mathcal{P}_{\mathbb{X}^n \times \mathbb{Y}^n}^{\tilde{t},\delta}$, $\tilde{\omega}_{k,m}(\delta,d)$ is defined in \eqref{eq:tilde-omega} and $f(\delta)$ is defined as
\begin{equation}
    f(\delta) := -2\log(1 - \frac{\left(\hat{\varepsilon}_2^\delta\right)^2}{2}) + 1.
\end{equation}
Since the states ${\sigma'}_{t,R_{XY}ECY_0}$ are pure, we have that $H(E)_t = H(R_{XY}CB)_t$ for all $t$. Therefore, we can rewrite the bound on the purity as
\begin{align}
    \text{Tr}\left[\left(\hat{\tau}_{2,R_T(R_{XY}Y_0)^\delta C_2^\delta}^{\tilde{t},\delta}\right)^2\right] \leq \exp{-n\left(\min_{t} H(R_{XY}CB)_{t} - \tilde{\omega}_{1,1}(\delta,|E_2|) - \frac{f(\delta)}{n}\right)}.
\end{align}

Putting everything together, we obtain
\begin{equation}
    \hat{\varepsilon}_{\text{dec}} \leq \sqrt{4\frac{\exp{n\left(\hat{H} + \tilde{\omega}_{3,6}(\delta,|Y_0R_{XY}CE_2|) + \frac{f(\delta)}{n}\right)}}{d_{\tilde{A}_{\text{out}}}^2d_{\tilde{C}}^2}},
\end{equation}
in which we have defined
\begin{equation}
    \hat{H} := \max_{t} H(C)_{t} + \max_{t} H(Y_0R_{XY})_{t} - \min_{t} H(Y_0R_{XY}C)_{t},
\end{equation}
and used the crude upper bound in~\eqref{eq:crude_tilde_omega_additivity}. The same reasoning for $\overline{\varepsilon}_{\text{dec}}$ from \eqref{eq:epsilon_dec_2_redistribution} gives
\begin{equation}
    \overline{\varepsilon}_{\text{dec}} \leq \sqrt{4\frac{\exp{n\left(\overline{H} + \tilde{\omega}_{3,6}(\delta,|Y_0R_{XY}CE_2|) + \frac{g(\delta)}{n}\right)}}{d_{\tilde{B}_{\text{in}}}^2d_{\tilde{C}}^2}},
\end{equation}
in which we have defined
\begin{equation}
    \overline{H} := \max_{t} H(C)_{t} + \max_{t} H(E_2R_{XY})_{t} - \min_{t} H(E_2R_{XY}C)_{t}.
\end{equation}
and $g(\delta)$ is defined as
\begin{equation}
    g(\delta) := -2\log(1 - \frac{\left(\overline{\varepsilon}_2^\delta\right)^2}{2}) + 1.
\end{equation}

We omit the details of having to choose the entanglement systems to be such that $C_2^\delta = \tilde{B}_{\text{in}} \otimes \tilde{A}_{\text{out}} \otimes \tilde{C}$, and instead consider this aspect when proving Theorem \ref{thm:prior_free_cq_channel_simulation_with_side_info}. This concludes the proof of Proposition \ref{prop:de_finetti_state_redistribution}, by noting that $\hat{\varepsilon}_0^\delta = \hat{\varepsilon}_2^\delta$ and $\overline{\varepsilon}_2^\delta = \overline{\varepsilon}_4^\delta$ Since they are the same $\delta$-typical projections on systems that differ by a relabeling.\hfill$\square$

To prove Theorem \ref{thm:prior_free_cq_channel_simulation_with_side_info}, we proceed with the same approach as with the case with no side information at the receiver, by utilizing the continuity of entropy and the monotonicity of the trace norm to transform the worst-case entropic quantities to those of the type $t$ of the sender and receiver's input $x^ny^n \in (X \times Y)^n$. Furthermore, we carefully choose the dimensions of the pre-shared entanglement to ensure that $C_2^\delta = \tilde{B}_{\text{in}} \otimes \tilde{A}_{\text{out}} \otimes \tilde{C}$. We take the dimensions of these systems to be rounded appropriately; we omit the details.

\textit{Proof of Theorem \ref{thm:prior_free_cq_channel_simulation_with_side_info}} By monotonicity of the trace distance, for every joint type $t' \in \mathcal{P}_{\mathbb{X}^n \times \mathbb{Y}^n}$, we know that
\begin{equation}
    \left\|\left(V_{\mathcal{N}} \otimes I_{Y}\right)(\sigma_{\tilde{t}}) - \left(V_{\mathcal{N}} \otimes I_{Y}\right)(\sigma_{t'})\right\|_1 = \left\|\sigma_{\tilde{t},XY} - \sigma_{t',XY}\right\|_1 \leq \delta,
\end{equation}
Now, by the continuity of the von Neumann entropy we can transform the maximums and minimums entropic quantities over every type in $\mathcal{P}_{\mathbb{X}^n \times \mathbb{Y}^n}$ by entropic quantities over the type $\tilde{t}$ in the upper bounds of $\hat{\varepsilon}_{\text{dec}}$ and $\overline{\varepsilon}_{\text{dec}}$ of the state redistribution protocol of Proposition \ref{prop:de_finetti_state_redistribution}. Applying the continuity of entropy a second time give entropic quantities in terms of the actual type $t$ of the input $x^ny^n$ of Alice and Bob's input. This gives
\begin{equation}
    \hat{\varepsilon}_{\text{dec}} \leq \sqrt{4\frac{\exp{n\left(\hat{H}_{t} + \tilde{\omega}_{9,6}(\delta,|YR_{XY}CE|) + \frac{f(\delta)}{n}\right)}}{d_{\tilde{A}_{\text{out}}}^2d_{\tilde{C}}^2}},
\end{equation}
and
\begin{equation}
    \overline{\varepsilon}_{\text{dec}} \leq \sqrt{4\frac{\exp{n\left(\overline{H}_{t} + \tilde{\omega}_{9,6}(\delta,|YR_{XY}CE|) + \frac{g(\delta)}{n}\right)}}{d_{\tilde{B}_{\text{in}}}^2d_{\tilde{C}}^2}},
\end{equation}
in which
\begin{equation}
    \hat{H}_{t} := H(C)_{t} + H(YR_{XY})_{t} - H(YR_{XY}C)_{t},
\end{equation}
and
\begin{equation}
    \overline{H}_{t} := H(C)_{t} + H(ER_{XY})_{t} - H(ER_{XY}C)_{t}.
\end{equation}

We now define some variables that allows us to define the right dimensions for the shared entanglement such that $C_2^\delta = \tilde{B}_{\text{in}} \otimes \tilde{A}_{\text{out}} \otimes \tilde{C}$. Let 
\begin{equation}
\Delta^t_{\tilde{t},\delta} := \log(d_{C^\delta}) - n\left(H(C)_t + \tilde{\omega}_{3,2}(\delta,|YR_{XY}CE|)\right)
\end{equation}
be the difference between the actual value of the dimension (in log) of the $\delta$-typical subspace on system $C^\delta$ from the projections of \eqref{eq:projection_de_finetti_state_redistribution_2_E} and \eqref{eq:projection_de_finetti_state_redistribution_2_B} and its upper bound in \eqref{eq:upper_bound_on_dc_2_delta_dimension} (by upper bounding $\tilde{\omega}_{3,2}(\delta,|C|)$ by $\tilde{\omega}_{3,2}(\delta,|YR_{XY}CE|)$). We can express $d_{C_2^\delta}$ as
\begin{align}
    d_{C_2^\delta} &= \exp{n\left(H(C)_t + \tilde{\omega}_{3,2}(\delta,|YR_{XY}CE|)\right) + \Delta^t_{\tilde{t},\delta}}\\
    &=\exp{n\left(H(C)_t + \frac{1}{3}\tilde{\omega}_{9,6}(\delta,|YR_{XY}CE|)\right) + \Delta^t_{\tilde{t},\delta}}.
\end{align}
From the lower-bound on the size of the set of $\delta$-typical types \eqref{eq:type-class-size} and continuity of entropy, we know that
\begin{equation}
    |\Delta_{\tilde{t},\delta}^t| \leq \tilde{\omega}_{4,3}(\delta,|YR_{XY}CE|)
\end{equation}
We define a function $f_{\text{max}}(\delta)$ that upper bounds both $f(\delta)$ and $g(\delta)$. Since they both depend on some $\hat{\varepsilon}_{\delta}$ and $\overline{\varepsilon}_{\delta}$ respectively, in which the only difference between these two terms is the sum over the number of types of the alphabets of the spaces on which we project on the $\delta$-typical subspaces, we can just define some $\varepsilon_{\delta,\text{max}}$ to upper bound these two terms to define our function $f_{\text{max}}(\delta)$. We can choose
\begin{equation}
    \varepsilon_{\delta,\text{max}} := 2^{-n\frac{\delta^2}{4\ln(2)} + 1}\sqrt{2|\mathcal{P}_{Y^nR_{XY}^n}| + 2|\mathcal{P}_{E^nR_{XY}^n}| + |\mathcal{P}_{C^n}|}.
\end{equation}
With this, we define the function
\begin{equation}
    f_{\text{max}}(\delta) := -2\log(1 - \frac{\left(\varepsilon_{\delta,\text{max}}\right)^2}{2}) + 1,
\end{equation}
Finally, we define $\varepsilon_{\text{red}} := 2\hat{\varepsilon}_{\delta} + 2\overline{\varepsilon}_{\delta} + \varepsilon_{\text{bidec}}$ to be the total error of the state redistribution,
with\newline~$\varepsilon_{\text{bidec}}~:=~2\sqrt{\hat{\varepsilon}_{\text{dec}}}~+~2\sqrt{\overline{\varepsilon}_{\text{dec}}}$ being the total decoupling error.

With these variables defined we can choose the entanglement systems to be
\begin{equation}
    d_{\tilde{A}_{\text{out}}} = \frac{1}{2^{5}}\left(\varepsilon_{\text{red}} - 2\hat{\varepsilon}_{\delta} - 2\overline{\varepsilon}_{\delta}\right)^{2}\exp{\frac{n}{2}\left(I(C;Y)_{t} - \frac{1}{3}\tilde{\omega}_{9,6}(\delta,|YR_{XY}CE|) - \frac{f_{\text{max}}(\delta)}{n}\right) + \Delta^t_{\tilde{t},\delta}}
\end{equation}
and
\begin{equation}
    d_{\tilde{B}_{\text{in}}} = \frac{1}{2^{5}}\left(\varepsilon_{\text{red}} - 2\hat{\varepsilon}_{\delta} - 2\overline{\varepsilon}_{\delta}\right)^{2}\exp{\frac{n}{2}\left(I(C;E)_{t} - \frac{1}{3}\tilde{\omega}_{9,6}(\delta,|YR_{XY}CE|) - \frac{f_{\text{max}}(\delta)}{n}\right) + \Delta^t_{\tilde{t},\delta}}
\end{equation}
then, we get
\begin{equation}
    \hat{\varepsilon}_{\text{dec}} \leq \frac{2^{6}\exp{\frac{n}{2}\left(I(C;R_{XY}|Y)_{t} + \frac{4}{3}\tilde{\omega}_{9,6}(\delta,|YR_{XY}CE|) + \frac{2f_{\text{max}}(\delta)}{n}\right) - \Delta^t_{\tilde{t},\delta}}}{\left(\varepsilon_{\text{red}} - 2\hat{\varepsilon}_{\delta} - 2\overline{\varepsilon}_{\delta}\right)^{2}d_{\tilde{C}}},
\end{equation}
and
\begin{equation}
    \overline{\varepsilon}_{\text{dec}} \leq \frac{2^{6}\exp{\frac{n}{2}\left(I(C;R_{XY}|E)_{t} + \frac{4}{3}\tilde{\omega}_{9,6}(\delta,|YR_{XY}CE|) + \frac{2f_{\text{max}}(\delta)}{n}\right) - \Delta^t_{\tilde{t},\delta}}}{\left(\varepsilon_{\text{red}} - 2\hat{\varepsilon}_{\delta} - 2\overline{\varepsilon}_{\delta}\right)^{2}d_{\tilde{C}}}.
\end{equation}
Now, since $I(C;R_{XY}|E) = I(C;R_{XY}|Y)$ for pure states on systems $YR_{XY}CE$, these two decoupling bounds can be bounded by the same value. We get that the total decoupling error $\varepsilon_{\text{bidec}} := 2\sqrt{\hat{\varepsilon}_{\text{dec}}} + 2\sqrt{\overline{\varepsilon}_{\text{dec}}}$ is bounded by
\begin{equation}
    \varepsilon_{\text{bidec}} \leq 2^{5}\cdot\frac{\exp{\frac{n}{4}\left(I(C;R_{XY}|Y)_{t} + \frac{4}{3}\tilde{\omega}_{9,6}(\delta,|YR_{XY}CE|) + \frac{2f_{\text{max}}(\delta)}{n}\right) - \frac{1}{2}\Delta^t_{\tilde{t},\delta}}}{\left(\varepsilon_{\text{red}} - 2\hat{\varepsilon}_{\delta} - 2\overline{\varepsilon}_{\delta}\right)d_{\tilde{C}}^{1/2}}.
\end{equation}
Therefore, to get a fixed error $\varepsilon_{\text{red}}$ that does not depend on $\tilde{t}$, we  can choose the communication system from the sender to the receiver for every type $\tilde{t}$ to be
\begin{equation}
    d_{\tilde{C}} = \frac{2^{10}\cdot\exp{\frac{n}{2}\left(I(C;R_{XY}|Y)_{t} + \frac{4}{3}\tilde{\omega}_{9,6}(\delta,|YR_{XY}CE|) + \frac{2f_{\text{max}}(\delta)}{n}\right) - \Delta^t_{\tilde{t},\delta}}}{{\left(\varepsilon_{\text{red}} - 2\hat{\varepsilon}_{\delta} - 2\overline{\varepsilon}_{\delta}\right)^4}}.
\end{equation}
This choice of dimension for the entanglement system ensures that $C_2^\delta = \tilde{B}_{\text{in}} \otimes \tilde{A}_{\text{out}} \otimes \tilde{C}$. For a fixed error $\varepsilon_{\text{red}}$ and $\delta$, we have a protocol that achieves the prior-free simulation of the classical-quantum channel with side information at the receiver with a communication cost that scales at a rate of $I(C;R_{XY}|Y)$ up to negligible terms. The protocol consumes and generate entanglement at rates of $I(C;E)$ and $I(C;Y)$ respectively, and with an error that vanishes exponentially in $n$. This communication cost is in fact optimal for every input.

From Lemma 1 of Ref.~\cite{padda2025}, there exists a protocol that estimates the type $t$ of the input sequences $(x^n, y^n)$ to a type $\tilde{t}$ up to error $\delta$ in total variation distance, with high probability. This protocol uses an amount of communication between the two participants that is sublinear in $n$. We define the following protocol that simulates the isometric extension $V_{\mathcal{N}, X \rightarrow EC}^{\otimes n}$ and the transmission of the output system $C^n$ from a sender to a receiver.
\begin{tcolorbox}[boxrule=0.6pt,colback=white,breakable,sharp corners=all]
\begin{protocol}
\label{protocol:cq_channel_simulation_side_info}
Let $\delta > 0$ and $n \in \mathbb{N}_{> 0}$. Let $\mathcal{N}_{X \rightarrow C}$ be a classical-quantum channel and $V_{\mathcal{N}, X \rightarrow EC}$ be an isometric extension of that channel on a purification system $E$ with side information to the receiver on a system $Y$. The protocol works as follows:
\begin{enumerate}
    \item The sender Alice and the receiver Bob apply the protocol from Lemma 1 of Ref.~\cite{padda2025} so that they both obtain an estimate $\tilde{t}$ of their input's joint type, using an amount of communication sublinear in $n$, with $n$ being the length of their input sequences.
    \item Alice applies the isometric extension $V_{\mathcal{N}}^{\otimes n}$ of the channel $\mathcal{N}^{\otimes n}$ on her input sequence $x^n$.
    \item Alice and Bob apply the state redistribution protocol of the $\delta$-typical de Finetti state associated to the estimated joint type $\tilde{t}$.
\end{enumerate}
\end{protocol}
\end{tcolorbox}
\noindent We show the Protocol \ref{protocol:cq_channel_simulation_side_info} successfully simulates the action of the channel $V_{\mathcal{N}}^{\otimes n}$ on any $\delta$-typical classical joint inputs. Let $\Pi_{\text{sim}}$ be our simulation protocol defined as $\Pi_{\text{sim}} := \Pi_{sr}^{\tilde{t},\delta}V_{\mathcal{N}}^{\otimes n}$, in which $\Pi_{sr}^{\tilde{t},\delta}$ is the state redistribution protocol of the $\delta$-typical de Finetti state associated with the estimated joint type $\tilde{t}$ from Proposition \ref{prop:de_finetti_state_redistribution}. We extend the protocol on the whole space by appending the identity on the subspace orthogonal to it. This makes $\Pi_{sr}^{\otimes n}$ permutation-covariant since the $\delta$-typical subspace is invariant under permutation. We know that $V_{\mathcal{N}}^{\otimes n}$ is also permutation-covariant with respect to the action of any permutation~$\pi$ on the inputs. Therefore, $\Pi_{\text{sim}}$ is also permutation-covariant and so is $\Delta := \Pi_{\text{sim}} - V_{\mathcal{N}}^{\otimes n}$. By our de Finetti reduction in the form of Theorem \ref{thm:restrictedClassicalQuantumDeFinettiReduction}, for every $x^ny^n \in (X \times Y)^n$ we have the following:
\begin{align}
\left\|\left( \Pi_{\text{sim}} - V_{\mathcal{N}}^{\otimes n} \right) \left(\ketbra{x^ny^n}{x^ny^n}_{X^nY^n}\right)\right\|_1 &= \left\|\left( \Pi_{\text{sim}} - V_{\mathcal{N}}^{\otimes n} \right) \left(\ketbra{x^ny^n}{x^ny^n}_{X^nY^n} \otimes \ketbra{x^ny^n}{x^ny^n}_{R_{XY}^n}\right)\right\|_1\\
&\leq \left\|\left( \Pi_{\text{sim}} - V_{\mathcal{N}}^{\otimes n}\right)\right\|_{\diamond,cl}^{\tilde{t},\delta}\\
&\leq\left|\mathcal{P}_{\mathbb{X}^n \times \mathbb{Y}^n}\right|\left|\mathcal{P}_{\mathbb{X}^n \times \mathbb{Y}^n}^{\tilde{t},\delta}\right|\left\| \left(\Delta \otimes \textbf{id}_{R}\right) \left(\ketbra{\tau^{\tilde{t},\delta}}{\tau^{\tilde{t},\delta}}_{X^nY^nR}\right) \right\|_1\\
&\leq \left|\mathcal{P}_{\mathbb{X}^n \times \mathbb{Y}^n}\right|\left|\mathcal{P}_{\mathbb{X}^n \times \mathbb{Y}^n}^{\tilde{t},\delta}\right| \left(2\hat{\varepsilon}_{\delta} + 2\overline{\varepsilon}_{\delta} + \varepsilon_{\text{bidec}} \right)\\
&\leq (n+1)^{2|X||Y|} \left(2\hat{\varepsilon}_{\delta} + 2\overline{\varepsilon}_{\delta} + \varepsilon_{\text{bidec}} \right).
\end{align}
Here, in the first equality, we implicitly act as the identity when adding the reference system $R_{XY}^n$. This ends the proof of Theorem \ref{thm:prior_free_cq_channel_simulation_with_side_info}, by noting that $|R_{XY}| = |X||Y|$, and that $I(C;R_{XY}|Y)=I(C;XY|Y)=I(C;X|Y)$ for classical states on $XY$ passed through classical-quantum channels $\mathcal{N}_{X \rightarrow C}$ (viewed as a classical-quantum state).\begin{flushright}
    $\square$
\end{flushright}
\section{Prior-Free Compression of Quantum Interactive Communication Protocols on Classical Inputs}
\label{section:QuantumInteractiveCommunicationSimulation}
We now move on to our last main result, which is the simulation of quantum interactive protocols with classical inputs in the prior-free setting. Interactive protocols are more general than one-way protocols, as they allow for multiple rounds of communication between the two parties. This added complexity introduces new challenges in the compression of such protocols, especially when considering prior-free scenarios where no assumptions are made about the input distribution.

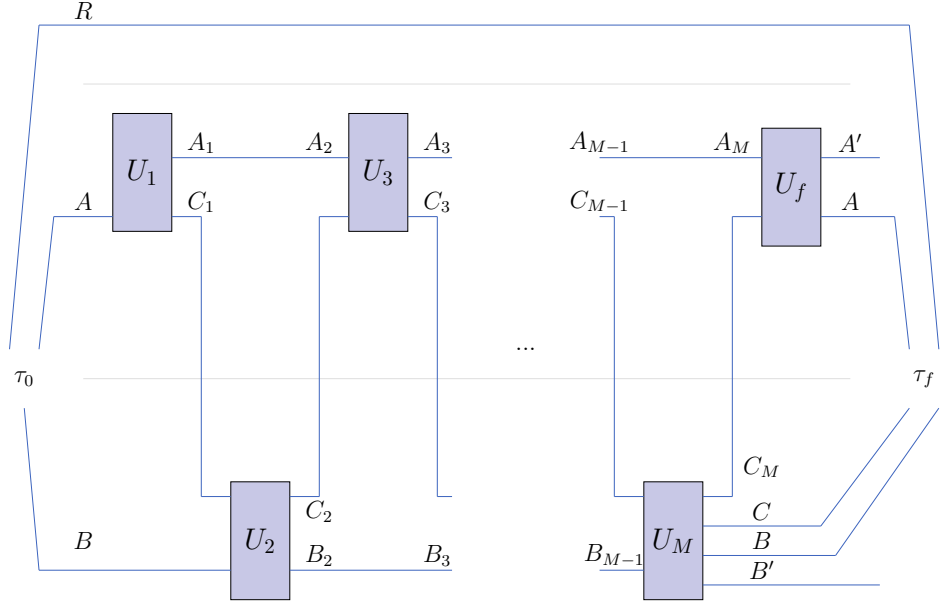
\begin{figure*}[t]
\centering
\resizebox{0.75\textwidth}{!}{%
\begin{circuitikz}
\tikzstyle{every node}=[font=\LARGE]
\node [font=\large] at (7.25,13.75) {$\tau_0$};
\draw [ fill={rgb,255:red,200; green,200; blue,230} ] (10.75,12) rectangle (11.75,10);
\draw [ fill={rgb,255:red,200; green,200; blue,230} ] (8.75,18.25) rectangle (9.75,16.25);
\draw [ fill={rgb,255:red,200; green,200; blue,230} ] (12.75,18.25) rectangle (13.75,16.25);
\draw [ fill={rgb,255:red,200; green,200; blue,230} ] (17.75,12) rectangle (18.75,10);
\draw [ fill={rgb,255:red,200; green,200; blue,230} ] (19.75,18) rectangle (20.75,16);
\node [font=\large] at (15.75,14.25) {$. . .$};
\draw [ color={rgb,255:red,222; green,222; blue,222}, line width=0.2pt, short] (8.25,13.75) -- (21.25,13.75);
\draw [ color={rgb,255:red,222; green,222; blue,222}, line width=0.2pt, short] (8.25,18.75) -- (21.25,18.75);
\draw [ color={rgb,255:red,60; green,100; blue,190}, short] (7.5,14.25) -- (7.75,16.5);
\draw [ color={rgb,255:red,60; green,100; blue,190}, short] (7.75,16.5) -- (8.75,16.5);
\draw [ color={rgb,255:red,60; green,100; blue,190}, short] (7,14.25) -- (7.5,19.75);
\draw [ color={rgb,255:red,60; green,100; blue,190}, short] (7.5,19.75) -- (22.25,19.75);
\draw [ color={rgb,255:red,60; green,100; blue,190}, short] (9.75,16.5) -- (10.25,16.5);
\draw [ color={rgb,255:red,60; green,100; blue,190}, short] (10.25,16.5) -- (10.25,11.75);
\draw [ color={rgb,255:red,60; green,100; blue,190}, short] (10.25,11.75) -- (10.75,11.75);
\draw [ color={rgb,255:red,60; green,100; blue,190}, short] (11.75,11.75) -- (12.25,11.75);
\draw [ color={rgb,255:red,60; green,100; blue,190}, short] (12.25,11.75) -- (12.25,16.5);
\draw [ color={rgb,255:red,60; green,100; blue,190}, short] (12.25,16.5) -- (12.75,16.5);
\draw [ color={rgb,255:red,60; green,100; blue,190}, short] (9.75,17.5) -- (12.75,17.5);
\draw [ color={rgb,255:red,60; green,100; blue,190}, short] (13.75,16.5) -- (14.25,16.5);
\draw [ color={rgb,255:red,60; green,100; blue,190}, short] (14.25,16.5) -- (14.25,11.75);
\draw [ color={rgb,255:red,60; green,100; blue,190}, short] (14.25,11.75) -- (14.5,11.75);
\draw [ color={rgb,255:red,60; green,100; blue,190}, short] (13.75,17.5) -- (14.5,17.5);
\draw [ color={rgb,255:red,60; green,100; blue,190}, short] (17.75,11.75) -- (17.25,11.75);
\draw [ color={rgb,255:red,60; green,100; blue,190}, short] (17.25,11.75) -- (17.25,16.5);
\draw [ color={rgb,255:red,60; green,100; blue,190}, short] (17.25,16.5) -- (17,16.5);
\draw [ color={rgb,255:red,60; green,100; blue,190}, short] (18.75,11.75) -- (19.25,11.75);
\draw [ color={rgb,255:red,60; green,100; blue,190}, short] (19.25,11.75) -- (19.25,16.5);
\draw [ color={rgb,255:red,60; green,100; blue,190}, short] (19.25,16.5) -- (19.75,16.5);
\draw [ color={rgb,255:red,60; green,100; blue,190}, short] (17,17.5) -- (19.75,17.5);
\draw [ color={rgb,255:red,60; green,100; blue,190}, short] (11.75,10.5) -- (14.5,10.5);
\draw [ color={rgb,255:red,60; green,100; blue,190}, short] (17,10.5) -- (17.75,10.5);
\draw [ color={rgb,255:red,60; green,100; blue,190}, short] (7.5,10.5) -- (10.75,10.5);
\draw [ color={rgb,255:red,60; green,100; blue,190}, short] (7.25,13.25) -- (7.5,10.5);
\draw [ color={rgb,255:red,60; green,100; blue,190}, short] (18.75,10.75) -- (21,10.75);
\draw [ color={rgb,255:red,60; green,100; blue,190}, short] (18.75,10.25) -- (21.75,10.25);
\draw [ color={rgb,255:red,60; green,100; blue,190}, short] (18.75,11.25) -- (20.75,11.25);
\draw [ color={rgb,255:red,60; green,100; blue,190}, short] (20.75,11.25) -- (22.25,13.25);
\draw [ color={rgb,255:red,60; green,100; blue,190}, short] (21,10.75) -- (22.75,13.25);
\draw [ color={rgb,255:red,60; green,100; blue,190}, short] (20.75,16.5) -- (22,16.5);
\draw [ color={rgb,255:red,60; green,100; blue,190}, short] (22.75,14.25) -- (22.25,19.75);
\draw [ color={rgb,255:red,60; green,100; blue,190}, short] (22,16.5) -- (22.25,14.25);
\draw [ color={rgb,255:red,60; green,100; blue,190}, short] (20.75,17.5) -- (21.75,17.5);
\node [font=\Large] at (9.25,17.25) {$U_1$};
\node [font=\Large] at (11.25,11) {$U_2$};
\node [font=\Large] at (13.25,17.25) {$U_3$};
\node [font=\Large] at (20.25,17) {$U_f$};
\node [font=\Large] at (18.25,11) {$U_M$};
\node [font=\large] at (8.25,20) {$R$};
\node [font=\large] at (19.75,11.5) {$C$};
\node [font=\large] at (19.75,12.25) {$C_M$};
\node [font=\large] at (19.25,17.75) {$A_M$};
\node [font=\large] at (17.25,10.75) {$B_{M-1}$};
\node [font=\large] at (17,16.75) {$C_{M-1}$};
\node [font=\large] at (17,17.75) {$A_{M-1}$};
\node [font=\large] at (14.25,10.75) {$B_3$};
\node [font=\large] at (14.25,17.75) {$A_3$};
\node [font=\large] at (14.25,16.75) {$C_3$};
\node [font=\large] at (12.25,10.75) {$B_2$};
\node [font=\large] at (12.25,11.5) {$C_2$};
\node [font=\large] at (12.25,17.75) {$A_2$};
\node [font=\large] at (10.25,16.75) {$C_1$};
\node [font=\large] at (10.25,17.75) {$A_1$};
\node [font=\large] at (8.25,11) {$B$};
\node [font=\large] at (8.25,16.75) {$A$};
\node [font=\large] at (21.25,16.75) {$A$};
\node [font=\large] at (21.25,17.75) {$A'$};
\node [font=\large] at (19.75,10.5) {$B'$};
\node [font=\large] at (19.75,11) {$B$};
\node [font=\large] at (22.5,13.75) {$\tau_f$};
\end{circuitikz}
}%
\caption{A general quantum interactive communication protocol $\Pi$ with classical inputs, consisting of $r$ rounds of communication between Alice and Bob. The protocol starts with Alice and Bob receiving classical inputs $x$ and $y$, respectively, drawn from an unknown distribution over the input space $\mathbb{X} \times \mathbb{Y}$. At each round~$i$, one party applies a quantum operation in the form of an isometry $U_i$ on their local quantum system and the quantum system received from the other party in the previous round (if any), and then sends a quantum system to the other party. After~$r$ rounds of communication, both parties perform measurements on their local quantum systems to produce their respective outputs. The total communication cost of the protocol is defined as the sum of the sizes (in qubits) of all quantum systems exchanged during the~$r$~rounds of communication. Note, the pre-shared entanglement at every round is not represented in the figure for clarity.}
\label{fig:interactive_protocol}
\end{figure*}

Our goal is to simulate an interactive protocol $\Pi$ in the asymptotic and prior-free setting, while minimizing the total communication cost. In the asymptotic setting, the protocol to simulate is $\Pi^{\otimes n}$, which consists of $n$ independent instances of the original protocol $\Pi$, and Alice and Bob are given $(x^n,y^n) \in \mathbb{X}^n \times \mathbb{Y}^n$ as classical inputs, respectively. In the prior-free setting, we assume that Alice and Bob do not have any prior knowledge about the joint distribution of their sequences of inputs $(x^n,y^n)$, and we thus ensure that the protocol works well on the worst-case input distribution. We denote $\Pi_{\text{sim}}$ the protocol that simulates $\Pi^{\otimes n}$.

In order to simulate the protocol $\Pi^{\otimes n}$, we first note that this protocol is classical-quantum since the inputs are classical. Therefore, we can apply our de Finetti reduction from Theorem \ref{thm:restrictedClassicalQuantumDeFinettiReduction}. If we use the subsampling protocol from Lemma 1 of Ref.~\cite{padda2025} on the input sequences $x^n$ and $y^n$ for a fixed $\delta > 0$, then with high probability we get an estimate $\tilde{t}$ of the true type $t$ up to a total variation distance $\delta$. Using Theorem \ref{thm:restrictedClassicalQuantumDeFinettiReduction}, we find that
\begin{align}
    &\max_{x^ny^n} \left\|\left(\left(\Pi_{\text{sim}} - \Pi^{\otimes n}\right)\otimes I_{R}\right)\left(x^ny^n \otimes {x^ny^n}_{R}\right) \right\|_1\\
    &= \max_{\tilde{t} \in \mathcal{P}_{\mathbb{X}^n \times \mathbb{Y}^n}} \max_{x^ny^n \in T_{\tilde{t}}} \left\|\left(\left(\Pi_{\text{sim}} - \Pi^{\otimes n}\right)\otimes I_{R}\right)\left(x^ny^n \otimes {x^ny^n}_{R}\right) \right\|_1\\
    &\leq \max_{\tilde{t} \in \mathcal{P}_{\mathbb{X}^n \times \mathbb{Y}^n}} \left\|\Pi_{\text{sim}} - \Pi^{\otimes n}\right\|_{\diamond,cl}^{\tilde{t},\delta} \\
    &\leq \max_{\tilde{t} \in \mathcal{P}_{\mathbb{X}^n \times \mathbb{Y}^n}} \left|\mathcal{P}_{\mathbb{X}^n \times \mathbb{Y}^n}\right|^2 \left\|\left(\left(\Pi_{\text{sim}} - \Pi^{\otimes n}\right) \otimes I_{R}\right)\left(\tau_{X^nY^nR}^{\tilde{t},\delta}\right)\right\|_1.
\end{align}
We can rewrite $\Pi^{\otimes n} = \prod_{i=1}^n U_i^{\otimes n}$. At every round $i$, applying $U_i^{\otimes n}$ on the input state $\tau_{X^nY^nR}^{\tilde{t},\delta}$ preserves the de Finetti structure of the state. Therefore, to simulate each unitary $U_i^{\otimes n}$, we can use our one-way prior-free compression protocol for de Finetti states from Proposition \ref{prop:de_finetti_state_redistribution}. Therefore, for each unitary $U_i^{\otimes n}$, we define $V_iW^\dagger_iU_i^{\otimes n}$ as the isometries that simulates $U_i^{\otimes n}$ on the de Finetti input state $\tau_{X^nY^nR}^{\tilde{t},\delta}$ with error at most $\varepsilon_{i,\text{red}}^{\tilde{t},\delta}$, in which $W_i^\dagger U_i^{\otimes n}$ is applied by the sender of round $i$ before sending the message, and $V_i$ is applied by the receiver of round $i$ after receiving the message.

By applying this protocol at each round of communication, we can simulate the entire interactive protocol $\Pi^{\otimes n}$ in the prior-free setting. Indeed, we can rewrite
\begin{align}
    &\left\|\left(\left(\left(\Pi_{\text{sim}} - \Pi^{\otimes n}\right) \otimes I_{R}\right)\right)\left(\tau_{X^nY^nR}^{\tilde{t},\delta}\right)\right\|_1 \\
    &= \left\| \left( \prod_{i=1}^r V_i W_i^\dagger U_i^{\otimes n} \otimes I_{R} - \prod_{i=1}^r U_i^{\otimes n} \otimes I_{R} \right) \left(\tau_{X^nY^nR}^{\tilde{t},\delta}\right) \right\|_1 \\
    &\leq \sum_{i=1}^r \left\| \left( V_i W_i^\dagger U_i^{\otimes n} \otimes I_{R} - U_i^{\otimes n} \otimes I_{R} \right) \left(\tau_{X^nY^nR}^{\tilde{t},\delta}\right) \right\|_1 \\
    &\leq \sum_{i=1}^r \varepsilon_{i,\text{red}}^{\tilde{t},\delta}.
\end{align}
Therefore, define $QCC(\Pi_{\text{sim}})_{A\to B}$ and $QCC(\Pi_{\text{sim}})_{B\to A}$ as the total quantum communication cost of the simulating protocol $\Pi_{\text{sim}}$ from Alice to Bob and from Bob to Alice, respectively. Let $QCC(\Pi_{\text{sim}}) = QCC(\Pi_{\text{sim}})_{A\to B} + QCC(\Pi_{\text{sim}})_{B\to A}$ be the total quantum communication cost of the simulating protocol $\Pi_{\text{sim}}$. We have the following theorem.
\begin{theorem}
\label{thm:int}
    Let $\Pi$ be a quantum interactive communication protocol with classical inputs, consisting of $r$ rounds of communication and $\delta > 0$. For any $n \in \mathbb{N}$, there exists a variable-length prior-free simulation protocol $\Pi_{\text{sim}}$ that simulates $\Pi^{\otimes n}$ on such that
    \begin{align}
        &\max_{x^ny^n} \left\| \left(\Pi_{\text{sim}} - \Pi^{\otimes n}\right)\otimes I_{R}\left((x^ny^n)^{\otimes 2}\right) \right\|_1 \leq \sum_{i=1}^r \varepsilon_{i,\text{red}}^{\tilde{t},\delta} \left|\mathcal{P}_{\mathbb{X}^n \times \mathbb{Y}^n}\right|^2,
    \end{align}
    in which $\varepsilon_{i,\text{red}}^{\tilde{t},\delta}$ is the error of simulating the isometry $U_i^{\otimes n}$ on the input de Finetti state $\tau_{X^nY^nR}^{\tilde{t},\delta}$ (the redistribution of the latter after applying the isometry $U_i^{\otimes n}$). Furthermore, the total quantum communication cost of the simulating protocol $\Pi_{\text{sim}}$ is given by
    \begin{equation}
        \label{eq:total_communication_cost_interactive_protocol}
        QCC(\Pi_{\text{sim}}) = QCC(\Pi_{\text{sim}})_{A\to B} + QCC(\Pi_{\text{sim}})_{B\to A},
    \end{equation}
    in which
    \begin{align}
        QCC(\Pi_{\text{sim}})_{A\to B} &= n\left(\sum_{i=0}I(C_{2i+1}:R_{XY}|B_{2i+1})_{\tilde{t}}\right),\\
        QCC(\Pi_{\text{sim}})_{B\to A} &= n\left(\sum_{i=1}I(C_{2i}:R_{XY}|A_{2i})_{\tilde{t}}\right),
    \end{align}
    up to negligible terms, with $I(C_{2i+1}:R_{XY}|B_{2i+1})_{\tilde{t}}$ and $I(C_{2i}:R_{XY}|A_{2i})_{\tilde{t}}$ being the conditional mutual informations evaluated on the state $\sigma_{i,\tilde{t}}$, the distribution of the type $\tilde{t}$ after the application of the isometries up to round $i$ in the original protocol $\Pi$. This proves that $QIC(\Pi) = AQCC(\Pi)$.
\end{theorem}
\begin{proof}
    We prove the part where $QIC(\Pi) = AQCC(\Pi)$. From our result, we know that once we fix $n \in \mathbb{N}$, then the result works for every $\delta$ up to errors $\sum_{i=1}^r \varepsilon_{i,\text{red}}^{\tilde{t},\delta} \left|\mathcal{P}_{\mathbb{X}^n \times \mathbb{Y}^n}\right|^2$. Now, for a fixed $n \in \mathbb{N}$ and simulation error $\varepsilon$, we can choose $\delta$ to scale as $\frac{1}{n^{1/2 - \alpha}}$ for some $\alpha \in (0,1/2)$. For simplicity, we consider $\delta = \frac{1}{n^{1/4}}$. With this definition of $\delta$, we ensure that the error terms $\varepsilon_{\delta}$ that depends on $\delta$ (and potentially $n$) vanish as $n$ goes to infinity. In particular, the largest power of $\delta$ in the error term is $\delta^2$, such that if $\delta = \frac{1}{n^{1/4}}$, then $\delta^2 = \frac{1}{n^{1/2}}$ which ensures that decreasing exponential in $n$ that are multiplied by $\delta^2$ still vanish as $n$ goes to infinity. Since the function $f_{\text{max}}(\delta)$ in the definition of $\varepsilon_{\text{red}}$ also depends on $\varepsilon_{\delta}$, it vanishes as $n$ goes to infinity. Furthermore, if $\varepsilon$ is fixed and $n$ tends to infinity, then the additionnal communication cost needed to compensate for the inaccuracy in the estimation of the type $\tilde{t}$ also vanishes. Therefore, if for each $\varepsilon$ and $n \in \mathbb{N}$ we choose $\delta = \frac{1}{n^{1/4}}$, then we have that there exists a simulation protocols $\Pi_{\text{sim},\varepsilon,n}$, such that
    \begin{align}
        \lim_{\varepsilon \to 0} \lim_{n \to \infty} \frac{QCC(\Pi_{\text{sim},\varepsilon,n})}{n} &= \sum_{i=0}I(C_{2i+1}:R_{XY}|B_{2i+1})_{\tilde{t}} + \sum_{i=1}I(C_{2i}:R_{XY}|A_{2i})_{\tilde{t}} \\
        &= QIC(\Pi).
    \end{align}
    Which concludes the proof.
\end{proof}

\section{Conclusion}
In this work, we gave a constrained de Finetti reduction where the de Finetti distribution contains only i.i.d. distributions of types that are within a distance $\delta$ of a distribution $p$. Using similar techniques from Ref.~\cite{christandl2009}, we showed that this reduction allows us to reduce the worst-case analysis of permutation-invariant classical-quantum channels, for which the input distribution is in the classical symmetric subset of some distribution $p$, to the analysis of these channels on a convex combination of i.i.d. inputs of types close to $p$. We first applied this reduction to the prior-free compression of classical-quantum channels in the quantum reverse Shannon setting, in which the receiver does not have an input sequence, to show that there exists a prior-free classical-quantum channel simulation protocol which is asymptotically optimal. We then applied a more general technique to the case where the receiver does have an input sequence, to construct a prior-free classical-quantum channel simulation protocol which is asymptotically optimal. This results in a novel extension of the quantum reverse Shannon theorem to the prior-free and asymptotic setting with side information at the receiver. Finally, we showed that for any interactive quantum communication protocol with classical inputs, there exists a prior-free simulation protocol which is asymptotically optimal. This extends the quantum reverse Shannon theorem to the prior-free and asymptotic setting for interactive quantum communication protocols with classical inputs. As a result, we have demonstrated the equivalence between the prior-free quantum information cost and the prior-free amortized quantum communication cost of an interactive quantum communication protocol with classical inputs.

Interesting extensions to this work would be to look at proving similar results but for general quantum inputs. In the case with no side information at the receiver, this was already done in Ref.~\cite{bennett_devetak_harrow_shor_winter_2014}. For the case with side information at the receiver, the simulation of a quantum channel in the prior-free setting might allow to show that the prior-free quantum information cost of a protocol equals the worst-case input amortized quantum communication cost for general interactive protocols with quantum inputs. Furthermore, improving on the bounds we have obtained could allow our result to be more easily applicable outside the theoretical setting.

\section*{Acknowledgements}
We acknowledge the support of the Natural Sciences and Engineering Research Council of Canada (NSERC), the NSERC‐Collaborative Research and Training Experience program QSciTech, the NSERC-funded Québec Ontario Consortium on Quantum Protocols (QUORUM), and the support of the Ministère de l'Enseignement Supérieur du Québec.


\appendix
\section{de Finetti Reductions Proofs}
\subsection{Proof of Lemma \ref{lemma:symmetricRestrictedMaximization}}
\label{appendix:symmetricRestrictedMaximization}

We have that
\begin{align}
    \left\|\left(\Delta \otimes \textbf{id}_{R^n}\right)\left(\rho_{X^nR^n}\right)\right\|_1 & = \left\|\left(\Delta \otimes \textbf{id}_{R^n}\right)\left(\sum_{x^n \in T^{p, \delta}_{\mathbb{X}^n}} p(x^n) \ketbra{x^n}{x^n}_{X^n} \otimes \ketbra{x^n}{x^n}_{R^n}\right)\right\|_1 \\
    &= \left\|\sum_{x^n \in T^{p, \delta}_{\mathbb{X}^n}} p(x^n) \Delta\left(\ketbra{x^n}{x^n}_{X^n}\right) \otimes \ketbra{x^n}{x^n}_{R^n}\right\|_1 \\
    &= \sum_{x^n \in T^{p, \delta}_{\mathbb{X}^n}} p(x^n) \left\|\Delta\left(\ketbra{x^n}{x^n}_{X^n}\right)\right\|_1\\
    &= \sum_{t \in \mathcal{P}^{p,\delta}_{\mathbb{X}^n}} \sum_{x^n \in T^t_{\mathbb{X}^n}} p(x^n) \left\|\Delta\left(\ketbra{x^n}{x^n}_{X^n}\right)\right\|_1.
\end{align}
We note that, by the hypothesis of permutation covariance, in which for all permutation $\pi$, there exists a CPTP map $\mathcal{K}_\pi$ such that
\begin{alignat}{2}
    &&\Delta \pi &= \mathcal{K}_\pi \Delta\\
    &&\Delta &= \mathcal{K}_\pi \Delta \pi^{-1},
\end{alignat}
we can define for each $\pi$ the map $\mathcal{K}_\pi'$ such that
\begin{equation}
    \Delta = \mathcal{K}_\pi' \Delta \pi.
\end{equation}
We get
\begin{align}
    \left\|\Delta \left(\ketbra{x^n}{x^n}\right)\right\|_1 &= \left\|\mathcal{K}_\pi' \Delta \pi \left(\ketbra{x^n}{x^n}\right)\right\|_1\\
    &\leq \left\|\Delta \pi \left(\ketbra{x^n}{x^n}\right)\right\|_1 && \text{since $\mathcal{K}_\pi'$ is CPTP}\\
    &= \left\|\mathcal{K}_\pi \Delta \left(\ketbra{x^n}{x^n}\right)\right\|_1\\
    &\leq \left\|\Delta \left(\ketbra{x^n}{x^n}\right)\right\|_1 && \text{since $\mathcal{K}_\pi$ is CPTP.}
\end{align}
Thus, we have
\begin{equation}
\left\|\Delta \left(\ketbra{x^n}{x^n}\right)\right\|_1 \leq \left\|\Delta \pi \left(\ketbra{x^n}{x^n}\right)\right\|_1 \leq \left\|\Delta \left(\ketbra{x^n}{x^n}\right)\right\|_1
\end{equation}
and therefore, for all $\pi$, we have that
\begin{equation}\left\|\Delta \left(\ketbra{x^n}{x^n}\right)\right\|_1 = \left\|\Delta \pi \left(\ketbra{x^n}{x^n}\right)\right\|_1.
\end{equation}

We can now reconstruct a symmetric state with the property we just proved. For each type, we know the sequences of this type are all equal up to permutations. We take the weighting of each $x^n$, denoted $p_{x^n}$, and distribute it uniformly over all other states of the same type. Let $\pi_{x^n \shortarrow \overline{x}^n}$ be the permutation that maps $x^n$ to $\overline{x}^n$ (which share the same type), we have
\begin{align}
    &\left\|\left(\Delta \otimes \textbf{id}_{R^n}\right)\left(\rho_{X^nR^n}\right)\right\|_1\\
    &= \sum_{t \in \mathcal{P}^{p,\delta}_{\mathbb{X}^n}} \sum_{x^n \in T^t_{\mathbb{X}^n}} p_{x^n} \left\|\Delta \left(\ketbra{x^n}{x^n}_{X^n}\right)\right\|_1\\
    &= \sum_{t \in \mathcal{P}^{p,\delta}_{\mathbb{X}^n}} \sum_{x^n \in T^t_{\mathbb{X}^n}} p_{x^n} \sum_{\overline{x}^n \in T^t_{\mathbb{X}^n}} \frac{1}{|T^t_{\mathbb{X}^n}|} \left\|\Delta \left(\ketbra{x^n}{x^n}_{X^n}\right)\right\|_1\\
    &= \sum_{t \in \mathcal{P}^{p,\delta}_{\mathbb{X}^n}} \sum_{x^n \in T^t_{\mathbb{X}^n}} p_{x^n} \sum_{\overline{x}^n \in T^t_{\mathbb{X}^n}} \frac{1}{|T^t_{\mathbb{X}^n}|} \left\|\Delta \pi_{x^n \shortarrow \overline{x}^n} \left(\ketbra{x^n}{x^n}_{X^n}\right)\right\|_1\\
    &= \sum_{t \in \mathcal{P}^{p,\delta}_{\mathbb{X}^n}} \sum_{x^n \in T^t_{\mathbb{X}^n}} p_{x^n} \sum_{\overline{x}^n \in T^t_{\mathbb{X}^n}} \frac{1}{|T^t_{\mathbb{X}^n}|} \left\|\Delta \left(\ketbra{\overline{x}^n}{\overline{x}^n}_{X^n}\right)\right\|_1\\
    &= \sum_{t \in \mathcal{P}^{p,\delta}_{\mathbb{X}^n}} p_t \sum_{\overline{x}^n \in T^t_{\mathbb{X}^n}} \frac{1}{|T^t_{\mathbb{X}^n}|} \left\|\Delta \left(\ketbra{\overline{x}^n}{\overline{x}^n}_{X^n}\right)\right\|_1 && \text{with $p_t := \sum_{x^n \in T^t_{\mathbb{X}^n}} p_{x^n}$}\\
    &= \left\|\sum_{t \in \mathcal{P}^{p,\delta}_{\mathbb{X}^n}} p_t \sum_{\overline{x}^n \in T^t_{\mathbb{X}^n}} \frac{1}{|T^t_{\mathbb{X}^n}|} \Delta \left(\ketbra{\overline{x}^n}{\overline{x}^n}_{X^n}\right) \otimes \ketbra{\overline{x}^n}{\overline{x}^n}_{R^n}\right\|_1\\
    &= \left\|\left(\Delta \otimes \textbf{id}_{R^n}\right)\left(\sum_{t \in \mathcal{P}^{p,\delta}_{\mathbb{X}^n}} p_t \sum_{\overline{x}^n \in T^t_{\mathbb{X}^n}} \frac{1}{|T^t_{\mathbb{X}^n}|} \ketbra{\overline{x}^n}{\overline{x}^n}_{X^n} \otimes \ketbra{\overline{x}^n}{\overline{x}^n}_{R^n}\right)\right\|_1\\
    &= \left\|\left(\Delta \otimes \textbf{id}_{R^n}\right)\left(\sum_{t \in \mathcal{P}^{p,\delta}_{\mathbb{X}^n}} p_t \overline{\omega}'_{t, X^n R^n}\right)\right\|_1 && \hspace{-3.8cm}\text{where } \overline{\omega}'_{t, X^n R^n} \text{ contains a copy of the components of }\overline{\omega}_{t, X^n}\\
    &= \left\|\left(\Delta \otimes \textbf{id}_{R^n}\right)\left(\Tilde{\rho}_{X^nR^n}\right)\right\|_1, && \text{where } \Tilde{\rho}_{X^nR^n} := \sum_{t \in \mathcal{P}^{p,\delta}_{\mathbb{X}^n}} p_t \overline{\omega}'_{t, X^n R^n}
\end{align}
in which we used the fact that the sequences of the same type are orthogonal to each other and to those of other types, allowing us to bring the sums inside the trace norm. Moreover, since $\Tilde{\rho}_{X^nR^n}$ is a linear combination of symmetric basis states $\overline{\omega}'_{t, X^n R^n}$ (copied), we have that $\Tilde{\rho}_{X^n} \in \text{Sym}^n_{\text{cl}}(X)$, thus $\Tilde{\rho}_{X^nR^n}$ is a copy of a symmetric state.
\begin{flushright}
    $\blacksquare$
\end{flushright}
\subsection{Proof of Lemma \ref{lemma:TypedStateLowerBound}}
\label{appendix:proofOfTypeStateLowerBound}
We want to show that
\begin{equation}
\sigma_t^{\otimes n} \geq \frac{1}{\left|\mathcal{P}_{\mathbb{X}^n}\right|} \overline{w}_{t}.
\end{equation}

It is sufficient to show this inequality on an orthonormal set of vectors. We thus use our canonical basis $\ket{x^n}$. For basis vectors $\ket{x^n}$ with $x^n$ that is not of type $t$, the inequality is trivial since $\bra{x^n} \overline{w}_{t} \ket{x^n} = 0$. It is therefore sufficient to consider basis vectors for which $x^n$ is of type $t$. Thus, let $x^n$ be of type $t$, we have that
\begin{alignat}{2}
    &&\bra{x^n} \sigma_t^{\otimes n} \ket{x^n} &\geq \frac{1}{\left|\mathcal{P}_{\mathbb{X}^n}\right|} \bra{x^n} \overline{w}_{t} \ket{x^n}\\
    \iff && t^{\otimes n}(x^n) &\geq \frac{1}{\left|\mathcal{P}_{\mathbb{X}^n}\right|} \frac{1}{|T^t_{\mathbb{X}^n}|}\\
    \iff &&\prod_{k=1}^d \left(\frac{N(k|t)}{n}\right)^{N(k|t)} &\geq \frac{1}{\left|\mathcal{P}_{\mathbb{X}^n}\right|} \left(\frac{1}{n!\prod_{m=1}^d \frac{1}{N(m|t)!}}\right)\\
    \iff &&\left|\mathcal{P}_{\mathbb{X}^n}\right| \prod_{k=1}^d \left(\frac{N(k|t)}{n}\right)^{N(k|t)} n! \prod_{m=1}^d \frac{1}{N(m|t)!} &\geq 1\\
    \iff &&\left|\mathcal{P}_{\mathbb{X}^n}\right|t^{\otimes n}(T_t^{X^n}) &\geq 1.
\end{alignat}

The last inequality holds because the probability of having a sequence of type $t$ for a distribution $t$ is greater than the probability of having any other type $t'$ for the same distribution $t$. Therefore, multiplying the probability by the number of types gives a value greater than 1.
\begin{flushright}
    $\blacksquare$
\end{flushright}
\subsection{Proof of Lemma \ref{lemma:deFinettiReductionCPTNIMap}}
\label{appendix:ProofOfCPTNIMap}
By Lemma \ref{lemma:TypedStateLowerBound},
\begin{equation}
\sigma_t^{\otimes n} \geq \frac{1}{\left|\mathcal{P}_{\mathbb{X}^n}\right|} \overline{w}_{t} \geq \frac{1}{\left|\mathcal{P}_{\mathbb{X}^n}\right|} p(t) \overline{w}_{t}
\end{equation}
for some probability distribution $p(t)$. Since for both $\sigma_t^{\otimes n}$ and $\overline{w}_{t}$, every state of type $t$ has the same probability, we have $\Pi_t\sigma_t^{\otimes n} \Pi_t \propto \overline{w}_{t}$ and there exists a $r_t \geq 0$ such that
\begin{equation}
\Pi_t\sigma_t^{\otimes n} \Pi_t = \left(\frac{1}{\left|\mathcal{P}_{\mathbb{X}^n}\right|} + r_t\right) p(t) \overline{w}_{t}.
\end{equation}
Let us define
\begin{equation}
r_t' := \frac{1}{1 + \left|\mathcal{P}_{\mathbb{X}^n}\right| r_t},
\end{equation}
which admits $r_t' \in [0,1]$ and
\begin{align}
    r_t'\left(\frac{1}{\left|\mathcal{P}_{\mathbb{X}^n}\right|} + r_t\right) & = \frac{\left(\frac{1}{\left|\mathcal{P}_{\mathbb{X}^n}\right|} + r_t\right)}{1 + \left|\mathcal{P}_{\mathbb{X}^n}\right| r_t}\\
    &= \frac{\left(\frac{1}{\left|\mathcal{P}_{\mathbb{X}^n}\right|} + r_t\right)}{\left|\mathcal{P}_{\mathbb{X}^n}\right|\left(\frac{1}{\left|\mathcal{P}_{\mathbb{X}^n}\right|} + r_t\right)}\\
    &= \frac{1}{\left|\mathcal{P}_{\mathbb{X}^n}\right|}.
\end{align}

Projecting the state of the system $R^n$ of $\sigma^{\otimes n}_{t,XR}$ onto the subspace of a type $t$ also projects the states of the system $X^n$ onto that subspace, since $R$ contains a copy of the components of the state on $X$, that is
\begin{equation}
\left(I_{X^n} \otimes \Pi_{t, R^n}\right)\left(\sigma_{t,X^nR^n}\right)\left(I_{X^n} \otimes \Pi_{t, R^n}\right) = \left(\Pi_{t, X^n} \otimes \Pi_{t, R^n}\right)\left(\sigma_{t,X^nR^n}\right)\left(\Pi_{t, X^n} \otimes \Pi_{t, R^n}\right).
\end{equation}
We define the map $\mathcal{T}_1$ by Kraus operators
\begin{equation}
    \left\{\sqrt{r_t'}\Pi_{t, R^n} \otimes \ketbra{t}{t}_{J}\right\}_{t \in \mathcal{P}_{\mathbb{X}^n}^{p,\delta}}.
\end{equation}
We have that
\begin{align}
    \sum_{t \in \mathcal{P}_{\mathbb{X}^n}^{p,\delta}} K_t^\dagger K_t & = \sum_{t \in \mathcal{P}_{\mathbb{X}^n}^{p,\delta}} \left(\sqrt{r_t'}\Pi_{t, R^n} \otimes \ketbra{t}{t}_{J}\right)^\dagger \left(\sqrt{r_t'}\Pi_{t, R^n} \otimes \ketbra{t}{t}_{J}\right)\\
    & = \sum_{t \in \mathcal{P}_{\mathbb{X}^n}^{p,\delta}}\Pi_{t, R^n}^\dagger r_t' \Pi_{t, R^n} \otimes \ketbra{t}{t}_{J}\\
    & = \sum_{t \in \mathcal{P}_{\mathbb{X}^n}^{p,\delta}}\Pi_{t, R^n} \otimes r_t'\ketbra{t}{t}_{J}\\
    &\leq\sum_{t \in \mathcal{P}_{\mathbb{X}^n}^{p,\delta}}I_{R^n} \otimes r_t'\ketbra{t}{t}_{J}\\
    &= I_{R^n} \otimes \sum_{t \in \mathcal{P}_{\mathbb{X}^n}^{p,\delta}} r_t'\ketbra{t}{t}_{J}\\
    &\leq I_{R^n} \otimes \sum_{t \in \mathcal{P}_{\mathbb{X}^n}^{p,\delta}} \ketbra{t}{t}_{J}\\
    &= I_{R^n} \otimes I_{J}\\
    &= I_{R^n J}.
\end{align}
Therefore, the map $\mathcal{T}_1$ is CPTNI. We have that
\begin{align}
    &\left(\textbf{id}_{X^n} \otimes \mathcal{T}_1\right)\left(\tau^{p,\delta}_{X^nR^nJ}\right)\\
    &= \sum_{t \in \mathcal{P}_{\mathbb{X}^n}^{p,\delta}} \left(\sqrt{r_t'}\Pi_{t, R^n} \otimes \ketbra{t}{t}_{J}\right) \left(\frac{1}{\left|\mathcal{P}_{\mathbb{X}^n}^{p,\delta}\right|} \sum_{t' \in \mathcal{P}_{\mathbb{X}^n}^{p,\delta}} \sigma_{t', X R}^{\otimes n} \otimes \ketbra{t'}{t'}_{J}\right) \left(\sqrt{r_t'}\Pi_{t, R^n} \otimes \ketbra{t}{t}_{J}\right)\\
    &= \frac{1}{\left|\mathcal{P}_{\mathbb{X}^n}^{p,\delta}\right|} \sum_{t \in \mathcal{P}_{\mathbb{X}^n}^{p,\delta}} r_t' \sum_{t' \in \mathcal{P}_{\mathbb{X}^n}^{p,\delta}} \Pi_{t, R^n} \sigma_{t', X R}^{\otimes n} \Pi_{t, R^n} \otimes \ketbra{t}{t}\!\ketbra{t'}{t'}\!\ketbra{t}{t}_{J}\\
    &= \frac{1}{\left|\mathcal{P}_{\mathbb{X}^n}^{p,\delta}\right|} \sum_{t \in \mathcal{P}_{\mathbb{X}^n}^{p,\delta}} r_t' \Pi_{t, R^n} \sigma_{t, X R}^{\otimes n} \Pi_{t, R^n} \otimes \ketbra{t}{t}_{J}\\
    &= \frac{1}{\left|\mathcal{P}_{\mathbb{X}^n}^{p,\delta}\right|} \sum_{t \in \mathcal{P}_{\mathbb{X}^n}^{p,\delta}} r_t' \left(\left(\frac{1}{\left|\mathcal{P}_{\mathbb{X}^n}\right|} + r_t\right) p(t) \overline{w}_{t, X^n R^n}\right) \otimes \ketbra{t}{t}_{J}\\
    &= \frac{1}{\left|\mathcal{P}_{\mathbb{X}^n}^{p,\delta}\right|} \sum_{t \in \mathcal{P}_{\mathbb{X}^n}^{p,\delta}} \frac{1}{\left|\mathcal{P}_{\mathbb{X}^n}\right|} p(t) \overline{w}_{t, X^n R^n}\otimes \ketbra{t}{t}_{J}\\
    &= \frac{1}{\left|\mathcal{P}_{\mathbb{X}^n}^{p,\delta}\right|\left|\mathcal{P}_{\mathbb{X}^n}\right|} \sum_{t \in \mathcal{P}_{\mathbb{X}^n}^{p,\delta}} p(t) \overline{w}_{t, X^n R^n}\otimes \ketbra{t}{t}_{J}.
\end{align}
Thus, we have
\begin{equation}\tilde{\rho}_{X^nR^nJ} = \left|\mathcal{P}_{\mathbb{X}^n}^{p,\delta}\right|\left|\mathcal{P}_{\mathbb{X}^n}\right|\left(\textbf{id}_{X^n} \otimes \mathcal{T}_1\right)\left(\tau^{p,\delta}_{X^nR^nJ}\right).
\end{equation}
We finally trace out the system $J$ with a CPTP map $\mathcal{T}_2$ to obtain $\mathcal{T} = \mathcal{T}_2 \circ \mathcal{T}_1$, which is CPTNI. Thus, we get
\begin{equation}
\tilde{\rho}_{X^nR^n} = \left|\mathcal{P}_{\mathbb{X}^n}^{p,\delta}\right|\left|\mathcal{P}_{\mathbb{X}^n}\right| \left(\textbf{id}_{X^n} \otimes \mathcal{T}\right)\left(\tau^{p,\delta}_{X^nR^nJ}\right).
\end{equation}
\begin{flushright}
    $\blacksquare$
\end{flushright}

\subsection{Proof of Theorem \ref{thm:restrictedClassicalQuantumDeFinettiReduction}}
\label{appendix:ProofClassicalQuantumdeFinettiReduction}
Let $\rho_{X^nR^n}$ be any state in the maximization of the $\delta$-typical classical-quantum channel norm. By Lemma \ref{lemma:symmetricRestrictedMaximization}, there exists a classical state $\Tilde{\rho}_{X^n} \in \mathcal{T}^{p,\delta}_{X^n} \cap ~\text{Sym}^n_{\text{cl}}(X)$ with a copy $\Tilde{\rho}_{X^nR^n}$ on $R^n$ such that
\begin{equation}
\left\|\left(\Delta \otimes \textbf{id}_{R^n}\right)\left(\rho_{X^n R^n}\right)\right\|_1 = \left\|\left(\Delta \otimes \textbf{id}_{R^n}\right)\left(\Tilde{\rho}_{X^n R^n}\right)\right\|_1.
\end{equation}
Furthermore, by Lemma \ref{lemma:deFinettiReductionCPTNIMap}, there exists a CPTNI map $\mathcal{T}$ such that
\begin{equation}
\Tilde{\rho}_{X^n R^n} = \left|\mathcal{P}^{p,\delta}_{\mathbb{X}^n}\right|\left|\mathcal{P}_{\mathbb{X}^n}\right| \left(\textbf{id}_{X^n} \otimes \mathcal{T}\right)\left(\tau^{p,\delta}_{X^nR^nJ}\right).
\end{equation}
We thus have
\begin{align}
    \left\|\left(\Delta \otimes \textbf{id}_{R^n}\right)\left(\Tilde{\rho}_{X^n R^n}\right)\right\|_1 & = \left\|\left(\Delta \otimes \textbf{id}_{R^nJ}\right)\left(\left|\mathcal{P}^{p,\delta}_{\mathbb{X}^n}\right|\left|\mathcal{P}_{\mathbb{X}^n}\right| \left(\textbf{id}_{X^n} \otimes \mathcal{T}\right)\left(\tau^{p,\delta}_{X^nR^nJ}\right)\right)\right\|_1\\
    & = \left|\mathcal{P}^{p,\delta}_{\mathbb{X}^n}\right|\left|\mathcal{P}_{\mathbb{X}^n}\right| \left\|\left(\Delta \otimes \mathcal{T}\right)\left(\tau^{p,\delta}_{X^nR^nJ}\right)\right\|_1\\
    &\leq \left|\mathcal{P}^{p,\delta}_{\mathbb{X}^n}\right|\left|\mathcal{P}_{\mathbb{X}^n}\right| \left\|\left(\Delta \otimes \textbf{id}_{R^nJ}\right)\left(\tau^{p,\delta}_{X^nR^nJ}\right)\right\|_1.
\end{align}
By choosing the reference system $R' := R^nJ$, we thus have that
\begin{equation}
\left\|\left(\Delta \otimes \textbf{id}_{R}\right)\left(\rho_{X^n R^n }\right)\right\|_1 \leq \left|\mathcal{P}^{p,\delta}_{\mathbb{X}^n}\right|\left|\mathcal{P}_{\mathbb{X}^n}\right| \left\|\left(\Delta \otimes \textbf{id}_{R'}\right)\left(\tau^{p,\delta}_{X^nR' }\right)\right\|_1.
\end{equation}
Since this statement is true for all $\rho_{X^n R^n}$ in the maximization of the $\delta$-typical classical-quantum channel norm, we can conclude that
\begin{equation}
\left\|\Delta\right\|^{p,\delta}_{\diamond, cl} \leq \left|\mathcal{P}^{p,\delta}_{\mathbb{X}^n}\right|\left|\mathcal{P}_{\mathbb{X}^n}\right| \left\|\left(\Delta \otimes \textbf{id}_{R'}\right)\left(\tau^{p,\delta}_{X^nR' }\right)\right\|_1.
\end{equation}
\begin{flushright}
    $\blacksquare$
\end{flushright}

\section{Classical-Quantum Channel Simulation Proofs}
\subsection{Proof of Lemma \ref{lemma:deltaTypicalTraceDistanceBound}}
\label{appendix:typicalProjectionsBounds}
In this section we give upper bounds on the error made when approximating a de Finetti state by its $\delta$-typical projected version. Let us start by proving the following lemma.
\begin{lemma}
\label{prop:operatorTraceInequality}
    Let $\left\{\Pi_i\right\}_{i=1}^n$ be a collection of linear operators $\Pi_i \in \mathcal{L}(\mathcal{H})$ such that $0 \leq \Pi_i \leq I$. Let $\rho_{\mathcal{H}_1 \ldots \mathcal{H}_n}$ be a state on $\mathcal{H} = \bigotimes_{i=1}^n \mathcal{H}_i$, then the following inequality holds
    \begin{equation}
    \text{Tr}\left(\left(\bigotimes_{i=1}^n \Pi_i\right)\rho_{\mathcal{H}_1 \ldots \mathcal{H}_n}\right) \geq 1 - n + \sum_{i=1}^n \text{Tr}\left(\Pi_i\rho_{\mathcal{H}_i}\right)
    \end{equation}
\end{lemma}
\begin{proof}
    We prove this by induction on $n$. For $n=2$, let $\Pi_i^{\perp}$ be the orthogonal projection to $\Pi_i$. We have that
    \begin{align}
        &\text{Tr}\left(\left(\Pi_1 \otimes \Pi_2\right)\rho_{\mathcal{H}_1\mathcal{H}_2}\right)\\
        &= \text{Tr}\left(\left(\left(\mathbb{I} - \Pi_1^{\perp}\right) \otimes \left(\mathbb{I} - \Pi_2^{\perp}\right)\right)\rho_{\mathcal{H}_1\mathcal{H}_2}\right)\\
        &= \text{Tr}\left(\rho_{\mathcal{H}_1\mathcal{H}_2}\right) - \text{Tr}\left(\left(\Pi_1^{\perp} \otimes \mathbb{I}\right)\rho_{\mathcal{H}_1\mathcal{H}_2}\right) - \text{Tr}\left(\left(\mathbb{I} \otimes \Pi_2^{\perp}\right)\rho_{\mathcal{H}_1\mathcal{H}_2}\right) + \text{Tr}\left(\left(\Pi_1^{\perp} \otimes \Pi_2^{\perp}\right)\rho_{\mathcal{H}_1\mathcal{H}_2}\right)\\
        &\geq 1 - \text{Tr}\left(\left(\Pi_1^{\perp} \otimes \mathbb{I}\right)\rho_{\mathcal{H}_1\mathcal{H}_2}\right) - \text{Tr}\left(\left(\mathbb{I} \otimes \Pi_2^{\perp}\right)\rho_{\mathcal{H}_1\mathcal{H}_2}\right)\\
        &= 1 - \text{Tr}\left(\Pi_1^{\perp}\rho_{\mathcal{H}_1}\right) - \text{Tr}\left(\Pi_2^{\perp}\rho_{\mathcal{H}_2}\right)\\
        &= 1 - \left(1 - \text{Tr}\left(\Pi_1\rho_{\mathcal{H}_1}\right)\right) - \left(1 - \text{Tr}\left(\Pi_2\rho_{\mathcal{H}_2}\right)\right)\\
        &= \text{Tr}\left(\Pi_1\rho_{\mathcal{H}_1}\right) + \text{Tr}\left(\Pi_2\rho_{\mathcal{H}_2}\right) - 1\\
        &= 1 - 2 + \text{Tr}\left(\Pi_1\rho_{\mathcal{H}_1}\right) + \text{Tr}\left(\Pi_2\rho_{\mathcal{H}_2}\right)\\
        &= 1 - n + \sum_{i=1}^n \text{Tr}\left(\Pi_i\rho_{\mathcal{H}_i}\right).
    \end{align}
    We assume that the property is true for $n=k$ for some $k \in \mathbb{N}_{\geq 2}$. We show that it is true for $n+1$. We have that
    \begin{align}
        &\text{Tr}\left(\left(\bigotimes_{i=1}^{k+1} \Pi_i\right)\rho_{\mathcal{H}_1 \ldots \mathcal{H}_{k+1}}\right)\\
        &\geq 1 - 2 + \text{Tr}\left(\bigotimes_{i=1}^{k} \Pi_i\rho_{\mathcal{H}_i}\right) + \text{Tr}\left(\Pi_{k+1}\rho_{\mathcal{H}_{k+1}}\right) && \text{by the case $n=2$}\\
        &\geq 1 - 2 + \left(1 - n + \sum_{i=1}^{k} \text{Tr}\left(\Pi_i\rho_{\mathcal{H}_i}\right)\right) + \text{Tr}\left(\Pi_{k+1}\rho_{\mathcal{H}_{k+1}}\right) && \text{by induction hypothesis}\\
        &= 1 - (k + 1) + \sum_{i=1}^{k+1} \text{Tr}\left(\Pi_i\rho_{\mathcal{H}_i}\right).
    \end{align}
\end{proof}

With this lemma, we can now bound the approximation made in trace distance. First, we lower bound the probability of measuring the $\delta$-typical projector on the de Finetti state to then use the gentle measurement lemma to upper bound the trace distance. We have that
\begin{align}
    &\text{Tr}\left(\left(\sum_{t \in \mathcal{P}_{\mathbb{X}^n}^{\tilde{t},\delta}} \left(\bigotimes_{i=1}^k \Pi_{A_i^n}^{t,\delta}\right) \otimes \ketbra{t}{t}_{R_T}\right) \tau_{R_T A_1^n \ldots A_k^n}\right)\\
    &= \text{Tr}\left(\left(\sum_{t \in \mathcal{P}_{\mathbb{X}^n}^{\tilde{t},\delta}} \left(\bigotimes_{i=1}^k \Pi_{A_i^n}^{t,\delta}\right) \otimes \ketbra{t}{t}_{R_T}\right)\left(\frac{1}{\left|\mathcal{P}_{\mathbb{X}^n}^{\tilde{t},\delta}\right|} \sum_{t',t'' \in \mathcal{P}_{\mathbb{X}^n}^{\tilde{t},\delta}} \ketbra{\sigma_{t'}}{\sigma_{t''}}^{\otimes n}_{} \otimes \ketbra{t'}{t''}_{R_T}\right)\right)\\
    &= \frac{1}{\left|\mathcal{P}_{\mathbb{X}^n}^{\tilde{t},\delta}\right|}\sum_{t,t',t'' \in \mathcal{P}_{\mathbb{X}^n}^{\tilde{t},\delta}} \text{Tr}\left(\left(\bigotimes_{i=1}^k \Pi_{A_i^n}^{t,\delta}\right) \otimes \ketbra{t}{t}_{R_T}\left(\ketbra{\sigma_{t'}}{\sigma_{t''}}^{\otimes n}_{} \otimes \ketbra{t'}{t''}_{R_T}\right)\right)\\
    &= \frac{1}{\left|\mathcal{P}_{\mathbb{X}^n}^{\tilde{t},\delta}\right|}\sum_{t,t',t'' \in \mathcal{P}_{\mathbb{X}^n}^{\tilde{t},\delta}} \text{Tr}\left(\left(\bigotimes_{i=1}^k \Pi_{A_i^n}^{t,\delta}\right)\ketbra{\sigma_{t'}}{\sigma_{t''}}^{\otimes n}_{}\right) \text{Tr}\left(\ketbra{t}{t}_{R_T}\ketbra{t'}{t''}_{R_T}\right)\\
    &= \frac{1}{\left|\mathcal{P}_{\mathbb{X}^n}^{\tilde{t},\delta}\right|}\sum_{t,t',t'' \in \mathcal{P}_{\mathbb{X}^n}^{\tilde{t},\delta}} \text{Tr}\left(\left(\bigotimes_{i=1}^k \Pi_{A_i^n}^{t,\delta}\right)\ketbra{\sigma_{t'}}{\sigma_{t''}}^{\otimes n}_{}\right) \delta_{t,t'}\delta_{t'',t}\\
    &= \frac{1}{\left|\mathcal{P}_{\mathbb{X}^n}^{\tilde{t},\delta}\right|}\sum_{t \in \mathcal{P}_{\mathbb{X}^n}^{\tilde{t},\delta}} \text{Tr}\left(\left(\bigotimes_{i=1}^k \Pi_{A_i^n}^{t,\delta}\right)\ketbra{\sigma_{t}}{\sigma_{t}}^{\otimes n}_{A_1 \ldots A_k}\right)\\
    &\geq \frac{1}{\left|\mathcal{P}_{\mathbb{X}^n}^{\tilde{t},\delta}\right|}\sum_{t \in \mathcal{P}_{\mathbb{X}^n}^{\tilde{t},\delta}} \left(1 - k + \sum_{i=1}^k \text{Tr}\left(\Pi_{A_i^n}^{t,\delta} \sigma_{t,A_i}^{\otimes n}\right)\right)\\
    &= \frac{1}{\left|\mathcal{P}_{\mathbb{X}^n}^{\tilde{t},\delta}\right|}\sum_{t \in \mathcal{P}_{\mathbb{X}^n}^{\tilde{t},\delta}} \left(1 - \sum_{i=1}^k \text{Tr}\left(\Pi_{A_i^n}^{t,\delta,\perp} \sigma_{t,A_i}^{\otimes n}\right)\right)\\
    &\geq \frac{1}{\left|\mathcal{P}_{\mathbb{X}^n}^{\tilde{t},\delta}\right|}\sum_{t \in \mathcal{P}_{\mathbb{X}^n}^{\tilde{t},\delta}} \left(1 - \sum_{i=1}^k |\mathcal{P}_{A_i^n}|2^{-n\frac{\delta^2}{2\ln(2)}}\right)\\
    &= 1 - 2^{-n\frac{\delta^2}{2\ln(2)}}\sum_{i=1}^k |\mathcal{P}_{A_i^n}|.
\end{align}
Now, by the gentle measurement lemma, we have that
\begin{equation}
    \left\|\ket{\Tilde{\tau}^{\delta,\tilde{t}}}_{R_TA_1^n \ldots A_k^n} - \ket{\tau^{\delta,\tilde{t}}}_{R_TA_1^n \ldots A_k^n}\right\|_1 \leq 2^{-n\frac{\delta^2}{4\ln(2)} + 1}\sqrt{\sum_{i=1}^k |\mathcal{P}_{A_i^n}|}.
\end{equation}
This concludes the proof of Lemma \ref{lemma:deltaTypicalTraceDistanceBound}.

\subsection{Proof of Lemma \ref{lemma:purityBoundDeltaTypicalDeFinetti}}
\label{appendix:purityBoundTypicalDeFinetti}
In this section, we prove an upper bound on the purity of the $\delta$-typical de Finetti state that is projected on the $\delta$-typical subspace. We first give a bound on the purity of a state $\sigma$ projected on its $\delta$-typical subspace.
\begin{proposition}
    \label{prop:purityBoundProjectedState}
    Let $\sigma_A$ be a density operator with eigenvalue distribution $p$ and $\Pi_{A^n}^{\sigma, \delta}$ be the $\delta$-typical projector of $\sigma_A$ on system $A^n$. Define the projected state $\tilde{\sigma}_{A^n} := \Pi_{A^n}^{\sigma, \delta} \sigma_A^{\otimes n} \Pi_{A^n}^{\sigma, \delta}$. Then, we have the following upper bound on the purity of the state $\tilde{\sigma}_{A^n}$:
    \begin{equation}
        \text{Tr}\left(\tilde{\sigma}_{A^n}^2\right) \leq \exp\left(-n\left(H(A)_\sigma - \tilde{\omega}_{1,1}(\delta,|A|)\right)\right),
    \end{equation}
    in which $H(A)_\sigma$ is the von Neumann entropy of the state $\sigma_A$ and $\tilde{\omega}_{1,1}(\delta,|A|)$ is defined in \eqref{eq:tilde-omega}.
\end{proposition}
\begin{proof}
    First, note that after projection, only the $\delta$-typical eigenvalues of $\sigma_A^{\otimes n}$ remain in the state $\tilde{\sigma}_{A^n}$, which are then squared. Therefore, the trace on $\tilde{\sigma}_{A^n}^2$ is the sum of the squared $\delta$-typical eigenvalues of $\sigma_A^{\otimes n}$, i.e., the sum of the squared probabilities of measuring a $\delta$-typical sequence according to the distribution $p^{\otimes n}$, that is
    \begin{align}
        \text{Tr}\left(\tilde{\sigma}_{A^n}^2\right) &= \sum_{a^n \in T_{A^n}^{p,\delta}} p^{\otimes n}(a^n)^2\\
        &= \sum_{t \in \mathcal{P}_{A^n}^{p,\delta}} \sum_{a^n \in T_{A^n}^{t}} p^{\otimes n}(a^n)^2.
    \end{align}
    Using the upper bound \eqref{eq:type-probability-upper-bound} on the probability of a sequence of type $t$ being generated by the distribution $p$, and the upper bound on the size of the type class $T_{A^n}^t$ \eqref{eq:type-class-size}, we get
    \begin{align}
        \sum_{a^n \in T_{A^n}^{t}} p^{\otimes n}(a^n)^2 &\leq \sum_{a^n \in T_{A^n}^{t}} \left(\frac{1}{2^{n(H(A)_t)}}\right)^2\\
        &\leq |T_{A^n}^{t}| \frac{1}{2^{2n(H(A)_t)}}\\
        &\leq 2^{nH(A)_t} \frac{1}{2^{2n(H(A)_t)}}\\
        &= 2^{-nH(A)_t}.
    \end{align}
    Therefore, we have
    \begin{align}
        \sum_{t \in \mathcal{P}_{A^n}^{p,\delta}} \sum_{a^n \in T_{A^n}^{t}} p^{\otimes n}(a^n)^2 &\leq \sum_{t \in \mathcal{P}_{A^n}^{p,\delta}} 2^{-nH(A)_t}\\
        &\leq \sum_{t \in \mathcal{P}_{A^n}^{p,\delta}} \exp{-n\min_{t \in \mathcal{P}_{A^n}^{p,\delta}} H(A)_t}\\
        &\leq |\mathcal{P}_{A^n}| \exp{-n(H(A)_p - \omega(\delta,|A|))}\\
        &\leq \exp{-n(H(A)_p - \tilde{\omega}_{1,1}(\delta,|A|))},
    \end{align}
    in which $\omega(\delta,|A|)$ and $\tilde{\omega}_{1,1}(\delta,|A|)$ are defined in \eqref{eq:omega-delta-d} and \eqref{eq:tilde-omega} respectively. Since the shannon entropy $H(A)_p$ of the distribution $p$ is equal to the von Neumann entropy $H(A)_\sigma$ of the state $\sigma_A$, this concludes the proof.
\end{proof}
Now, we prove a short lemma that we use in our computation of the purity.
\begin{lemma}
\label{lemma:purityInequality}
    Let $\sigma_t$ and $\sigma_{t'}$ be two density matrices, then
    \begin{equation}
        \label{eq:purity_inequality}
        \sigma_t\sigma_{t'} + \sigma_{t'}\sigma_t \leq \sigma_t^2 + \sigma_{t'}^2.
    \end{equation}
\end{lemma}
\begin{proof}
    \begin{align}
        \sigma_t^2 + \sigma_{t'}^2 - \left(\sigma_t\sigma_{t'} + \sigma_{t'}\sigma_t\right)
        &= \sigma_t\left(\sigma_t - \sigma_{t'}\right) + \sigma_{t'}\left(\sigma_{t'} - \sigma_t\right)\\
        &= \sigma_t\left(\sigma_t - \sigma_{t'}\right) - \sigma_t\left(\sigma_t - \sigma_{t'}\right)\\
        &= \left(\sigma_t - \sigma_{t'}\right)^2\\
        &\geq 0
    \end{align}
\end{proof}
We can now prove Lemma~\ref{lemma:purityBoundDeltaTypicalDeFinetti}. To make the notation compact, we denote $\tilde{\tau} := \tilde{\tau}^{\tilde{t},\delta}$ and $\tau := \tau^{\tilde{t},\delta}$.
We analyze the purity of the non-projected state on which we only project on the $\delta$-typical subspace of $A_{i'}$ and don't renormalize (we denote this state $\overline{\tau}_{\leq}$), and then use this to upper bound the purity of the reduced state~$\tilde{\tau}_{A_{i'}}$. We first argue a direct inequality from the definition of these two states, but when $\tilde{\tau}$ is not renormalized by the probability of measuring a $\delta$-typical state. We denote the non-renormalized version of $\tilde{\tau}$ by $\tilde{\tau}_{\leq}$. First, note that
\begin{align}
    \sum_{t \in \mathcal{P}_{\mathbb{X}^n}^{\tilde{t},\delta}}\!\left(\bigotimes_{i=1}^k \Pi_{A_i^n}^{t,\delta}\right)\!\otimes\ketbra{t}{t}_{R_T} &\leq\sum_{t \in \mathcal{P}_{\mathbb{X}^n}^{\tilde{t},\delta}}\!\left(\bigotimes_{i\neq i'}^k I_{A_i^n}\right)\!\otimes \Pi^{t,\delta}_{A_{i'}^n}\!\otimes \ketbra{t}{t}_{R_T}\\
    &=\!\left(\bigotimes_{i\neq i'}^k I_{A_i^n}\right) \otimes \left(\sum_{t \in \mathcal{P}_{\mathbb{X}^n}^{\tilde{t},\delta}}\!\Pi^{t,\delta}_{A_{i'}^n} \otimes \ketbra{t}{t}_{R_T}\right),
\end{align}
since the projectors $\Pi_{A_i^n}^{t,\delta}$ are smaller than the identity. Thus, we have that $\tilde{\tau}_{\leq, A_1^n \ldots A_k^n} \leq \overline{\tau}_{\leq, A_1^n \ldots A_k^n}$ (for non-normalized states after projection). This implies that
\begin{equation}
    \text{Tr}\left(\tilde{\tau}_{\leq,A_{i'}^n}^2\right) \leq \text{Tr}\left(\overline{\tau}_{\leq,A_{i'}^n}^2\right).
\end{equation}
But in particular, when $\tilde{\tau}$ is renormalized (note that this renormalization is $1 - \frac{\left(\varepsilon_{\delta}\right)^2}{2}$ as shown in Lemma~\ref{lemma:deltaTypicalTraceDistanceBound}, with $\varepsilon_{\delta}$ from that same lemma), we have that
\begin{equation}
    \label{eq:purityInequalityRenormalized}
    \text{Tr}\left(\tilde{\tau}_{A_{i'}^n}^2\right) \leq \frac{1 }{\left(1 - \frac{\left(\varepsilon_{\delta}\right)^2}{2}\right)^2} \text{Tr}\left(\overline{\tau}_{\leq,A_{i'}^n}^2\right),
\end{equation}
Now, we upper bound $\text{Tr}\left(\overline{\tau}_{\leq,A_{i'}^n}^2\right)$ to get an upper bound on $\text{Tr}\left(\tilde{\tau}_{A_{i'}^n}^2\right)$. Fix any order on the types, we denote~$t'~\geq~t$ if $t'$ comes after $t$ in this order. Let $\tilde{\sigma}_{t,A_{i'}^n} := \Pi_{A_{i'}}^{t,\delta} \sigma_{t,A_{i'}}^{\otimes n} \Pi_{A_{i'}}^{t,\delta}$, then we have that
\begin{align}
    \overline{\tau}^2_{\leq,A_{i'}} &= \frac{1}{\left|\mathcal{P}_{\mathbb{X}^n}^{\tilde{t},\delta}\right|^2} \sum_{t,t' \in \mathcal{P}_{\mathbb{X}^n}^{\tilde{t},\delta}} \tilde{\sigma}_{t,A_{i'}} \tilde{\sigma}_{t',A_{i'}}\\
    &\leq \frac{1}{\left|\mathcal{P}_{\mathbb{X}^n}^{\tilde{t},\delta}\right|^2} \sum_{t \in \mathcal{P}_{\mathbb{X}^n}^{\tilde{t},\delta}} \sum_{t' \geq t} \left(\tilde{\sigma}_{t,A_{i'}} \tilde{\sigma}_{t',A_{i'}} + \tilde{\sigma}_{t',A_{i'}} \tilde{\sigma}_{t,A_{i'}}\right)\\
    &\leq \frac{1}{\left|\mathcal{P}_{\mathbb{X}^n}^{\tilde{t},\delta}\right|^2}\sum_{t \in \mathcal{P}_{\mathbb{X}^n}^{\tilde{t},\delta}} \sum_{t' \geq t} \left(\tilde{\sigma}_{t,A_{i'}}^2 + \tilde{\sigma}_{t',A_{i'}}^2\right) && \text{by Lemma \ref{lemma:purityInequality}}\\
    &\leq \frac{1}{\left|\mathcal{P}_{\mathbb{X}^n}^{\tilde{t},\delta}\right|^2} \sum_{t \in \mathcal{P}_{\mathbb{X}^n}^{\tilde{t},\delta}} \sum_{t' \in \mathcal{P}_{\mathbb{X}^n}^{\tilde{t},\delta}} \left(\tilde{\sigma}_{t,A_{i'}}^2 + \tilde{\sigma}_{t',A_{i'}}^2\right)\\
    &= \frac{2}{\left|\mathcal{P}_{\mathbb{X}^n}^{\tilde{t},\delta}\right|} \sum_{t \in \mathcal{P}_{\mathbb{X}^n}^{\tilde{t},\delta}} \tilde{\sigma}_{t,A_{i'}}^2.
\end{align}
By the bound on the purity of the $\delta$-typical approximation (Proposition~\ref{prop:purityBoundProjectedState}), we have that
\begin{align}
    \text{Tr}\left(\overline{\tau}^2_{\leq,A_{i'}}\right) &\leq \frac{2}{\left|\mathcal{P}_{\mathbb{X}^n}^{\tilde{t},\delta}\right|} \sum_{t \in \mathcal{P}_{\mathbb{X}^n}^{\tilde{t},\delta}} \text{Tr}\left(\left(\Pi_{A_{i'}}^{t,\delta} \sigma_{t,A_{i'}}^{\otimes n} \Pi_{A_{i'}}^{t,\delta}\right)^2\right)\\
    &\leq \frac{2}{\left|\mathcal{P}_{\mathbb{X}^n}^{\tilde{t},\delta}\right|} \sum_{t \in \mathcal{P}_{\mathbb{X}^n}^{\tilde{t},\delta}} \exp{-n\left(H(A_{i'})_t - \tilde{\omega}_{1,1}(\delta,|A_{i'}|)\right)}\\
    &\leq \frac{2}{\left|\mathcal{P}_{\mathbb{X}^n}^{\tilde{t},\delta}\right|} \left|\mathcal{P}_{\mathbb{X}^n}^{\tilde{t},\delta}\right| \exp{-n\left(\min_{t \in \mathcal{P}_{\mathbb{X}^n}^{\tilde{t},\delta}}\left(H(A_{i'})_t - \tilde{\omega}_{1,1}(\delta,|A_{i'}|)\right)\right)}\\
    &= 2\exp{-n\left(\min_{t \in \mathcal{P}_{\mathbb{X}^n}^{\tilde{t},\delta}} H(A_{i'})_t - \tilde{\omega}_{1,1}(\delta,|A_{i'}|)\right)},
\end{align}
in which $H(A_{i'})_t$ is the von Neumann entropy of the reduced state $\sigma_{t,A_{i'}}$. Combining this with the inequality in~\eqref{eq:purityInequalityRenormalized}, we have that
\begin{equation}
    \text{Tr}\left(\tilde{\tau}_{A_{i'}}^2\right) \leq \frac{2}{\left(1 - \frac{\left(\varepsilon_{\delta}\right)^2}{2}\right)^2}\exp{-n\left(\min_{t \in \mathcal{P}_{\mathbb{X}^n}^{\tilde{t},\delta}} H(A_{i'})_t - \tilde{\omega}_{1,1}(\delta,|A_{i'}|)\right)}.
\end{equation}
If we define $f(\delta)$ (with implicit dependence on $n$), as
\begin{equation}
    f(\delta) := -2\log(1 - \frac{\left(\varepsilon_{\delta}\right)^2}{2}) + 1,
\end{equation}
We then get
\begin{equation}
    \text{Tr}\left(\tilde{\tau}_{A_{i'}}^2\right) \leq \exp{-n\left(\min_{t \in \mathcal{P}_{\mathbb{X}^n}^{\tilde{t},\delta}} H(A_{i'})_t - \tilde{\omega}_{1,1}(\delta,|A_{i'}|) - \frac{f(\delta)}{n}\right)}.
\end{equation}
which concludes the proof of Lemma~\ref{lemma:purityBoundDeltaTypicalDeFinetti}.
\begin{flushright}
    $\square$
\end{flushright}

\bibliographystyle{IEEEtran}
\bibliography{bibliofile}
\end{document}